\RequirePackage[l2tabu,orthodox]{nag}
\documentclass[11pt,a4paper]{article}

\usepackage[margin=1truein]{geometry}

\usepackage{lmodern}
\usepackage[T1]{fontenc}
\usepackage{textcomp}
\usepackage[utf8]{inputenc}

\usepackage{amsmath}
\usepackage[renew-matrix]{nicematrix}

\usepackage[bookmarksnumbered,hypertexnames=false,pdfdisplaydoctitle,pdfusetitle,unicode]{hyperref}
\usepackage{url}

\usepackage[capitalize,noabbrev]{cleveref}

\expandafter\def\csname ver@etex.sty\endcsname{3000/12/31}

\usepackage{autonum}

\usepackage{booktabs}
\usepackage{tablefootnote}

\usepackage{tikz}
\usetikzlibrary{arrows.meta,calc,decorations.pathmorphing,patterns,positioning,shapes}

\usepackage{comment}

\usepackage{okicmd}
\usepackage{okithm}

\newcommand{\detsum}{\mathbin{\nabla}}

\DeclareMathOperator{\Const}{Const}
\DeclareMathOperator{\E}{E}
\newcommand{\vdet}{\zeta}
\newcommand{\tvdet}{\hat{\zeta}}

\newcommand{\D}{\mathrm{D}}
\newcommand{\fsbr}[1]{\text{\textup{$\mathrm{<}\mkern-2.3mu\llap(\,#1\mathrm{>}\mkern-4.6mu\llap)$}}\,}

\title{\texorpdfstring{%
  Computing Valuations of the Dieudonné Determinants\thanks{%
    A preliminary version of the part of this paper about Edmonds' problem has been appeared at the 47th International Colloquium on Automata, Languages and Programming (ICALP ’20), July 2020, under the title of ``On solving (non)commutative weighted Edmonds’ problem''.
    The previous version of this paper was titled ``Computing the maximum degree of minors in skew polynomial matrices''.
  }%
}{%
  Computing Valuations of the Dieudonné Determinants
}}
\author{Taihei Oki%
  \texorpdfstring{\thanks{
    Department of Mathematical Informatics, Graduate School of Information Science and Technology, University of Tokyo, Tokyo 113-8656, Japan.
    E-mail: \href{mailto:taihei_oki@mist.i.u-tokyo.ac.jp}{\nolinkurl{taihei_oki@mist.i.u-tokyo.ac.jp}}
  }}{}
}
\newcommand{\mykeywords}{valuation skew fields, Dieudonné determinants, skew polynomials, differential equations, difference equations, combinatorial relaxation, matrix expansion}
\hypersetup{pdfkeywords={\mykeywords}}

\begin{document}
\maketitle

\begin{abstract}
  % !TEX encoding = UTF-8 Unicode
% !TEX root = main.tex
% !TEX spellcheck = en-US
% !TEX TS-program = latexmk

This paper addresses the problem of computing valuations of the Dieudonné determinants of matrices over discrete valuation skew fields (DVSFs).
Under a reasonable computational model, we propose two algorithms for a class of DVSFs, called split.
Our algorithms are extensions of the combinatorial relaxation of Murota~(1995) and the matrix expansion by Moriyama--Murota~(2013), both of which are based on combinatorial optimization.
While our algorithms require an upper bound on the output, we give an estimation of the bound for skew polynomial matrices and show that the estimation is valid only for skew polynomial matrices.

We consider two applications of this problem.
The first one is the noncommutative weighted Edmonds' problem (nc-WEP), which is to compute the degree of the Dieudonné determinants of matrices having noncommutative symbols.
We show that the presented algorithms reduce the nc-WEP to the unweighted problem in polynomial time.
In particular, we show that the nc-WEP over the rational field is solvable in time polynomial in the input bit-length.
We also present an application to analyses of degrees of freedom of linear time-varying systems by establishing formulas on the solution spaces of linear differential/difference equations.

\end{abstract}

\newpage
\tableofcontents

\newpage
% !TEX encoding = UTF-8 Unicode
% !TEX root = main.tex
% !TEX spellcheck = en-US
% !TEX TS-program = latexmk

\section{Introduction}\label{sec:introduction}
A (\emph{real}) \emph{valuation} on a field $F$ is a map $\funcdoms{v}{F}{\setR \cup \set{+\infty}}$ such that
\begin{enumerate}[label={\upshape{(V\arabic*)}}]
  \item $v(ab) = v(a) + v(b)$ for $a, b \in F$,\label{item:V1}
  \item $v(a+b) \ge \min\set{v(a), v(b)}$ for $a, b \in F$,\label{item:V2}
  \item $v(1) = 0$,\label{item:V3}
  \item $v(0) = +\infty$.\label{item:V4}
\end{enumerate}
A valuation is called \emph{discrete} if $v(F) = \setZ \cup \set{+\infty}$.
For example, the minus of the degree is a discrete valuation on the rational function field $K(s)$ over a field $K$, where $\deg p/q \defeq \deg p - \deg q$ for $p, q \in K[s]$.
The $p$-adic valuation on rationals $\setQ$ is another example.
A field equipped with a discrete valuation is called a \emph{discrete valuation field} (DVF).

Valuations of determinants of matrices over a DVF often appear as matrix formulations of combinatorial optimization problems.
For example, \emph{weighted Edmonds' problem}~(WEP), which is to compute the degree of the determinant of a polynomial matrix having symbols, reduces to the weighted bipartite matching problem and the weighted linear matroid intersection and parity problems depending on symbols' pattern~\cite{Hirai2019}.
Conversely, the degree of the determinant of an arbitrary polynomial matrix serves as a lower bound on the maximum weight of a perfect matching in the associated edge-weighted bipartite graph.
Based on this relation, the \emph{combinatorial relaxation} algorithm of Murota~\cite{Murota1995a} computes the degree of the determinant of a polynomial matrix by iteratively solving the weighted bipartite matching problem.

Computing valuations of determinants is also applied to linear differential equations.
Consider a linear differential equation
\begin{align}\label{eq:time-invariant-linear-differential-equation}
  A_0 y + A_1 y' + \dotsb + A_l y^{(l)} = 0
\end{align}
for $\funcdoms{y}{\setR}{\setC^n}$, where $A_0, \dotsc, A_l \in \setC^{n \times n}$.
The set of all solutions of~\eqref{eq:time-invariant-linear-differential-equation} forms a vector space over $\setC$.
Classical Chrystal's theorem~\cite{Chrystal1897} states that the dimension of the solution space of~\eqref{eq:time-invariant-linear-differential-equation} is equal to the degree of the determinant of $A_0 + A_1 s + \dotsb + A_l s^l \in {\setC[s]}^{n \times n} \hookrightarrow {\setC(s)}^{n \times n}$.
Hence one can analyze the degrees of freedom of linear time-invariant systems by computing valuations of determinants of matrices over a DVF.

This paper addresses a noncommutative generalization of computing valuations of determinants.
A \emph{discrete valuations skew field} (DVSF) is naturally defined as in the commutative case~\cite{Warner1993}.
The \emph{Dieudonné determinant}~\cite{Dieudonne1943}, denoted by $\Det$, is a generalization of the determinant for matrices over skew fields (see \cref{sec:matrices-over-skew-fields} for definition).
The Dieudonné determinant retains useful properties of the usual determinant such as $\Det AB = \Det A \Det B$.
While $\Det A$ for $A \in F^{n \times n}$ is no longer an element in a skew field $F$, when $F$ is a DVSF, its valuation $\vdet(A) \defeq v\prn{\Det A}$ is well-defined.

In the following of this introduction, we first describe applications of valuations of the Dieudonné determinants in \cref{sec:weighted-edmonds-problem,sec:linear-differential-difference-equations}.
Then \cref{sec:computational-model} states a computational model which we use and \cref{sec:contributions} presents our contributions.
Related work and organization of this paper are described in \cref{sec:related-work} and \cref{sec:organization}, respectively.

\subsection{Weighted Edmonds' Problem}\label{sec:weighted-edmonds-problem}
In 1967, Edmonds~\cite{Edmonds1967} posed a question whether there exists a polynomial-time algorithm to compute the rank of a \emph{linear} (\emph{symbolic}) \emph{matrix} $B$ over a field $K$, which is in the form
\begin{align} \label{eq:linear_matrix}
  B = B_0 + B_1 x_1 + \cdots + B_m x_m,
\end{align}
where $B_0, B_1 \ldots, B_m \in K^{n \times n}$ and $x_1, \ldots, x_m$ are commutative symbols.
Here, $B$ is regarded as a matrix over the polynomial ring $K[x_1, \ldots, x_m]$ or the rational function field $K(x_1, \ldots, x_m)$.
In case where $B$ is the Edmonds or Tutte matrix of a bipartite or nonbipartite graph $G$, the rank computation for $B$ corresponds to solving the maximum matching problem on $G$.
More generally, Lovász~\cite{Lovasz1989} showed that Edmonds' problem is equivalent to a linear matroid intersection problem if all $B_i$ are of rank 1, and to a linear matroid parity problem if all $B_i$ are skew-symmetric matrices of rank 2.
For general linear matrices, the celebrated Schwartz--Zippel lemma~\cite{Schwartz1980} provides a simple randomized algorithm if $\card{K}$ is large enough~\cite{Lovasz1989}.
However, no deterministic polynomial-time algorithm still has been known; the existence of such an algorithm would imply nontrivial circuit complexity lower bounds~\cite{Kabanets2004,Valiant1979}.

Recent studies~\cite{Garg2016,Hamada2020,Ivanyos2018} address the noncommutative version of Edmonds' problem (nc-Edmonds' problem).
This is a problem of computing the \emph{noncommutative rank} (nc-rank) of $B$, which is the rank defined by regarding $x_1, \ldots, x_m$ as pairwise noncommutative, i.e., $x_i x_j \neq x_j x_i$ if $i \neq j$.
In this way, $B$ is viewed as a matrix over the free ring $K \agbr{x_1, \ldots, x_m}$ generated by noncommutative symbols $x_1, \ldots, x_m$.
The nc-rank of $B$ is precisely the rank of $B$ over a skew (noncommutative) field $K \fsbr{x_1, \ldots, x_m}$, called a $\emph{free skew field}$, which is the quotient of $K \agbr{x_1, \ldots, x_m}$ defined by Amitsur~\cite{Amitsur1966}.
We call a linear matrix over $K$ having noncommutative symbols an \emph{nc-linear matrix} over $K$.
The recent studies~\cite{Garg2016,Hamada2020,Ivanyos2018} revealed that nc-Edmonds' problem is deterministically tractable.
For the case where $K$ is the set $\setQ$ of rational numbers, Garg et al.~\cite{Garg2016} proved that Gurvits' \emph{operator scaling algorithm}~\cite{Gurvits2004} deterministically computes the nc-rank of $B$ in $\poly(n, m)$ arithmetic operations on $\setQ$.
Algorithms over general field $K$ were later given by Ivanyos et al.~\cite{Ivanyos2018} and Hamada--Hirai~\cite{Hamada2020} exploiting the min-max theorem established for nc-rank.
When $K = \setQ$, these algorithms run in time polynomial in the bit-length of the input.

Hirai~\cite{Hirai2019} introduced a weighted version of Edmonds' problem.
First, consider commutative symbols $x_1, \ldots, x_m$ and an extra commutative symbol $s$.
Define a matrix
\begin{align}\label{eq:linear_polynomial_matrix}
  A = A_l + A_{l-1} s + \cdots + A_0 s^l,
\end{align}
where $A_d = A_{d,0} + A_{d,1}x_1 + \cdots + A_{d,m}x_m \in {K[x_1, \ldots, x_m]}^{n \times n}$ is a linear matrix over $K$ for $d = 0, \ldots, l$.
We call~\eqref{eq:linear_polynomial_matrix} a \emph{linear polynomial matrix} over $K$.
The \emph{weighted Edmonds' problem} (WEP) is the problem to compute the degree (in $s$) of the determinant of $A$.
Analogously to Edmonds' problem, WEP includes a bunch of weighted combinatorial optimization problems as special cases, such as a maximum weighted perfect matching problem, a weighted linear matroid intersection problem and a weighted linear matroid parity problem; see~\cite[Section~5]{Hirai2019}.

Next, let $x_1, \ldots, x_m$ be noncommutative symbols and $s$ an extra symbol that commutes with any element in $K\agbr{x_1, \ldots, x_m}$.
An \emph{nc-linear polynomial matrix} $A$ over $K$ is a matrix in the form of~\eqref{eq:linear_polynomial_matrix} with each $A_d$ regarded as an nc-linear matrix.
Then $A$ can be viewed as a matrix over the rational function (skew) field $F \defeq K\fsbr{x_1, \ldots, x_m}(s)$.
Now $F$ is a DVSF equipped with discrete valuation $-\deg$.
\emph{Noncommutative weighted Edmonds' problem} (nc-WEP) is the problem to compute $\deg \Det$ of a given nc-linear polynomial matrix.
Hirai~\cite{Hirai2019} formulated the dual problem of nc-WEP as the minimization of an \emph{L-convex function} on a \emph{uniform modular lattice}, and gave an algorithm based on the steepest gradient descent.
Hirai's algorithm uses $\poly(n,m,l)$ arithmetic operations on $K$ while no bit-length bound has been given for $K = \setQ$.

\subsection{Linear Differential/Difference Equations}\label{sec:linear-differential-difference-equations}
Polynomials in differential or difference operators give rise to noncommutative valuations.
Let $K$ be a skew field, $\funcdoms{\sigma}{K}{K}$ a ring automorphism, and $\funcdoms{\delta}{K}{K}$ a (\emph{left}) \emph{$\sigma$-derivation}; that is, it is additive, i.e., $\delta(a+b) = \delta(a) + \delta(b)$, and $\delta(ab) = \sigma(a)\delta(b) + \delta(a)b$ for $a,b \in K$.
A \emph{skew polynomial}, or an \emph{Ore polynomial} due to Ore~\cite{Ore1933}, over $(K, \sigma, \delta)$ in indeterminate $s$ is a polynomial over $K$ with the usual addition and a twisted multiplication defined by the commutation rule
\begin{align}\label{eq:skew_polynomial_commutation_rule}
  sa = \sigma(a)s + \delta(a)
\end{align}
for $a \in K$.
The \emph{skew polynomial ring} over $(K, \sigma, \delta)$ is denoted by $K[s; \sigma, \delta]$.
Besides the polynomial ring $K[s]$, the ring $\setC(t)[\partial; \id, ']$ of differential operators is an example of a skew polynomial ring, where $\funcdoms{'}{\setC(t)}{\setC(t)}$ is the usual differentiation.
Another example is the ring $\setC(t)[S; \tau, 0]$ of shift operators, where $\funcdoms{\tau}{\setC(t)}{\setC(t)}$ is defined by $f(t) \mapsto f(t+1)$ for $f \in \setC(t)$.
The degree of a skew polynomial is naturally defined and it extends to the \emph{skew rational function field} $K(s; \sigma, \delta)$, which is the \emph{Ore quotient skew field} of $K[s; \sigma, \delta]$.
Then $K(s; \sigma, \delta)$ is a DVSF with valuation $-\deg$.

Let $K$ be a field of characteristic $0$ equipped with an ($\id$-)derivation $\delta$.
Consider a linear differential equation
\begin{align}\label{eq:time-varying-linear-differential-equation}
  A_0 y + A_1 \delta(y) + \dotsb + A_l \delta^l(y) = 0
\end{align}
for $y \in K^n$ with $A_0, \dotsc, A_l \in K^{n \times n}$.
Taelman~\cite{Taelman2006} showed that the dimension of the solution space (over an adequate field extension of $K$) of~\eqref{eq:time-varying-linear-differential-equation} is equal to $\deg \Det A$ with $A \defeq A_0 + A_1 s + \dotsb + A_l s^l \in {K[s; \id, \delta]}^{n \times n}$.
This is a ``time-varying'' generalization of Chrystal's theorem.
We show that the assumption on the characteristic can be removed and a similar formula holds for linear difference equations using two kinds of valuations (see \cref{sec:application2}).
In this way, computing valuations of the Dieudonné determinants of matrices over DVSFs can be applied to analysis of time-varying linear differential or difference equations.

\subsection{Computational Model}\label{sec:computational-model}
We design algorithms to compute $\vdet(A)$ for a matrix $A$ over a DVSF $F$ without restricting $F$ to a skew rational function field so that the algorithms can be applied as widely as possible.
Here we need to clarify a computational model to deal with representation of elements in $F$ and operations on $F$.
The simplest model is the arithmetic model on $F$, i.e., an element in $F$ is stored in a unit memory cell and we can perform arithmetic operations on $F$ in constant time.
In this model, one can compute $\zeta(A)$ in $\Order(n^\omega)$-time by the Gaussian elimination, where $\omega$ is the exponent in the time complexity of multiplying two matrices.
However, this model is too simplified and cannot catch the computational cost needed in the standard representation of some DVSF like $F = K(s)$.

As a representation of elements in $F$, we adopt the \emph{$\pi$-adic expansion}, in which each $a \in F$ is expressed as a formal Laurent series
\begin{align}
  a = \sum_{d=l}^\infty a_d \pi^d.
\end{align}
Here, $l \in \setZ$, $\pi \in F$ is a fixed element with $v(\pi) = 1$ called a \emph{uniformizer}, and $a_l, a_{l+1}, \dotsc$ are elements in a fixed subset $Q \subseteq F$, called a \emph{representative set}.
The representative set is selected so that the $\pi$-adic expansion is unique; such $Q$ exists for any DVSF.
While we would like to adopt the ``arithmetic model on $Q$'', the set $Q$ might not be a skew field, i.e., arithmetic operations on $Q$ might not be closed.
We thus require $F$ to have a representative set that is a skew subfield of $F$.
Such a DVSF is called \emph{split}~\cite{Dumas1992}.

Let $F$ be a split DVSF with a closed representative set $K$, called the \emph{coefficient skew subfield}.
The ring structure of $F$ is completely determined from the commutation rule between a uniformizer $\pi \in F$ and each $a \in K$.
The element $\pi a$ is uniquely expressed as
\begin{align}\label{eq:higher_sigma_derivation_pi}
  \pi a = \sum_{d=0}^\infty \delta_d(a) \pi^{d+1},
\end{align}
where $\funcdoms{\delta_d}{K}{K}$ is a map satisfying the axioms of \emph{higher $\sigma$-derivations}~\cite{Roux1986}.
We also assume the oracle access to each $\delta_d$, i.e., we can compute $\delta_d(a)$ in constant time for each $d \in \setN$ and $a \in K$.

\subsection{Contributions}\label{sec:contributions}
Under the above setting, this paper presents two algorithms to compute $\vdet(A)$ for $A \in F^{n \times n}$, both of which are based on combinatorial optimization.
The first algorithm is a generalization of the \emph{combinatorial relaxation} of Murota~\cite{Murota1995a} that computes $\deg \det$ of polynomial matrices over a field.
Constructing an edge-weighted bipartite graph $G(A)$ from $A$ reflecting the valuation of each entry, one can show that $\vdet(A)$ is lower bounded by the minimum weight of a perfect matching of $G(A)$.
Based on this relation, the combinatorial relaxation algorithm computes $\vdet(A)$ by iteratively solving the weighted matching problem.

The second algorithm generalizes the \emph{matrix expansion}, which reduces the computation of $\vdet(A)$ to the rank computation of a block matrix over $K$ obtained by arranging coefficient matrices of $\pi^i A$ with $i \in \setN$.
The correctness of the matrix expansion essentially relies on the \emph{Legendre conjugacy} between integer sequences of the valuations of minors of $A$ and ranks of block matrices.
The Legendre conjugacy is an important duality relation on discrete convex and concave functions treated in \emph{discrete convex analysis}~\cite{Murota2003}.
Our matrix expansion generalizes algorithms of Van Dooren et al.~\cite{Vandooren1979} for $\deg \det$ on $\setC(s)$ and Moriyama--Murota~\cite{Moriyama2013} for $\deg \det$ on $K(s)$ with a field $K$.

The running times of our algorithms are estimated as follows.

\begin{theorem}\label{thm:complexity-of-combinatorial-relaxation}
  Let $F$ be a split DVSF with uniformizer $\pi$ and coefficient skew subfield $K$.
  Let $A = \sum_{d=0}^l A_d \pi^d \in F^{n \times n}$ be a square matrix over $F$ with $A_0, \dotsc, A_l \in K^{n \times n}$.
  Given $A_0, \dotsc, A_l$ and $M \in \setN$ such that $\vdet(A) \le M$ or $A$ is singular, we can compute $\vdet(A)$ by the combinatorial relaxation algorithm in $\Order\prn{M^3n^2 + M^2n^\omega + Mn^{2.5}}$-time and by the matrix expansion algorithm in $\Order\prn{M^3n^2 + M^\omega n^\omega}$-time.
\end{theorem}

As shown in \cref{thm:complexity-of-combinatorial-relaxation}, our algorithms additionally require an upper bound $M$ on $\vdet(A)$ by technical reasons.
While estimating such $M$ seems to be difficult for general DVSFs, one can adopt $M \defeq ln$ for $A = \sum_{d=0}^l A_l s^d \in {K[s]}^{n \times n}$ with $K$ being a field.
This indeed holds for skew polynomial rings and it yields the following corollary:

\begin{theorem}\label{thm:complexity-for-deg-Det-of-skew-polynomial-matrices}
  Let $A = \sum_{d=0}^l A_l s^d \in {K[s; \sigma, \delta]}^{n \times n}$ be a square skew polynomial matrix over a skew field $K$.
  Under the arithmetic model on $K$ and oracle access to $\sigma^{-1}$ and $\delta$, we can compute $\deg \Det A$ in $\Order\prn{l^2 n^{\omega + 2} + ln^{4.5}}$-time by the combinatorial relaxation algorithm and in $\Order\prn{l^\omega n^{2\omega}}$-time by the matrix expansion algorithm.
\end{theorem}

We further show that the converse holds, i.e., $\vdet(A) \le ln$ for any nonsingular $A \in F^{n \times n}$ only if $F$ is isomorphic to (an extension of) a skew rational function field.
This fact indicates that skew polynomial rings are characterized as the most general ring structure that admits natural extensions of the combinatorial relaxation and matrix expansion algorithms.

We cannot directly apply \cref{thm:complexity-of-combinatorial-relaxation} to weighted Edmonds' problem because arithmetic operations on $K(x_1, \ldots, x_m)$ nor $K\fsbr{x_1, \ldots, x_m}$ cannot be performed in constant time under the arithmetic model on $K$.
However, using the min-max formula on nc-Edmonds' problem by Fortin--Reutenauer~\cite{Fortin2004}, we can modify the combinatorial relaxation algorithm so that it can be used for reducing the nc-WEP to the unweighted problem.
This algorithm coincides with that given by Hirai~\cite{Hirai2019}.
Furthermore, the matrix expansion algorithm can be used for reductions of both commutative and noncommutative problems.
Using polynomial-time algorithms for nc-Edmonds' problem, we show:

\begin{theorem}\label{thm:nc-WEP-complexity}
  The nc-WEP over a field $K$ can be deterministically solved using polynomially many arithmetic operations on $K$.
  When $K = \setQ$, the algorithm runs in time polynomial in the binary encoding length of the input.
\end{theorem}

\subsection{Related Work}\label{sec:related-work}
In computer algebra, algorithms were proposed for computing various kinds of canonical forms of a skew polynomial matrix $A \in {K[s; \sigma, \delta]}^{n \times n}$ such as the \emph{Jacobson normal form}~\cite{Levandovskyy2011}, the \emph{Hermite normal form}~\cite{Giesbrecht2013}, the \emph{Popov normal form}~\cite{Khochtali2017} and their weaker form called a \emph{row-reduced form}~\cite{Abramov2014b,Beckermann2006}.
One can use these algorithms to calculate $\deg \Det A$ since it is immediately obtained from the canonical forms of $A$.
These algorithms iteratively solve systems of linear equations over $K$.
Our algorithms are faster than the existing algorithms.
The fastest known algorithm given by Giesbrecht--Kim~\cite{Giesbrecht2013} runs in $\Order\big(l^\omega n^{2\omega+2} \log l n\big)$-time, whereas our two algorithms require only $\Order\prn{l^2 n^{\omega+2} + ln^{4.5}}$-time and $\Order\prn{l^\omega n^{2\omega}}$-time as seen in \cref{thm:complexity-for-deg-Det-of-skew-polynomial-matrices}.

Hamada--Hirai~\cite{Hamada2020} presents an algorithm for nc-Edmonds' problem over $\setQ$ that runs in time polynomial in the bit-length of the input.
They introduce a quantity conceptually corresponding to $p$-adic valuations of the Dieudonné determinants for matrices over $F \defeq \setQ\fsbr{x_1, \dotsc, x_m}$ and the algorithm computes it based on the procedure of the combinatorial relaxation.
Since $\setQ$ with the $p$-adic valuation is not split, their algorithm can be seen as a kind of an extension of the combinatorial relaxation to a special but non-split DVSF, except that the quantity has not been proved to be some discrete valuation of the Dieudonné determinants on $F$ indeed.

\subsection{Organization}\label{sec:organization}
The rest of this paper is organized as follows.
\Cref{sec:preliminaries-on-valuations-skew-fields,sec:preliminaries-on-matrices} describe preliminaries on valuation skew fields and matrices over them, respectively.
\Cref{sec:combinatorial-aspects-of-valuations-and-matrices} explains that relations between matrices over valuation fields and combinatorial optimization problems, which are well-known for the commutative case, still hold in the noncommutative case.
\Cref{sec:combinatorial-relaxation-algorithm,sec:matrix-expansion-algorithm} propose our algorithms, the combinatorial relaxation and matrix expansion algorithms, respectively.
\Cref{sec:estimating-upper-bounds} discusses an estimation of the upper bound $M$ on $\vdet(A)$.
Finally, \cref{sec:application1,sec:application2} describe applications to weighted Edmonds' problem and linear differential/difference equations, respectively.

\section{Preliminaries on Valuation Skew Fields}\label{sec:preliminaries-on-valuations-skew-fields}
We denote the set of nonnegative integers by $\setN$, the integers by $\setZ$, the rational numbers by $\setQ$, the real numbers by $\setR$, and the complex numbers by $\setC$.
For $n \in \setN$, define $\intset{n} \defeq \set{1, 2, \dotsc, n}$ and $\intset{0, n} \defeq \set{0, 1, 2, \dotsc, n}$.
All rings are assumed to have the multiplicative identity.

\subsection{Valuation Skew Fields}\label{sec:valuation-skew-fields}
A \emph{skew field}, or a \emph{division ring} is a ring $F$ such that every nonzero element has a multiplicative inverse in $F$.
A (\emph{real}) \emph{valuation skew field}~\cite[Chapter~IV]{Warner1993} is a skew field $F$ endowed with a (\emph{real}) \emph{valuation}, that is, a map $\funcdoms{v}{F}{\setR \cup \set{+\infty}}$ satisfying~\ref{item:V1}--\ref{item:V4}.
A valuation skew field is called a \emph{valuation field} if it is a field.
The value $v(a)$ for $a \in F$ is called the \emph{valuation} of $a$.

By~\ref{item:V1} and~\ref{item:V3}, it holds $v(-a) = v(a)$ and $v\prn{a^{-1}} = -v(a)$ for all $a \in \mult{F}$, where $\mult{F} = F \setminus \set{0}$ is the multiplicative group of $F$.
In particular, we have $v(a) < +\infty$ for $a \in \mult{F}$.
The equality in~\ref{item:V2} is attained whenever $v(a) \ne v(b)$; otherwise, if $v(a) < v(a+b)$ and $v(a) < v(b)$, it holds
\begin{align}
  v(a) = v((a+b)-b) \ge \min\set{v(a+b), v(-b)} = \min\set{v(a+b), v(b)} > v(a),
\end{align}
a contradiction.

The (\emph{invariant}) \emph{valuation ring} of a valuation skew field $F$ with respect to a valuation $v$ is a set
\begin{align}
  R \defeq \set{a \in F}[v(a) \ge 0].
\end{align}
Then $R$ is a subring of $F$ by~\ref{item:V1} and~\ref{item:V2}, and is a \emph{domain}, i.e., $R$ has no zero-divisors.
It also satisfies the following~\cite[Chapter~1]{Krylov2008}:
\begin{enumerate}[label={\upshape{(VR\arabic*)}}]
  \item either $a \in R$ or $a^{-1} \in R$ for $a \in \mult{F}$,\label{item:VR1}
  \item $aR = Ra$ for $a \in \mult{F}$.\label{item:VR2}
\end{enumerate}
In addition, $R$ is a \emph{local ring}, i.e., it has a unique maximal right (and indeed a unique maximal left) ideal $J(R)$, which coincides with $R \setminus \mult{R}$ with $\mult{R} = \set{a \in F}[v(a) = 0]$.
Namely, it holds
\begin{align}\label{eq:real_valuation_ideal}
  J(R) = \set{a \in F}[v(a) > 0].
\end{align}
The quotient ring $R \extends J(R)$ forms a skew field, called the \emph{residue skew field} of $F$ (or a \emph{residue field} if it is a field).

A \emph{representative set} of $F$ is a subset $Q$ of $R$ such that $0 \in Q$ and the restriction to $Q$ of the canonical homomorphism from $R$ to the residue skew field $K \defeq R \extends J(R)$ is a bijection from $Q$ to $K$.
Then for $a \in R$, there uniquely exists $a_0 \in Q$ such that $a \in a_0 + J(R)$.
Hence $a - a_0 \in J(R)$, which means:

\begin{proposition}\label{prop:representation_on_valuation_ring}
  Let $F$ be a valuation skew field with valuation $v$, valuation ring $R$, and representative set $Q$.
  Then any $a \in R$ is uniquely expressed as $a = a_0 + \tilde{a}$, where $a_0 \in Q$ and $\tilde{a} \in J(R)$.
\end{proposition}

The \emph{value group} of $v$ is the additive subgroup $v\prn{\mult{F}}$ of $\setR$.
A \emph{discrete valuation} is a valuation $F$ whose value group is $\setZ$.
A valuation skew field equipped with a discrete valuation is called a \emph{discrete valuation skew field} (DVSF), which is of the main interest of this thesis.
If $F$ is a field, we call $F$ a \emph{discrete valuation field} (DVF).

Let $F$ be a DVSF with discrete valuation $v$ and the valuation ring $R$.
Then~\eqref{eq:real_valuation_ideal} is
\begin{align}\label{eq:valuation_ideal}
  J(R) = \set{a \in F}[v(a) \ge 1].
\end{align}
Any element $\pi \in R$ with $v(\pi) = 1$ is called a \emph{uniformizer} or a \emph{prime element} of $F$.
In addition to~\ref{item:VR1} and~\ref{item:VR2}, $R$ enjoys the following properties~\cite[Chapter~1]{Krylov2008}:
\begin{enumerate}[label={\upshape{(DVR\arabic*)}}]
  \item $J(R) = \pi R = R \pi$,\label{item:DVR1}
  \item $\displaystyle \bigcap_{d=1}^\infty {J(R)}^d = \set{0}$.\label{item:DVR2}
\end{enumerate}
Note that it holds
\begin{align}\label{eq:any_ideal_of_DVSF}
  {J(R)}^d = \pi^d R = R \pi^d = \set{a \in F}[v(a) \ge d]
\end{align}
by~\eqref{eq:valuation_ideal} and~\ref{item:DVR1} for $d \in \setN$.
In addition, any right ideal and left ideal of $R$ are two-sided and are in the form of~\eqref{eq:any_ideal_of_DVSF}.
This mean that $R$ is a (right and left) \emph{principal ideal domain} (PID), which is a domain whose every (right and left) ideal is generated by one element.
More strongly, any DVR is a (right and left) \emph{Euclidean domain}~\cite{Brungs1973} as is well-known for commutative DVRs.
Here, a domain $R$ is said to be \emph{Euclidean} if there exists a map $\funcdoms{f}{R}{\setN \cup \set{-\infty}}$, called an \emph{Euclidean map}, such that for every $a,b \in R$ with $b \ne 0$, there exist $q,r,q',r' \in R$ such that $a = bq + r = q'b + r'$ and $f(r), f(r') < f(b)$.
In case of a valuation ring of a DVSF, $-v$ serves as an Euclidean map.
We remark that Euclidean domains are proper subclass of PIDs even for noncommutative rings~\cite{Brungs1973}.

\begin{remark}\label{rem:DVR}
  In general, a local ring $R$ satisfying~\ref{item:DVR1} and~\ref{item:DVR2} for some non-nilpotent element $\pi \in R$ is called a \emph{discrete} (\emph{invariant}) \emph{valuation ring} (DVR).
  Here, an element $a \in R$ is said to be \emph{nilpotent} if $a^k = 0$ for some $k \in \setN$ and \emph{non-nilpotent} if not.
  The valuation ring of any DVSF is a DVR as described above.
  Indeed, any DVR $R$ is the valuation ring of some DVSF~\cite{Krylov2008}; here we give a construction of the DVSF briefly.
  First, it follows from~\ref{item:DVR1} and~\ref{item:DVR2} that $R$ is a PID\@.
  Then $R$ is also a (right and left) \emph{Ore domain}, which is a domain such that for each $s, t \in R \setminus \set{0}$, there exist $x,y,z,w \in R \setminus \set{0}$ satisfying $sx = ty$ and $zs = wt$~\cite[Corollarly~6.7]{Goodearl2004}.
  This property enables for $R$ to have the \emph{Ore quotient skew field} $F$, which is a skew field of fractions each of whose elements $a \in F$ is expressed as $a = sx^{-1} = y^{-1}t$ for some $s, t \in R$ and $x, y \in R \setminus \set{0}$.
  In particular, $a \in \mult{F}$ can be uniquely expressed as $a = \pi^k p = q \pi^k$ for some $p,q \in \mult{R}$ and $k \in \setZ$.
  Denote this $k$ by $v(a)$ for $a \in \mult{F}$ and let $v(0) \defeq +\infty$.
  Then $\funcdoms{v}{F}{\setZ \cup \set{+\infty}}$ is a discrete valuation on $F$, whose valuation ring coincides with $R$.
  We refer to the restriction of $v$ onto $R$ as the valuation of $R$ and a representative set of $R$ means that of $F$.
  See~\cite[Chapter~1]{Krylov2008} for details of DVRs and~\cite[Chapter~6]{Goodearl2004} for Ore domains and quotient skew fields.
\end{remark}

Let $F$ be a DVSF with valuation $v$ and uniformizer $\pi$.
For an arbitrary real number $c > 1$, we define $\funcdoms{d}{F \times F}{\setR}$ as
\begin{align}
  d(a, b) \defeq c^{-v(a-b)}
\end{align}
for $a, b \in F$ (where $c^{-\infty} \defeq 0$).
Then $d$ forms a metric on $F$.
The \emph{$\pi$-adic topology} is the ring topology on $F$ induced by $d$, which does not depend on the choice of $c$.
On this topology, $\set[\big]{a + {J(R)}^k}[k \in \setN]$ is an open neighborhood system of $a \in F$ by~\eqref{eq:any_ideal_of_DVSF}.
A DVSF is said to be \emph{complete} if it is complete as a metric space.
Then any DVSF can be extended to a complete DVSF as follows.

\begin{theorem}[{\cite[Theorem~17.2]{Warner1993}}]\label{thm:completion}
  Let $F$ be a DVSF with discrete valuation $v$.
  Then there uniquely exists a complete DVSF $\hat{F}$ with discrete valuation $\hat{v}$ such that $\hat{F}$ contains $F$ as a dense subring and $\hat{v}$ extends $v$.
  In addition, the residue skew field of $\hat{F}$ is isomorphic to that of $F$.
\end{theorem}

The complete DVSF $\hat{F}$ in \cref{thm:completion} is called the \emph{completion} of $F$.
By \cref{thm:completion}, it is convenient to consider complete DVSFs from the beginning.
See~\cite{Warner1993} for details of topological rings and the $\pi$-adic topology.

Let $F$ be a DVSF with uniformizer $\pi$, valuation ring $R$, and representative set $Q$.
By \cref{prop:representation_on_valuation_ring} and~\ref{item:DVR1}, we can express $a \in R$ as $a = a_0 + a'\pi$ by some $a_0 \in Q$ and $a' \in R$.
By the same argument, there are unique $a_1 \in Q$ and $a'' \in R$ such that $a' = a_1 + a''\pi$.
Therefore, we have $a = a_0 + a_1 \pi + a''\pi^2$.
Repeating this argument, we can represent $a$ as a power series in $\pi$ with coefficient $Q$, which is formally stated as follows.

\begin{proposition}[{\cite[Theorem~18.5]{Warner1993}}]\label{prop:dvsf_as_power_series}
  Let $F$ be a DVSF with discrete valuation $v$ and let $\pi$ and $Q$ be a uniformizer and a representative set of $F$, respectively.
  \begin{enumerate}
    \item For every $a \in F$, there uniquely exists a sequence $\prn{a_d}_{d \in \setZ}$ of elements in $Q$ such that $a_d = 0$ for all but finitely many $d < 0$ and a power series
    \begin{align}\label{eq:power_of_pi}
      \sum_{d \in \setZ} a_d \pi^d
    \end{align}
    converges to $a$ in the $\pi$-adic topology.
    If $l \defeq v(a) \in \setZ$, then $a_d = 0$ for $d < l$ and $a_l \ne 0$.

    \item If $F$ is complete and $\prn{a_d}_{d \in \setZ}$ is a sequence of elements in $Q$ such that $a_d = 0$ for all but finitely many $d < 0$, the power series~\eqref{eq:power_of_pi} converges to an element $a$ of $F$.
    Its valuation $v(a)$ is equal to the minimum $l \in \setZ$ such that $a_d = 0$ for $d < l$ and $a_l \ne 0$.
  \end{enumerate}
\end{proposition}

We call~\eqref{eq:power_of_pi} the \emph{$\pi$-adic expansion} of $a \in F$.

\subsection{Examples of Valuation Skew Fields}\label{sec:examples_of_DVSFs}
We present several examples of valuation skew fields.
All examples are DVSFs except for \cref{ex:formal_laurent_series_with_real_exponents}.

\begin{example}[formal Laurent series]\label{ex:formal-Laurent-series}
  Let $K$ be a skew field.
  Denote by $K[s]$ the polynomial ring over $K$ in indeterminate $s$ that commutes with any element of $K$.
  Since $K[s]$ is an Ore domain, it has the quotient skew field $K(s)$, called the \emph{rational function} (skew) \emph{field}.
  The \emph{order} $\ord p$ of $p \in K[s] \setminus \set{0}$ is the minimum $d \in \setN$ such that the coefficient of $s^d$ in $p$ is nonzero.
  We also define $\ord f$ for $f \in K(s) \setminus \set{0}$ as $\ord f \defeq \ord p - \ord q$, where $f = p/q$ with $p,q \in K[s] \setminus \set{0}$.
  Set $\ord 0 \defeq +\infty$.
  Then it is well-known that the order is a discrete valuation on $K(s)$ and the residue skew field is $K$.
  A canonical (but not unique) choice of a uniformizer is $s$.
  The completion of $K(s)$ is the \emph{formal Laurent series} (skew) \emph{field} $K\pprn{s}$ over $K$ in $s$, whose each element is expressed as
  \begin{align}\label{eq:laurent_fps}
    f = \sum_{d=l}^\infty a_d s^{d}
  \end{align}
  with $l \in \setZ$ and $a_l, a_{l+1}, \dotsc \in K$.
  If $a_l \ne 0$, then $l = \ord f$.
  The valuation ring of $K\pprn{s}$ is called the \emph{formal power series} (skew) \emph{field} $K\ssqbr{s}$ over $K$ in $s$, which is the subring of $K\pprn{s}$ consisting of formal power series
  \begin{align}\label{eq:fps}
    f = \sum_{d=0}^\infty a_d s^{d}
  \end{align}
  with $a_0, a_1, \dotsc \in K$.

  Similarly, the \emph{degree} $\deg p$ of $p \in K[s] \setminus \set{0}$ is defined by replacing ``minimum'' with ``maximum'' in the definition of $\ord p$.
  Define $\deg f$ for $f = p/q \in \mult{K(s)}$ with $p,q \in K[s] \setminus \set{0}$ as $\deg f \defeq \deg p - \deg q$ and $\deg 0 \defeq -\infty$ as well.
  Since $\deg f(s) = -\ord f\prn{s^{-1}}$, the minus of the degree is a discrete valuation on $K(s)$ with uniformizer $s^{-1}$ and residue skew field $K$.
  The completion of $K(s)$ with respect to the minus degree is $K\pprn{s^{-1}}$, which is a field isomorphic to $K\pprn{s}$.
\end{example}

\begin{example}[{formal Laurent series with real exponents}]\label{ex:formal_laurent_series_with_real_exponents}
  Let $K$ be a skew field.
  A subset $X$ of $\setR$ is said to be \emph{well-ordered} if any nonempty subset of $X$ has the minimum element.
  We consider \emph{formal Laurent series with real exponents}, each of which is in the following form
  \begin{align}\label{eq:formal_series_with_real}
    f = \sum_{x \in X} a_x s^x,
  \end{align}
  where $X \subsetneq \setR$ is well-ordered, $a_x \in \mult{K}$ for $x \in X$, and $s$ is a formal ``indeterminate'' that satisfies $s^{x+y} = s^x s^y$ and $as^x = s^x a$ for $x,y \in \setR$ and $a \in K$.
  Addition on these series is naturally defined, and the multiplication of $f = \sum_{x \in X} a_x s^x$ and $g = \sum_{y \in Y} b_y s^y$ is given by
  \begin{align}
    fg \defeq \sum_{z \in \setR} \prn{\sum_{\substack{x \in X, y \in Y\\x+y = z}} a_x b_y} s^z.
  \end{align}
  For every $z \in \setR$, the number of $(x, y) \in X \times Y$ satisfying $x+y = z$ is finite from the assumption that $X$ and $Y$ are well-ordered, and the set
  \begin{align}
    \set*{z \in \setR}[the coefficient of $s^z$ in $fg$ is nonzero]
  \end{align}
  is well-ordered as well.
  Hence $fg$ is a formal Laurent series again in the sense defined above.
  By these operations, the set $\Sigma$ of formal Laurent series with real exponents forms a skew field~\cite[Theorem~5.7]{Neumann1949}.

  Define the \emph{order} $\ord f$ of~\eqref{eq:formal_series_with_real} as the minimum $x \in X$.
  We also define $\ord 0 \defeq +\infty$.
  Then as Neumann~\cite{Neumann1949} indicated, $\ord$ is a valuation on $\Sigma$ that is not discrete.
  The residue skew field of $\Sigma$ is $K$.
  The skew field $\Sigma$ contains $K\pprn{s}$ as a subfield, and the restrictions of the order onto $K\pprn{s}$ coincides that on $K\pprn{s}$.
  Reversing the ordering of $\setR$, we can also define $\deg f$ consistent with $K\pprn{s^{-1}}$ in the completely analogous way.
\end{example}

\begin{example}[{$p$-adic numbers}]
  Let $p$ be a prime number.
  The \emph{$p$-adic valuation} $v_p(n)$ of $n \in \setZ \setminus \set{0}$ is the maximum $k \in \setN$ such that $p^k$ divides $n$, and is extended to $\mult{\setQ}$ as $v_p(x) \defeq v_p(n) - v_p(m)$ for $x = n/m \in \mult{\setQ}$ with $n, m \in \setZ \setminus \set{0}$.
  Also we define $v_p(0) \defeq +\infty$.
  Then $v_p$ is a discrete valuation on $\setQ$ with uniformizer $p$.
  The residue field is $\setF_p$.
  The completion of $\setQ$ with respect to $v_p$ is the field $\setQ_p$ of \emph{$p$-adic numbers}.
\end{example}

\begin{example}[skew (inverse) Laurent series]\label{ex:skew_laurent_series}
  Let $K$ be a skew field, $\funcdoms{\sigma}{K}{K}$ a ring automorphism, and $\funcdoms{\delta}{K}{K}$ a \emph{left $\sigma$-derivation}; that is, it is additive, i.e., $\delta(a+b) = \delta(a) + \delta(b)$, and it satisfies $\delta(ab) = \sigma(a)\delta(b) + \delta(a)b$ for all $a,b \in K$.
  The (\emph{left}) \emph{skew polynomial ring}, or the \emph{Ore polynomial ring} due to Ore~\cite{Ore1933} over $(K, \sigma, \delta)$ in indeterminate $s$, which is denoted by $K[s; \sigma, \delta]$, is a polynomial ring over $K$ with the usual addition and a twisted multiplication defined by the commutation rule~\eqref{eq:skew_polynomial_commutation_rule} for $a \in K$.
  Elements in $K[s; \sigma, \delta]$ are called \emph{skew polynomials}.
  If $\delta = 0$, then $K[s; \sigma, 0]$ is denoted by $K[s; \sigma]$.
  When $\sigma$ is the identity map $\id$ and $\delta = 0$, the skew polynomial ring is nothing but the polynomial ring $K[s]$, which means $K[s] = K[s; \id]$.
  A typical nontrivial example of skew polynomial rings is the ring $\setC(t)[\partial; \id, ']$ of differential operators, where $\funcdoms{'}{\setC(t)}{\setC(t)}$ is the usual differentiation.
  Another example of skew polynomial rings the ring $\setC(t)[S; \tau]$ of shift operators, where $\funcdoms{\tau}{\setC(t)}{\setC(t)}$ is defined by $f(t) \mapsto f(t+1)$ for $f \in \setC(t)$.

  Applying the commutation rule~\eqref{eq:skew_polynomial_commutation_rule} iteratively, we can uniquely represent any skew polynomial $p \in K[s; \sigma, \delta] \setminus \set{0}$ as $p = a_0 + a_1 s + \dotsb + a_l s^l$, where $l \in \setN$ and $a_0, \dotsc, a_l \in K$ with $a_l \ne 0$.
  This $l$ is called the \emph{degree} of $p$ and is denoted by $\deg p$.
  We set $\deg 0 \defeq -\infty$.
  Since a skew polynomial ring $K[s; \sigma, \delta]$ is an Ore domain (see, e.g.,~\cite[Exercise~6F]{Goodearl2004}), it has the quotient skew field $K(s; \sigma, \delta)$, called the \emph{skew rational function field}.
  Its element $f \in K(s; \sigma, \delta)$, called a \emph{skew rational function}, has the degree defined by $\deg f \defeq \deg p - \deg q$ with $f = pq^{-1}$ and $p,q \in K[s; \sigma, \delta]$.
  Then $-\deg$ is a discrete valuation on $K(s; \sigma, \delta)$ with residue skew field $K$.
  Its completion is the \emph{skew inverse Laurent series field} $K\pprn[\big]{s^{-1}; \sigma, \delta}$, which is the skew field of formal power series over $K$ in the form of
  \begin{align}\label{eq:inverse_laurent_fps}
    f = \sum_{d=l}^\infty a_d s^{-d}
  \end{align}
  for some $l \in \setZ$ and $a_l, a_{l+1}, \dotsc \in K$~\cite[Section~2.3]{Cohn1995}.
  This skew field has the natural addition and a multiplication defined by~\eqref{eq:skew_polynomial_commutation_rule} and
  \begin{align}\label{eq:inverse-skew-polynomial-commutation-formula}
    s^{-1}a = \sum_{d=0}^\infty \delta_d(a)s^{-(d+1)}
  \end{align}
  for $a \in K$, where
  \begin{align}\label{eq:higher_derivation_for_skew_inverse_laurent}
    \delta_d \defeq \sigma^{-1}\prn[\big]{-\delta \sigma^{-1}}^{d}
  \end{align}
  for $d \in \setN$ (the multiplication of maps means the composition)~\cite{Paykan2017}.
  This is determined so that $ss^{-1}a = a$.

  One can define the \emph{order} $\ord p$ of a skew polynomial $p \in K[s; \sigma, \delta]$ similarly to the usual polynomials, i.e., $\ord p$ is the minimum $l \in \setN$ such that $p$ is represented as $p = a_l s^l + \dotsb + a_L s^L$ for some $L \in \setN$ and $a_l, \dotsc, a_L \in K$ with $a_l \ne 0$.
  Set $\ord 0 \defeq +\infty$ in the same way.
  However, if $a \in \mult{K}$ satisfies $\delta(a) \ne 0$, then $\ord s = 1$, $\ord a = 0$ and $\ord sa = \ord (\sigma(a)s + \delta(a)) = 0$, which violate~\ref{item:V1}.
  Thus $\ord$ cannot be extended to a discrete valuation on $K(s; \sigma, \delta)$.
  Nevertheless, in case of $\delta = 0$, the order satisfies~\ref{item:V1}--\ref{item:V3} and thus $K(s; \sigma) \defeq K(s; \sigma, 0)$ becomes a DVSF equipped with a discrete valuation $\ord f \defeq \ord p - \ord q$ for $f = pq^{-1} \in K(s; \sigma)$ with $p,q \in K[s; \sigma]$.
  This is because the change of variable $\phi \vcentcolon f(s) \mapsto f\prn{s^{-1}}$ provides an isomorphism between $K(s; \sigma)$ and $K(s; \sigma^{-1})$ and $\ord f = -\deg \phi(f)$ for $f \in K(s; \sigma)$.
  The completion of $K(s; \sigma)$ with respect to $\ord$ is the \emph{skew Laurent series field} $K\pprn{s; \sigma}$, whose elements are represented as formal Laurent series~\eqref{eq:laurent_fps}~\cite[Section~2.3]{Cohn1995}.
  The residue skew field of $K\pprn{s; \sigma}$ is clearly $K$.

  See~\cite[Chapter~2]{Cohn1995},~\cite[Section~7.3]{Cohn2003}, and~\cite[Chapter~2]{Goodearl2004} for details of skew polynomials,~\cite{Paykan2017} for skew inverse Laurent series fields.
\end{example}

\subsection{Split DVSFs}
A DVSF $F$ is said to be \emph{spilt} if it has a representative set $Q$ such that it is a subring of the valuation ring $R$ of $F$.
Similarly, a DVR $R$ is called \emph{split} if its quotient skew field $F$ (see \cref{rem:DVR}) is split.
Such $Q$ is called a \emph{coefficient skew subfield} or a \emph{Cohen skew subfield} of $F$ and of $R$.

Let $F$ be a split DVSF with coefficient skew subfield $Q$ and residue skew field $K$.
Since elements in $Q$ and $K$ correspond bijectively, $Q$ and $K$ must be isomorphic skew fields.
We thus call $Q$ ``the'' coefficient skew subfield of $F$.
This observation also implies that $F$ could be split only if $F$ is \emph{equicharacteristic}, i.e., $F$ and $K$ have the same characteristic.
For example, the field of $p$-adic numbers is not split as the characteristics of $\setQ$ and $\setF_p$ are different.
Indeed, if $F$ is a field, then $F$ is split if and only if $F$ is equicharacteristic~\cite[Theorem~9]{Cohen1946}.
Therefore, by \cref{prop:dvsf_as_power_series}, a complete split DVF $F$ is isomorphic to the Laurent series field $K\pprn{s}$ over the residue field $K$ of $F$.
This is a special case of the Cohen structure theorem for complete commutative Noetherian local rings~\cite{Cohen1946}.

The situation is much more complicated in the general noncommutative case.
No characterization of a DVSF to be split is yet known; Vidal~\cite{Vidal1981} gave an equicharacteristic but non-split example of a DVSF\@.
Nevertheless, as we have seen in \cref{sec:examples_of_DVSFs}, a skew inverse Laurent series field $K\pprn{s^{-1}; \sigma, \delta}$ and a skew Laurent series field $K\pprn{s; \sigma}$ over a skew field $K$ are split, where their coefficient skew subfields are both $K$.

Let $F$ be a complete split DVSF, $K$ the coefficient skew subfield and $\pi$ a uniformizer.
Then \cref{prop:dvsf_as_power_series} implies that the commutation rule between $\pi$ and each $a \in K$ completely determines the ring structure of $F$.
The element $\pi a$ can be uniquely expressed as~\eqref{eq:higher_sigma_derivation_pi}, where $\funcdoms{\delta_d}{K}{K}$ is some map for all $d \in \setN$.
The family of maps $\prn{\delta_d}_{d \in \setN}$ satisfies the following~\cite{Roux1986}:
\begin{enumerate}[{label={\upshape{(HD\arabic*)}}}]
  \item $\delta_d$ is additive for $d \in \setN$.\label{item:HD1}
  \item $\displaystyle \delta_d(ab) = \sum_{i=0}^d \delta_i(a) \Delta_i^d(b)$ for $d \in \setN$ and $a,b \in K$, where $\funcdoms{\Delta_i^d}{K}{K}$ is defined by
        \begin{align}
          \Delta_i^d \defeq \sum_{\substack{j_0, \dotsc, j_i \in \setN \\ j_0 + \dotsb + j_i = d-i}} \delta_{j_0} \dotsm \delta_{j_i}
        \end{align}
        for $d \in \setN$ and $0 \in \intset{0, d}$.\label{item:HD2}
  \item $\delta_0$ is an automorphism on $K$.\label{item:HD3}
\end{enumerate}

In fact,~\ref{item:HD1} and~\ref{item:HD2} are derived from the distributive law $\pi(a+b) = \pi a + \pi b$ and the associative law $\pi(ab) = (\pi a)b$, respectively~\cite{Elliger1967,Smits1968}.
From~\ref{item:HD1},~\ref{item:HD2} for $d = 0$, and $\delta_0(1) = 1$ by $\pi1 = 1\pi$, the leading map $\delta_0$ must be a homomorphism on $K$.
It further must be surjective by~\ref{item:DVR1}, which implies~\ref{item:HD3}.

Generally, a sequence $\prn{\delta_d}_{d \in \setN}$ of maps on a skew field $K$ is called a \emph{higher $\sigma$-derivation}~\cite{Elliger1967,Smits1968} of $K$ (with $\sigma \defeq \delta_0$) if it satisfies~\ref{item:HD1}--\ref{item:HD3}.
For a higher $\sigma$-derivation $\prn{\delta_d}_{d \in \setN}$, we denote by $K\ssqbr{s; (\delta_d)}$ the ring of formal power series over $K$ in indeterminate $s$, whose every element $f$ is uniquely expressed as~\eqref{eq:fps}.
The addition on $K\ssqbr{s; (\delta_d)}$ is naturally defined and the multiplication is induced from
\begin{align}
  sa = \sum_{d=0}^\infty \delta_d(a) s^{d+1}
\end{align}
for $a \in K$.
This ring is an Ore domain and thus has a quotient skew field $K\pprn{s; (\delta_d)}$.
As the usual formal power series ring, each $f \in K\pprn{s; (\delta_d)}$ is represented as a formal Laurent series
\begin{align}
  f = \sum_{d=l}^\infty a_d s^d
\end{align}
with $a_d \in K$ for every $d \in \setZ$.
Defining the \emph{order} of $f \in K\pprn{s; (\delta_d)}$ as the minimum $l \in \setN$ with $a_l \ne 0$, the skew field $K\pprn{s; (\delta_d)}$ becomes a complete split DVSF with respect to the order~\cite{Roux1986}; its valuation ring is $K\ssqbr{s; (\delta_d)}$, its (one choice of a) uniformizer is $s$, and its coefficient skew subfield is $K$.
Conversely, as seen above, we have:

\begin{proposition}[{\cite[Proposition~1.6 in p.~292]{Roux1986}}]
  Let $F$ be a complete split DVSF with coefficient skew subfield $K$.
  Then $F$ is isomorphic to $K\pprn{s; (\delta_d)}$, where $\prn{\delta_d}_{d \in \setN}$ is the higher $\delta_0$-derivation of $K$ determined by~\eqref{eq:higher_sigma_derivation_pi}.
\end{proposition}

\begin{corollary}
  Let $R$ be a complete split DVR with coefficient skew subfield $K$.
  Then $R$ is isomorphic to , where $\prn{\delta_d}_{d \in \setN}$ is the higher $\delta_0$-derivation of $K$ determined by~\eqref{eq:higher_sigma_derivation_pi}.
\end{corollary}

Note that since any split DVSF $F$ and DVR $R$ are a skew subfield and a subring of a complete split DVSF and DVR (see \cref{thm:completion}, $F$ and $R$ are isomorphic to a skew subfield of $K\pprn{s; (\delta_d)}$ and a subring of $K\pprn{s; (\delta_d)}$, respectively.

\begin{example}\label{ex:higher-derivations}
  We give some examples of higher $\sigma$-derivations and corresponding complete split DVSFs.
  Let $K$ be a skew field and $\sigma$ an automorphism on $K$.
  Then $(\sigma, 0, 0, \dotsc)$ is a higher $\sigma$-derivation and $K\pprn{s; (\sigma, 0, 0, \dotsc)} = K\pprn{s; \sigma}$.
  In particular, the case when $K$ is a field and $\sigma = \id$ corresponds to the representation of complete equicharacteristic
   described above.
  More generally, let $\delta$ be a \emph{right $\sigma$-derivation}, i.e., an additive map satisfying $\delta(ab) = \delta(a)\sigma(b) + a\delta(b)$ for $a,b \in K$.
  Then $\prn{\sigma, \sigma\delta, \sigma\delta^2, \dotsc}$ is a higher $\sigma$-derivation~\cite[Section~2.1]{Cohn1977}.
  If $\delta$ is a left $\sigma$-derivation instead of the right one, $-\sigma^{-1}\delta$ is a right $\sigma^{-1}$-derivation, and hence $\prn{\delta_d}_{d \in \setN}$ defined by~\eqref{eq:higher_derivation_for_skew_inverse_laurent} is a higher $\sigma^{-1}$-derivation; this is consistent with the fact that $K\pprn{s^{-1}; \sigma, \delta}$ is isomorphic to $K\pprn{t; (\delta_d)}$.
  Another type of a higher $\sigma$-derivation is given in~\cite{Brungs1984}.
  Dumas~\cite{Dumas1992} provides a survey for higher $\sigma$-derivations.
\end{example}

The following lemma provides a relation between coefficients in the $\pi$-adic expansions of $a \in R$ and $\pi a$.

\begin{lemma}
  Let $R$ be a split DVR with coefficient skew subfield $K$ and uniformizer $\pi$, and $(\delta_d)$ the higher $\delta_0$-derivation such that $R$ is isomorphic to $K\ssqbr{s; (\delta_d)}$
  For $a = \sum_{d=0}^\infty a_d \pi^d \in R$ with $a_0, a_1, \dotsc \in K$, the coefficient $b_d$ of $\pi^d$ in the $\pi$-adic expansion of $\pi a$ satisfies
  \begin{align}\label{eq:coefficient-in-pi-a}
    b_d = \begin{cases}
      \displaystyle \sum_{k=0}^{d-1} \delta_k(a_{d-k-1}) & (d \ge 1), \\
      0 & (d = 0).
    \end{cases}
  \end{align}
\end{lemma}

\begin{proof}
  Using~\eqref{eq:higher_sigma_derivation_pi}, we can rewrite $\pi a$ as
  \begin{align}
    \pi a
    = \sum_{d=0}^\infty \pi a_d \pi^d
    = \sum_{d=0}^\infty \prn{\sum_{k=0}^\infty \delta_k(a_d) \pi^{k+1}} \pi^{d}
    = \sum_{d=1}^\infty \prn{\sum_{k=0}^{d-1} \delta_k(a_{d-k-1})} \pi^d
  \end{align}
  as required.
\end{proof}

Let $F$ be a split DVSF with coefficient skew subfield $K$ and associated higher $\delta_0$-derivative $(\delta_d)_{d \in \setN}$.
As a computational model, we adopt the arithmetic model on $K$ and assume that one can compute $\delta_d(a)$ for every $d \in \setN$ and $a \in K$ in constant time.
In this model, if we know the leading $M+1$ coefficients $a_0, \dotsc, a_M$ in the $\pi$-adic expansion of $a \in K$, we can compute those of $\pi a$ in $\Order\prn{M^2}$-time by~\eqref{eq:coefficient-in-pi-a}.

\section{Preliminaries on Matrices}\label{sec:preliminaries-on-matrices}
For a ring $R$ and $n,n' \in \setN$, we denote the ring of $n \times n'$ matrices over $R$ by $R^{n \times n'}$.
We also denote by $Q^{n \times n'}$ the set of all $n \times n'$ matrices over a subset $Q$ of $R$.
A square matrix $A \in R^{n \times n}$ is said to be \emph{invertible} if there (uniquely) exists an $n \times n$ matrix over $R$, denoted by $A^{-1}$, such that $AA^{-1} = A^{-1}A = I_n$, where $I_n$ is the identity matrix of order $n$.
When $R$ can be extended to a skew field $F$, we call $A$ \emph{nonsingular} if $A$ is invertible over $F$ and \emph{singular} if not; the nonsingularity does not depend on the choice of $F$.
We denote by $\GL_n(R)$ the group of $n \times n$ invertible matrices over $R$, i.e., $\GL_n(R) \defeq \mult{\prn{R^{n \times n}}}$.

For $a \in \mult{R}$ and $\alpha = \prn{\alpha_i}_{i \in \intset{n}} \in \setZ^n$, we define $D(a^\alpha) \defeq \diag\prn[\big]{a^{\alpha_i}}_{i \in \intset{n}}$, where $\diag$ denotes the diagonal matrix.
For an additive map $\funcdoms{\phi}{R}{R}$ and $A \in R^{n \times n'}$, let $\phi(A)$ denote the $n \times n'$ matrix over $R$ obtained by applying $\phi$ to each entry in $A$.

Let $A \in R^{n \times n'}$ be a matrix.
For $I \subseteq \intset{n}$ and $J \subseteq \intset{n'}$, we denote by $A[I, J]$ the submatrix of $A$ consisting of rows $I$ and columns $J$.
When $I = \intset{n}$, we simply write $A[J] \defeq A[\intset{n}, J]$.

\subsection{Matrices over Skew Fields}\label{sec:matrices-over-skew-fields}
Let $F$ be a skew field.
A right (left) $F$-module is especially called a \emph{right} (\emph{left}) \emph{$F$-vector space}.
The \emph{dimension} of a right (left) $F$-vector space $V$ is defined as the rank of $V$ as a module, that is, the cardinality of any basis of $V$.
The usual facts from linear algebra on independent sets and generating sets in vector spaces are valid even on skew fields~\cite{Lam1999}.

The \emph{rank} $\rank A$ of a matrix $A \in F^{n \times n'}$ is the dimension of the right $F$-vector space spanned by the column vectors of $A$, and is equal to the dimension of the left $F$-vector space spanned by the row vectors of $A$.
The rank is invariant under (right and left) multiplication of nonsingular matrices.
It is observed that a square matrix $A \in F^{n \times n}$ is nonsingular if and only if $\rank A = n$.
The rank of $A \in F^{n \times n'}$ is equal to the minimum $r \in \setN$ such that there exists a decomposition $A = BC$ by some $B \in F^{n \times r}$ and $C \in F^{r \times n'}$~\cite{Cohn1985}.
Here we give another characterization of the rank, which is well-known on the commutative case.
\begin{proposition}\label{prop:max_nonsingular_submatrix}
  The rank of a matrix $A \in F^{n \times n'}$ over a skew field $F$ is equal to the maximum $r \in \setN$ such that $A$ has a nonsingular $r \times r$ submatrix.
  In addition, $A$ has a nonsingular $k \times k$ submatrix for all $k \in \intset{0, r}$.
\end{proposition}

\begin{proof}
  We first show the latter part.
  For $k \in \intset{0, \rank A}$, we can take a column subset $J \subseteq \intset{n'}$ of cardinality $k$ such that the column vectors of $A[J]$ are linearly independent.
  Since $\rank A[J] = k$, there must be $I \subseteq \intset{n}$ of cardinality $k$ such that the row vectors of $A[I, J]$ is linearly independent.
  Then $A[I, J]$ is a $k \times k$ nonsingular submatrix of $A$ due to $\rank A[I, J] = k$.

  The former part is shown as follows.
  Let $r \in \setN$ be the maximum size of a nonsingular submatrix of $A$.
  It holds $\rank A \le r$ by the latter part of the claim.
  To show $\rank A \ge r$, take an $r \times r$ nonsingular submatrix $A[I, J]$ of $A$.
  Since $\rank A[I, J] = r$, the set of column vectors of $A$ indexed by $J$ is linearly independent.
  Thus we have $\rank A \ge r$.
\end{proof}

We next define the \emph{Dieudonné determinant} for nonsingular matrices over a skew field.
To describe this, we introduce the \emph{Bruhat decomposition} as follows.
A lower (upper) unitriangular matrix is a lower (resp.\ upper) triangular matrix whose diagonal entries are 1.

\begin{proposition}[{Bruhat decomposition~\cite[Theorem~9.2.2]{Cohn2003}}]\label{prop:bruhat}
  A square matrix $A \in F^{n \times n}$ over a skew field $F$ can be decomposed as $A = LDPU$, where $L$ is lower unitriangular, $D$ is diagonal, $P$ is a permutation matrix, and $U$ is upper unitriangular.
  If $A$ is nonsingular, this decomposition is unique.
\end{proposition}

Let $\abel{\mult{F}} \defeq \mult{F} \extends  \comm{\mult{F}}$ denote the abelianization of $\mult{F}$, where $\comm{F^\times} \defeq \agbr{\set{aba^{-1}b^{-1}}[a,b \in F^\times]}$ is the commutator subgroup of $F^\times$.
The \emph{Dieudonné determinant} $\Det A$ of $A \in \GL_n(F)$, which is decomposed as $A = LDPU$ by \cref{prop:bruhat}, is an element of $\abel{\mult{F}}$ defined by
\begin{align}
  \Det A \defeq \sgn(P) e_1 e_2 \cdots e_n \bmod \comm{\mult{F}},
\end{align}
where $\sgn(P) \in \set{+1, -1}$ is the sign of the permutation $P$ and $e_1, \ldots, e_n \in \mult{F}$ are the diagonal entries of $D$~\cite{Dieudonne1943}.
%For a singular matrix $A \in F^{n \times n}$, we define $\Det A$ as $0$, which is an element that annihilates any element of $\mult{\abel{F}}$.
In case where $F$ is commutative, the Dieudonné determinant coincides with the usual determinant.

An \emph{elementary matrix} over $F$ is a unitriangular matrix $E_n(i, j; e) \in \GL_n(F)$ whose the $(i,j)$th entry $(i \ne j)$ is $e \in F$ and other off-diagonal entries are 0.
An \emph{elementary operation} on $A \in F^{n \times m}$ is the (left or right) multiplication of $A$ by an elementary matrix, which corresponds to adding a left (right) multiple of a row (resp.\ column) to another row (resp.\ column) of $A$.
Denote by $\E_n(F)$ the subgroup of $\GL_n(F)$ generated by elementary matrices.
If $F$ is a field, $\E_n(F)$ is nothing but the special linear group $\SL_n(F) \defeq \set{A \in \GL_n(F)}[\det A = 1]$~\cite[Theorem~3.5.1]{Cohn2003}.
This can be extended to the Dieudonné determinant as follows:

\begin{theorem}[{\cite[Theorem~9.2.6]{Cohn2003}}]
  For a skew field $F$ and $n \in \setN$, the Dieudonné determinant gives rise to an exact sequence of groups
  \begin{align}
    1 \longrightarrow \E_n(F) \longrightarrow \GL_n(F) \overset{\Det}{\longrightarrow} \abel{\mult{F}} \longrightarrow 1.
  \end{align}
\end{theorem}

Namely, $\funcdoms{\Det}{\GL_n(F)}{\abel{\mult{F}}}$ is a surjective map satisfying
\begin{enumerate}[{label={\upshape{(D\arabic*)}}}]
  \item $\Det AB = \Det A \Det B$ for $A, B \in \GL_n(F)$,\label{item:D1}
  \item $\Det A = 1$ for $A \in \E_n(F)$,\label{item:D2}
\end{enumerate}
where the inverse of~\ref{item:D2} also holds, i.e., $\E_n(F) = \set{A \in \GL_n(F)}[\Det A = 1]$.
It further follows immediately from the definition of $\Det$ that
\begin{enumerate}[{label={\upshape{(D3)}}}]
  \item $\displaystyle \Det \diag \prn{e_1, \dotsc, e_n} = \prod_{i=1}^n e_i \bmod \comm{\mult{F}}$ for $e_1, \dotsc, e_n \in \mult{F}$,\label{item:D3}
\end{enumerate}
where $\diag\prn{e_1, \dotsc, e_n}$ is the diagonal matrix with diagonal entries $e_1, \dotsc, e_n$.
Indeed, $\Det$ is the unique map satisfying~\ref{item:D1}--\ref{item:D3} since unitriangular matrices are in $\E_n(F)$ and any permutation matrix $P$ can be brought into $\diag(\sgn(P), 1, \dotsc, 1)$ by elementary operations.

\subsection{Matrices over Valuation Skew Fields}\label{sec:matrices-over-valuation-skew-fields}
Let $F$ be a valuation skew field with valuation $v$.
For any $A \in \GL_n(F)$, we denote by $\vdet(A)$ the valuation of any representative of $\Det A$; this is well-defined because all commutators of $\mult{F}$ have valuation $0$.
We also define $\vdet(A) \defeq +\infty$ for singular $A \in F^{n \times n}$.
By~\ref{item:V1},~\ref{item:V3} and~\ref{item:D1}--\ref{item:D3}, it holds
\begin{enumerate}[{label={\upshape{(VD\arabic*)}}}]
  \item $\vdet(AB) = \vdet(A) + \vdet(B)$ for $A, B \in F^{n \times n}$,\label{item:VD1}
  \item $\vdet(A) = 0$ for $A \in \E_n(F)$,\label{item:VD2}
  \item $\displaystyle \vdet\prn{\diag\prn{d_1, \dotsc, d_n}} = \sum_{i=1}^n v(d_i)$ for $d_1, \dotsc, d_n \in F$.\label{item:VD3}
\end{enumerate}
By the Bruhat decomposition, $\funcdoms{\vdet}{F^{n \times n}}{\setR \cup \set{+\infty}}$ is the unique map satisfying~\ref{item:VD1}--\ref{item:VD3}, as Taelman~\cite{Taelman2006} observed for $\deg \Det$ of skew polynomials.

Let $\M(F)$ denote the set of all square matrices of finite order over $F$.
If we see $\vdet$ as a function on $\M(F)$, it satisfies the (real) \emph{matrix valuation} axioms.
To describe this, we shall define the \emph{determinantal sum} for two matrices $A, B \in F^{n \times n'}$ such that their columns are identical except for the first columns.
The \emph{determinantal sum} of $A$ and $B$ with respect to the first column is an $n \times n'$ matrix over $F$ whose first column is the sum of those of $A$ and $B$, and other columns are the same as $A$.
The determinantal sums with respect to other columns and rows are also defined.
We denote the determinantal sum of $A$ and $B$ (with respect to an appropriate column or row) by $A \detsum B$.

A (real) \emph{matrix valuation}~\cite{Hezavehi1982} on a skew field $F$ is a map $\funcdoms{V}{\M(F)}{\setR \cup \set{+\infty}}$ that satisfies
\begin{enumerate}[{label={\upshape{(MV\arabic*)}}}]
  \item $V\begin{pmatrix} A & O \\ O & B\end{pmatrix} = V(A) + V(B)$ for $A, B \in \M(F)$, where $O$ denotes the zero matrix of appropriate size,\label{item:MV1}
  \item $V(A \detsum B) \ge \min\set{V(A), V(B)}$ for $A, B \in \M(F)$ such that $A \detsum B$ is defined,\label{item:MV2}
  \item $V(1) = 0$,\label{item:MV3}
  \item $V(A) = +\infty$ for singular $A \in \M(F)$,\label{item:MV4}
  \item $V(A)$ is unchanged if a column or a row of $A$ is multiplied by $-1$.\label{item:MV5}
\end{enumerate}
These axioms derive extra useful formulas as follows.

\begin{proposition}[\cite{Hezavehi1982}]\label{prop:matrix_valuation}
  For a matrix valuation $V$ on a skew field $F$, the following hold:
  \begin{enumerate}
    \item $V(AB) = V(A) + V(B)$ for $A,B \in F^{n \times n}$.\label{item:matrix_valuation_hom}
    \item $V\begin{pmatrix}A & * \\ O & B\end{pmatrix} = V\begin{pmatrix}A & O \\ * & B\end{pmatrix} = V(A) + V(B)$ for $A, B \in \M(F)$, where $*$ denotes any matrix of appropriate size.
    \item The equality in~\textup{\ref{item:MV2}} holds whenever $V(A) \ne V(B)$.
  \end{enumerate}
\end{proposition}

By \cref{prop:matrix_valuation}~\ref{item:matrix_valuation_hom} and~\ref{item:MV2}--\ref{item:MV4}, a matrix valuation $V$ restricted to $F$ ($1 \times 1$ matrices) is exactly a valuation $v$ on $F$.
This can be extended to $\M(F)$ as $\vdet$, i.e., $V = \vdet$ holds.
In general, for any valuation $v$ of $F$, $\vdet$ is a matrix valuation on $F$~\cite{Hezavehi1982}; the correspondence between $v$ and $V$ is clearly bijective.
Therefore, a matrix valuation is nothing but a valuation of the Dieudonné determinant.
See also~\cite[Section~9.3]{Cohn1995}.

For a matrix $A \in F^{n \times n'}$ over a valuation skew field $F$ with valuation $v$, we define
\begin{align}\label{def:zeta_k}
  \zeta_k(A) \defeq \min \set{\vdet\prn{A[I, J]}}[I \subseteq \intset{n}, J \subseteq \intset{n'}, \card{I} = \card{J} = k]
\end{align}
for $k \in \intset{0, \min\set{n, n'}}$.
Note that $\zeta_0(A) = 0$, $\zeta_1(A)$ is equal to the minimum of the valuation of an entry in $A$, and $\zeta_n(A) = \vdet(A)$ for $A \in F^{n \times n}$.
In addition, $\zeta_k(A) \ne +\infty$ if and only if $k \le \rank A$ by \cref{prop:max_nonsingular_submatrix}.

\Cref{prop:representation_on_valuation_ring,prop:dvsf_as_power_series} are naturally extended to matrices over valuation skew fields and DVSFs as follows.

\begin{proposition}\label{prop:representation_on_valuation_ring_matrix}
  Let $F$ be a valuation skew field with valuation $v$, valuation ring $R$, and representative set $Q$.
  Then any $A \in R^{n \times n'}$ is uniquely expressed as $A = A_0 + \tilde{A}$, where $A_0 \in Q^{n \times n'}$ and $\tilde{A} \in {J(R)}^{n \times n'}$.
\end{proposition}

\begin{proposition}\label{prop:dvsf_as_power_series_matrix}
  Let $F$ be a DVSF with discrete valuation $v$ and let $\pi$ and $Q$ be a uniformizer and a representative set of $F$, respectively.
  \begin{enumerate}
    \item For every $A \in F^{n \times n'}$, there uniquely exists a sequence $\prn{A_d}_{d \in \setZ}$ of $n \times n'$ matrices over $Q$ such that $A_d = O$ for all but finitely many $d < 0$ and
    \begin{align}\label{eq:power_of_pi_matrix}
      A = \sum_{d \in \setZ} A_d \pi^d
    \end{align}
    in the $\pi$-adic topology.
    If $l \defeq \zeta_1(A) \in \setZ$, then $A_d = O$ for $d < l$ and $A_l \ne 0$.

    \item If $F$ is complete and $\prn{A_d}_{d \in \setZ}$ is a sequence of elements in $Q$ such that $A_d = O$ for all but finitely many $d < 0$, the power series~\eqref{eq:power_of_pi} converges to an $n \times n'$ matrix $A$ over $F$.
  \end{enumerate}
\end{proposition}

For a matrix $A$ over a DVR, the matrices $A_0$ in \cref{prop:representation_on_valuation_ring_matrix,prop:dvsf_as_power_series_matrix} are the same.

\subsection{Canonical Forms}\label{sec:canonical-forms}

Let $F$ be a valuation skew field with valuation ring $R$.
A matrix over $F$ is called \emph{proper} if its entries are in $R$.
A proper matrix $A \in F^{n \times n}$ is particularly called \emph{biproper} if it is nonsingular and its inverse is also proper, i.e., $A \in \GL_n(R)$.
The (right or left) multiplication by biproper matrices are called \emph{biproper transformations}.
We establish the \emph{Smith--McMillan form} of matrices over $F$, which is a canonical form under biproper transformations.
This is well-known for matrices over $\setC(s)$ as the \emph{Smith--McMillan form at infinity}~\cite{Murota2000,Verghese1981} in the context of control theory.

\begin{proposition}[{Smith--McMillan form}]\label{prop:smith_mcmillan}
  Let $F$ be a valuation skew field with valuation $v$ and valuation ring $R$.
  For $A \in F^{n \times n'}$ of rank $r$, there exist $S \in \GL_n(R)$, $T \in \GL_{n'}(R)$ and $d_1, \dotsc, d_r \in \mult{F}$ such that $v(d_1) \le \dotsb \le v(d_r)$ and
  \begin{align}\label{eq:smith_mcmillan}
    SAT = \begin{pmatrix}
      \diag\prn{d_1, \dotsc, d_r} & O \\
      O & O
    \end{pmatrix}.
  \end{align}
  In addition, the element $d_i$ for $i \in \intset{r}$ is unique up to multiplication by a unit of $R$ and its valuation satisfies
  \begin{align}\label{eq:smith_mcmillan_vd_i}
    v(d_i) = \zeta_i(A) - \zeta_{i-1}(A).
  \end{align}
\end{proposition}

\begin{proof}
  We first construct the desired diagonalization.
  Suppose that $A \ne O$ and $d_1 \in \mult{F}$ is an entry in $A$ such that $v(d_1) = \zeta_1(A)$.
  Multiplying permutation matrices to $A$ from left and right, we move $d_1$ to the top-left entry.
  Note that permutation matrices are clearly biproper.
  Then we eliminate the first column of $A$ other than the top entry using $d_1$.
  This can be achieved by multiplying an elementary matrix $E_n\prn[\big]{1, i; a{d_1}^{-1}}$ to $A$ from left for $i = 2, \dotsc, n$, where $a$ is the $(i, 1)$st entry of $A$.
  Since $a{d_1}^{-1} \in R$ by $v(d_1) \le v(a)$, this elementary matrix is biproper.
  We similarly eliminate the first row of $A$ other than the left entry.
  Now $A$ is in the form $\begin{psmallmatrix} d_1 & 0 \\ 0 & B\end{psmallmatrix}$ with $B \in F^{(n-1) \times (n'-1)}$.
  Iteratively applying the same operation for $B$ as long as $B \ne O$, we obtain the decomposition~\eqref{eq:smith_mcmillan}.
  Note that $\zeta_1(A) \le \zeta_1(B)$ by~\ref{item:V1} and~\ref{item:V2} and hence $v(d_1) \le \dotsb \le v(d_r)$.

  We next show the uniqueness part.
  Since units of $R$ has valuation $0$, the formula~\eqref{eq:smith_mcmillan_vd_i} implies the uniqueness of $v(d_1), \dotsc, v(d_r)$.
  Let $D$ be the diagonal matrix constructed above.
  By the ordering of $d_1, \dotsc, d_r$, it holds $v(d_i) = \zeta_i(D) - \zeta_{i-1}(D)$.
  Therefore, it suffices to show that $\zeta_k(A)$ is invariant throughout the above procedure for $k \in \intset{0, r}$.
  It is clear that $\zeta_k(A)$ does not change by row and column permutations.
  Consider multiplying an elementary matrix $E_n(i, j; a)$ to $A$ from left, where $i, j \in \intset{n}$ with $i \ne j$ and $a \in R$.
  This corresponds to the operation of adding the $i$th row multiplied by $a$ to the $j$th row.
  Put $A' \defeq E_n(i, j; e)A$ and consider a submatrix with rows $I \subseteq \intset{n}$ and columns $J \subseteq intset{n'}$ of cardinality $k$.
  If $j \notin I$, then $A'[I, J] = A[I, J]$.
  If $i,j \in I$, then $A[I, J] = EA[I, J]$ for some elementary matrix $E$ of order $k$, which means $\vdet\prn{A'[I, J]} = \vdet\prn{A[I, J]}$ by~\ref{item:VD1} and~\ref{item:VD2}.
  In the remaining case, i.e., $i \notin I \ni j$, we have
  \begin{align}
    A'[I, J] = A[I, J] \detsum (FA[I', J]),
  \end{align}
  where $I' \defeq (I \cup \set{i}) \setminus \set{j}$ and $C \in F^{n \times n}$ is the diagonal matrix having $a$ for the $i$th diagonal entry and $1$ for other diagonals.
  By~\ref{item:MV2}, it holds
  \begin{align}
    \vdet\prn{A'[I, J]}
    &\ge \min\set{\vdet\prn{A[I, J]}, \vdet\prn{CA[I', J]}}\label{eq:smith_mcmillan_proof_ineq}\\
    &= \min \set{\vdet\prn{A[I, J]}, \vdet\prn{A[I', J]} + v(a)}.
  \end{align}
  Since $a \in R$, we have $\vdet\prn{A'[I, J]} \ge \zeta_k(A)$.
  Suppose $\zeta_k(A) = \vdet\prn{A[I, J]}$.
  If $\zeta_k(A) > \vdet\prn{A[I', J]} + v(a)$, the equality of~\eqref{eq:smith_mcmillan_proof_ineq} is attained.
  If $\zeta_k(A) = \vdet\prn{A[I', J]} + v(a)$, then $\zeta_k(A) = \vdet\prn{A[I', J]}$ by $v(a) \ge 0$ and $\vdet\prn{A[I', J]} \ge \zeta_k(A)$.
  In addition, we have $\vdet\prn{A'[I', J]} = \vdet\prn{A[I', J]}$ from $j \notin I'$, which means $\vdet\prn{A'[I', J]} = \zeta_k(A)$.
  Hence we have $\zeta_k(A') = \zeta_k(A)$ in all cases.
  The proof of the right multiplication of elementary matrices is the same.
\end{proof}

Solving~\eqref{eq:smith_mcmillan_vd_i} for $\zeta_k(A)$, we have
\begin{align}\label{eq:zeta_k_as_sum_of_v}
  \zeta_k(A) = \sum_{i=1}^k v(d_i)
\end{align}
for $k \in \intset{0, \rank A}$.
It is worth mentioning that $v(d_i) \ge 0$ for any $A \in R^{n \times n'}$ and $i \in \intset{\rank A}$ since $v(d_1) = \zeta_1(A) \ge 0$.

If $A$ is a matrix over a DVSF $F$, diagonal entries of the Smith--McMillan form of $A$ can be taken as powers of a uniformizer of $F$ as follows.

\begin{proposition}[{Smith--McMillan form for DVSFs}]\label{prop:dvsf_smith_mcmillan}
  Let $F$ be a DVSF with valuation ring $R$ and uniformizer $\pi$.
  For $A \in F^{n \times n'}$ of rank $r$, there exist $S \in \GL_n(R)$, $T \in \GL_{n'}(R)$, and unique $\alpha = \prn{\alpha_i}_{i \in \intset{r}} \in \setZ^r$ such that $\alpha_1 \le \dotsb \le \alpha_r$ and
  \begin{align}\label{eq:dvsf_smith_mcmillan}
    SAT = \begin{pmatrix}
      D(\pi^\alpha) & O \\
      O & O
    \end{pmatrix}.
  \end{align}
  For $i \in \intset{r}$, the integer $\alpha_i$ is determined by
  \begin{align}\label{eq:dvsf_smith_mcmillan_vd_i}
    \alpha_i = \zeta_i(A) - \zeta_{i-1}(A).
  \end{align}
\end{proposition}

\begin{proof}
  Let $D = S'AT$ be the Smith--McMillan form of $A$ given in \cref{prop:smith_mcmillan}.
  For $i \in \intset{r}$, we define $\alpha_i$ as the valuation of the $i$th diagonal entry $d_i$ of $D$.
  Then~\eqref{eq:dvsf_smith_mcmillan_vd_i} follows from~\eqref{eq:smith_mcmillan_vd_i}.
  Define a biproper matrix
  \begin{align}
    W \defeq \begin{pmatrix}
      \diag\prn[\big]{\pi^{\alpha_1}{d_1}^{-1}, \dotsc, \pi^{\alpha_r}{d_r}^{-1}} & O \\
      O & I_{n-r}
    \end{pmatrix} \in \GL_n(R).
  \end{align}
  Then $WD = WS'AT = UAV$ with $S \defeq WS'$ is equal to the right hand side of~\eqref{eq:dvsf_smith_mcmillan}, as required.
\end{proof}

The equation~\eqref{eq:zeta_k_as_sum_of_v} is rewritten as
\begin{align}\label{eq:zeta-as-sum-of-alpha}
  \zeta_k(A) = \sum_{i=1}^k \alpha_i
\end{align}
for $k \in \intset{0, \rank A}$.
This equation plays an important role in \cref{sec:matrix_expansion}.

We present two propositions for matrices over $R$ which are obtained as corollaries of the Smith--McMillan form.
The first one claims that $\zeta_k(A)$ is nonnegative for any proper matrix $A \in R^{n \times n'}$.

\begin{proposition}\label{prop:nonnegativity_of_zeta_k}
  Let $R$ be the valuation ring of a valuation skew field.
  For $A \in R^{n \times n'}$ and $k \in \intset{0, \min\set{n, n'}}$, it holds $\zeta_k(A) \ge 0$.
\end{proposition}

\begin{proof}
  If $k > r$ with $r \defeq \rank A$, we have $\zeta_k(A) = +\infty > 0$.
  If $k \le r$, the claim holds from~\eqref{eq:zeta_k_as_sum_of_v} and $v(d_1), \dotsc, v(d_r) \ge 0$.
\end{proof}

The second proposition is a characterization of biproper matrices.

\begin{proposition}\label{prop:biproper_equivalence}
  Let $F$ be a valuation skew field with valuation ring $R$, residue skew field $K$, and representative set $Q$, and let $\funcdoms{\phi}{R}{K}$ be the natural homomorphism.
  Also, let $A \in R^{n \times n}$ be a square proper matrix and $A_0 \in Q^{n \times n}$ the matrix in \cref{prop:representation_on_valuation_ring_matrix} with respect to $A$.
  Then the following are equivalent:
  \begin{enumerate}
    \item $A$ is biproper.\label{prop:biproper_equivalence_1}
    \item $\vdet(A) = 0$.\label{prop:biproper_equivalence_2}
    \item $\phi(A_0)$ is nonsingular.\label{prop:biproper_equivalence_3}
  \end{enumerate}
\end{proposition}

\begin{proof}
  Let $SAT = D \defeq \diag(d_1, \dotsc, d_n)$ be the Smith--McMillan form of $A$.
  Since $S$ and $T$ are biproper, $A$ is biproper if and only if so is $D$.
  This is equivalent to $v(d_i) = 0$ for all $i \in \intset{n}$, where $v$ is the valuation of $F$.
  Since $v(d_i)$ is nonnegative for $i \in \intset{n}$, this condition is further equivalent to $\zeta(A) = \sum_{i=1}^n v(d_i) = 0$, where the first equality is from~\eqref{eq:zeta_k_as_sum_of_v}.
  Thus~\ref{prop:biproper_equivalence_1} and~\ref{prop:biproper_equivalence_2} are equivalent.

  We next consider~\ref{prop:biproper_equivalence_3}.
  Let $D_0 \in Q^{n \times n}$ be the matrix obtained from $D$ by \cref{prop:representation_on_valuation_ring_matrix}.
  By the above argument, $A$ is biproper if and only if $v(d_i) = 0$ for every $i \in \intset{n}$.
  This is equivalent to the nonsingularity of $\phi(D)$ because for $i \in \intset{n}$, the $i$th diagonal of $\phi(D)$ is nonzero if and only if $v(d_i) = 0$.
  Applying $\phi$ to $D = SAT$ and $A = S^{-1}DT^{-1}$, we obtain $\phi(D) = \phi(S)\phi(A)\phi(T)$ and $\phi(A) = \phi\prn{S^{-1}}\phi(D)\phi\prn{T^{-1}}$.
  These imply $\rank \phi(D) = \rank \phi(A)$.
  In addition, it holds $\phi(A) = \phi(A_0)$ and $\phi(D) = \phi(D_0)$ from $A - A_0, D - D_0 \in {J(R)}^{n \times n}$.
  Thus all the statements in \cref{prop:biproper_equivalence} are equivalent.
\end{proof}

Finally, we introduce the \emph{Jacobson normal form} for matrices over PIDs.
As stated in \cref{sec:valuation-skew-fields}, any DVR is a PID.\@
For a commutative PID $R$, the \emph{Smith normal form} is a celebrated canonical form of matrices over $R$ under transformations by $\GL_n(R)$.
The \emph{Jacobson normal form}~\cite{Jacobson1943} is its generalization to general noncommutative PIDs.
It can also be seen as a generalization of the Smith--McMillan form over DVRs.
Recall from~\cite{Cohn2003,Jacobson1943} that a nonzero element $c$ of a domain $R$ is said to be \emph{invariant} if $cR = Rc$ and $a \in R \setminus \set{0}$ is called a \emph{total divisor} of $b \in R \setminus \set{0}$ if there exists invariant $c \in R$ such that $bR \subseteq cR \subseteq aR$.

\begin{proposition}[{Jacobson normal form~\cite[Theorem~16 in Chapter~3]{Jacobson1943}; see~\cite[Theorem~7.2.1]{Cohn2003}}]\label{prop:jacobson}
  Let $A \in R^{n \times m}$ be a matrix of rank $r$ over a PID $R$\footnote{%
    As explained in \cref{sec:valuation-skew-fields}, any PID is an Ore domain, i.e., $R$ can be extended to a skew field $F$.
    Thus the rank of $A$ can be defined as that of a matrix over $F$.
  }.
  There exist $U \in \GL_n(R)$, $V \in \GL_m(R)$ and $e_1, \dotsc, e_r \in R \setminus \set{0}$ such that $e_i$ is a total divisor of $e_{i+1}$ for $i \in \intset{r-1}$ and
  \begin{align}
    UAV = \begin{pmatrix}
      \diag\prn{e_1, \dotsc, e_r} & O \\
      O & O
    \end{pmatrix}.
  \end{align}
\end{proposition}

We can also prove \cref{prop:dvsf_smith_mcmillan} by using \cref{prop:jacobson}.
Namely, the Smith--McMillan form over a DVR $R$ can also be seen as a variant of the Jacobson normal form over $R$ regarded as a PID\@.

\section{Combinatorial Aspects of Valuations and Matrices}\label{sec:combinatorial-aspects-of-valuations-and-matrices}
\subsection{Bipartite Matchings and Matrix Ranks}\label{sec:bipartite-matching}
Let $G = (V, E)$ be an undirected graph.
A \emph{matching} of $G$ is an edge subset $M \subseteq E$ such that no two distinct edges in $M$ share the same end.
A matching $M$ is said to be \emph{perfect} if every vertex of $G$ is covered by some edge in $G$.
The \emph{matching problem} on $G$ is to find a maximum-cardinality matching of $M$.
An undirected graph is called \emph{bipartite} if there exists a bipartition of vertices such that every edge is between different parts in the bipartition.
The \emph{Kőnig--Egerváry theorem} is a min-max theorem for the bipartite matching problem.
To describe it, we shall define a \emph{vertex cover} of a graph $G$ as a vertex subset that includes at least one end of every edge of $G$.

\begin{theorem}[{Kőnig--Egerváry theorem~\cite{Konig1931}; see~\cite[Theorem~16.2]{Schrijver2003}}]\label{thm:konig}
  The maximum size of a matching in a bipartite graph $G$ is equal to the minimum size of a vertex cover of $G$.
\end{theorem}

Bipartite matching and ranks of matrices are closely related.
Let $A = \prn{A_{i,j}} \in F^{n \times n'}$ be a matrix over a skew field $F$.
We associate to $A$ a bipartite graph $G(A)$ with vertex set $\intset{n} \sqcup \intset{n'}$ and edge set
\begin{align}
  E(A) \defeq \set*{(i, j)}[$i \in \intset{n}$, $j \in \intset{n'}$, $A_{i,j} \ne 0$].
\end{align}
The \emph{term-rank} of $A$, introduced by Ore~\cite{Ore1955}, is the maximum size of a matching in $G(A)$.
We denote the term-rank of $A$ by $\trank A$.
By \cref{thm:konig}, $\trank A$ is equal to the optimal value of the following problem:
\Minimize{
  n + n' - s - t
}{
  \text{$A$ has a zero block of size $s \times t$,} \\
  s \in \intset{0, n}, t \in \intset{0, n'}.
}

Indeed, $\trank A$ serves as a combinatorial upper bound on $\rank A$ as we well see below.
When $F$ is a field, it immediately follows from the definition of the determinant.

\begin{proposition}\label{prop:upper_bound_on_rank}
  Let $A \in F^{n \times n'}$ be a matrix over a skew field $F$.
  Then it holds $\rank A \le \trank A$.
\end{proposition}

\begin{proof}
  Permuting rows and columns of $A$, we assume that $A$ is in form of $A = \begin{psmallmatrix} X & Y\\ Z & O \end{psmallmatrix}$, where $O$ is the zero matrix of size $s \times t$ and $\trank A = n + n' - s - t$.
  Then we can decompose $A$ as
  \begin{align}\label{eq:trank_proof}
    A = \begin{pmatrix}
      X & Y \\
      Z & O
    \end{pmatrix}
    =
    \begin{pmatrix}
      X & I_{n'-t} \\
      Z & O
    \end{pmatrix}
    \begin{pmatrix}
      I_{n-s} & O \\
      O & Y
    \end{pmatrix}.
  \end{align}
  The size of matrices in the right hand side of~\eqref{eq:trank_proof} is $n \times p$ and $p \times n'$ with $p \defeq \trank A$.
  Hence $\rank A \le \trank A$ by the characterization of $\rank A$ (see \cref{sec:matrices-over-skew-fields}).
\end{proof}

\subsection{Weighted Bipartite Matchings and Valuations of Determinants}\label{sec:weighted-bipartite-matching}

We next consider the weighted bipartite matching problem, which is also called the \emph{assignment problem}.
Let $G = (U \cup V, E)$ be a bipartite graph with $n \defeq \card{U} = \card{V}$ and $\funcdoms{w}{E}{\setR}$ an edge weight.
The \emph{minimum-weight perfect matching problem}, or simply the \emph{weighted matching problem}, on $G$ with respect to $w$ is defined as the problem of finding a perfect matching $M$ of $G$ having the minimum weight $w(M)$ among all perfect matchings of $G$.
The dual problem of the linear programming (LP) relaxation of the weighted bipartite matching problem on $G$ is the following (see~\cite[Theorem~17.5]{Schrijver2003}):
\Maximize{
  \sum_{i \in U} p_i + \sum_{j \in V} q_j
}{
  p_i + q_j \le w(e) & (i \in U, j \in V, e = \set{i, j} \in E), \\
  p_i, q_j \in \setR & (i \in U, j \in V).
}
By the strong duality of linear programming, the optimal value of the dual problem is equal to the minimum-weight of a perfect matching in $G$.
In addition, if $w$ is integer-valued, then we can take optimal $(p, q)$ as integer vectors.

The following \emph{complementarity theorem} plays an important role in the combinatorial relaxation algorithm.
Let $G = (U \cup V, E)$ be a bipartite graph equipped with an edge weight $\funcdoms{w}{E}{\setR}$.
For a dual feasible solution $(p, q)$, we define a bipartite graph $G^\# = (U \cup V, E^\#)$ by
\begin{align}\label{def:tight-edges}
  E^\# \defeq \set*{e \in E}[$p_i + q_j = w(e)$ with $e = \set{i, j}, i \in U, j \in V$].
\end{align}
Namely, $G^\#$ is the subgraph of $G$ obtained by collecting only the ``tight'' edges.
Then the following holds from the complementarity theorem of linear programming.

\begin{proposition}[{complementarity theorem; see~\cite[Lemma~2.6]{Murota1995a}}]\label{prop:bipartite-matching-complementarity}
  Under the above setting, $(p, q)$ is optimal if and only if $G^\#$ has a perfect matching.
\end{proposition}

Analogously to the relation between the bipartite matching problem and the rank computation, solving the weighted bipartite matching problem corresponds to computing the valuation of the Dieudonné determinant.
Let $A = \prn{A_{i,j}} \in F^{n \times n}$ be a square matrix over a valuation skew field $F$ with valuation $v$.
Recall from \cref{sec:matrices-over-valuation-skew-fields} that $\vdet(A)$ denotes the valuation of the Dieudonné determinant of $A$.
For the bipartite graph $G(A)$ associated with $A$, we set an edge weight $\funcdoms{w}{E(A)}{\setR}$ as $w(e) \defeq v\prn{A_{i,j}}$ for $e = \set{i,j} \in E(A)$.
We denote by $\tvdet(A)$ the minimum-weight of a perfect matching in $G(A)$ with respect to the edge weight $w$.
If $G(A)$ has no perfect matching, put $\tvdet(A) \defeq +\infty$.
If $F$ is a field, then $\tvdet(A) \le \vdet(A)$ by the definition of the determinant and the axioms~\ref{item:V1},~\ref{item:V2} of valuations.
This inequality is indeed valid even for noncommutative matrices:

\begin{proposition}\label{prop:tvdet_le_vdet}
  Let $A \in F^{n \times n}$ be a square matrix over a valuation skew field $F$.
  Then it holds $\tvdet(A) \le \vdet(A)$.
\end{proposition}

\begin{proof}
  By \cref{prop:upper_bound_on_rank}, $\tvdet(A) = +\infty$ implies $\vdet(A) = +\infty$.
  Suppose $\tvdet(A) < +\infty$, i.e., $G(A)$ has a perfect matching.
  Let $(p, q)$ be a dual optimal solution of the maximum-weight perfect matching problem on $A$.
  We take diagonal matrices $P, Q \in \GL_n(F)$ such that the valuation of the $i$th and the $j$th diagonal entries of $P$ and $Q$ are $p_i$ and $q_j$, respectively, for every $i, j \in \intset{n}$\footnote{%
    By the existence of augmenting path algorithms for the weighted matching problem, we can assume that every component of $p$ and $q$ are integer combination of edge weights.
    Therefore, for every $i,j \in \intset{n}$, there must exist $a, b \in F$ such that $v(a) = p_i$ and $v(b) = q_j$, where $v$ is the valuation on $F$.
    The matrices $P$ and $Q$ are obtained by arranging these elements in diagonals.
  }.
  Put $B \defeq P^{-1}AQ^{-1}$.
  Then the valuation of the $(i, j)$th entry of $B$ is $w(\set{i, j}) - p_i + q_j \ge 0$ for all $\set{i, j} \in E(A)$.
  Thus $B$ is a matrix over the valuation ring of $F$, and hence $\vdet(B) \ge 0$ by \cref{prop:nonnegativity_of_zeta_k}.
  By $\vdet(B) = \vdet(A) - \tvdet(A)$, the desired inequality is proved.
\end{proof}

\subsection{Valuated Matroids}

A \emph{valuated matroid}, introduced by Dress--Wenzel~\cite{Dress1990,Dress1992}, on a finite set $E$ is a function $\funcdoms{\omega}{2^E}{\setR \cup \set{-\infty}}$ satisfying the following condition:

\begin{enumerate}[label={\upshape{(VM)}}]
  \item\label{item:VM}
    For any $j \in X \setminus Y$, there exists $j' \in Y \setminus X$ such that $\omega(X) + \omega(Y) \le \omega(X \cup \set{j'} \setminus \set{j}) + \omega(Y \cup \set{j} \setminus \set{j'})$.
\end{enumerate}

It is easily confirmed that the family $\set{X \subseteq E}[\omega(X) > -\infty]$ forms a base family of a matroid over $E$ (assuming the family is nonempty), which means that valuated matroids are a generalization of matroids.
In addition, valuated matroids can be maximized by a greedy algorithm.
Conversely, $\funcdoms{\omega}{2^E}{\setR \cup \set{-\infty}}$ is a valuated matroid if and only if $\omega + p$ is maximized by the greedy algorithm for any linear function $\funcdoms{p}{2^E}{\setR \cup \set{-\infty}}$~\cite{Dress1990}.
In this way, valuated matroids are recognized as a kind of ``concave function'' on $2^E \simeq \set{0, 1}^n$.

A typical example of valuated matroids arises from the valuation of determinants of matrices over a valuation field~\cite{Dress1990,Dress1992}.
Since the proof essentially relies on the \emph{Grassmann--Plücker identity}, which is an expansion formula of determinants, it cannot be directly applied to valuation skew fields.
Nevertheless, Hirai~\cite[Proposition~2.12]{Hirai2019} presented another proof which is valid for the degree of rational functions over skew fields.
This can be straightforwardly extended to general valuation skew fields as follows.

\begin{proposition}\label{prop:valuated-matroid-and-valuation-skew-field}
  Let $A \in F^{n \times n'}$ be a matrix over a valuation skew field $F$.
  The function $\funcdoms{\omega}{2^{\intset{n'}}}{\setR \cup \set{-\infty}}$ given by
  \begin{align}
    \omega(J) \defeq \begin{cases}
      -\vdet(A[X]) & (\card{J} = n), \\
      -\infty      & (\text{otherwise})
    \end{cases}
  \end{align}
  for $X \subseteq \intset{n'}$ is a valuated matroid on $\intset{n'}$.
\end{proposition}

\begin{proof}
  A local characterization~\cite[Theorem~5.2.25]{Murota2000} of valuated matroids claims that $\omega$ is a valuated matroid if and only if (i) $\set{X \subseteq \intset{n'}}[\omega(X) \ne -\infty]$ forms a base family of a matroid and (ii) $\omega$ satisfies~\ref{item:VM} for $X, Y \subseteq \intset{n'}$ with $\card{X \setminus Y} = \card{Y \setminus X} = 2$.
  The condition (i) holds since the linear independence of column vectors of $A$ defines a matroid.

  We show the condition (ii).
  Let $X, Y \subseteq E$ with $\omega(X), \omega(Y) \ne -\infty$ and $\card{X \setminus Y} = \card{Y \setminus X} = 2$.
  Put $A' \defeq A[X \cup Y]$.
  By a column permutation, we arrange columns of $X \cap Y$ in the left $n - 2$ columns of $A'$ without changing $\omega$.
  In addition, by elementary row operations, we can assume without changing $\omega$ that $A'$ is in the form of $\begin{psmallmatrix} S & T \\ O & U \end{psmallmatrix}$, where $S$ is a nonsingular $(n-2) \times (n-2)$ matrix, $T$ is an $(n-2) \times 4$ matrix, and $U$ is a $2 \times 4$ matrix.
  Assume that $X \setminus Y = \set{1, 2}$ and $Y \setminus X = \set{3, 4}$.
  For distinct $j, j' \in \set{1, 2, 3, 4}$, define $u_{j,j'}$ as the valuation of the Dieudonné determinant of the $2 \times 2$ submatrix of $U$ with column set $\set{j, j'}$.
  Then $\omega((X \cap Y) \cup \set{j, j'}) = -\vdet(S) - u_{j,j'}$ for any distinct $j, j' \in \set{1, 2, 3, 4}$.
  Hence~\ref{item:VM} is equivalent to the following:
  \begin{enumerate}[label={\upshape{(4PT)}}]
    \item\label{item:4PT}
      The minimum value of $u_{1,2} + u_{3,4}$, $u_{1,3} + u_{2,4}$, $u_{1,4} + u_{2,3}$ is attained at least twice.
  \end{enumerate}

  Now $u_{1,2} \ne -\infty$ by $\omega(X) \ne -\infty$.
  By a column permutation, we assume that the $(1,1)$st entry of $U$ is nonzero.
  In addition, we make the $(2,1)$st entry of $U$ zero using an elementary row operation.
  If the $(2,3)$rd entry is nonzero, make the $(1,3)$rd entry zero in the same way.
  Then $U$ is in form of either
  \begin{align}
    U = \begin{pmatrix}
      a & c & d & e \\
      0 & b & 0 & f
    \end{pmatrix}
    \quad\text{or}\quad
    \begin{pmatrix}
      a & c & 0 & e \\
      0 & b & d & f
    \end{pmatrix}.
  \end{align}

  In the left case, $u_{1,2} + u_{3,4} = u_{1,4} + u_{2,3} = v(a) + v(b) + v(d) + v(f)$ and $u_{1,3}+u_{2,4} = +\infty$, where $v$ is the valuation of $F$.
  In the right case, $u_{1,2} + u_{3,4} = v(a) + v(b) + v(d) + v(e)$, $u_{1,4} + u_{2,3} = v(a) + v(f) + v(c) + v(d)$ and $u_{1,3} + u_{2,4} = v(a) + v(d) + \vdet \begin{psmallmatrix} c & e \\ b & f \end{psmallmatrix} \ge v(a) + v(d) + \max \set{v(c) + v(f), v(b) + v(e)}$ by \cref{prop:matrix_valuation}~(3).
  The equality is attained if $v(c) + v(f) \ne v(b) + v(e)$.
  Hence~\ref{item:4PT} is satisfied for all cases.
\end{proof}

Let $R$ and $C$ be finite sets.
Murota~\cite{Murota1995c} introduced a \emph{valuated bimatroid} over $(R, C)$ as a function $\funcdoms{w}{2^R \times 2^C}{\setR \cup \set{-\infty}}$ satisfying the following conditions:

\begin{enumerate}[label={\upshape{(VBM\arabic*)}},itemindent=*]
  \item\label{item:VBM1}
    For any $i' \in I' \setminus I$, at least one of the following holds:
    \begin{enumerate}[label={\upshape{(\alph*1)}}]
      \item $\exists j' \in J' \setminus J$: $w(I, J) + w(I', J') \le w(I \cup \set{i'}, J \cup \set{j'}) + w(I' \setminus \set{i'}, J' \setminus \set{j'})$,
      \item $\exists i \in I \setminus I'$: $w(I, J) + w(I', J') \le w(I \cup \set{i'} \setminus \set{i}, J) + w(I \cup \set{i} \setminus \set{i'}, J')$.
    \end{enumerate}
  \item\label{item:VBM2}
    For any $j' \in J' \setminus J$, at least one of the following holds:
    \begin{enumerate}[label={\upshape{(\alph*2)}}]
      \item $\exists i \in I \setminus I'$: $w(I, J) + w(I', J') \le w(I \setminus \set{i}, J \setminus \set{j}) + w(I' \cup \set{i}, J' \cup \set{j})$,
      \item $\exists j' \in J' \setminus J$: $w(I, J) + w(I', J') \le w(I, J \cup \set{j'} \setminus \set{j}) + w(I', J \cup \set{j} \setminus \set{j'})$.
    \end{enumerate}
\end{enumerate}

The following is a noncommutative generalization of~\cite[Remark~2]{Murota1995c}.

\begin{proposition}\label{prop:deg_det_is_valuated_bimatroid}
  Let $A \in F^{n \times n'}$ be a matrix over a valuation skew field $F$.
  Define $\funcdoms{w}{2^{\intset{n}} \times 2^{\intset{n'}}}{\setR \cup \set{-\infty}}$ as
  \begin{align}\label{eq:valuated-bimatroid-from-val-det}
    w(I, J) \defeq \begin{cases}
      -\vdet(A[I, J]) & (\card{I} = \card{J}), \\
      -\infty         & (\text{otherwise})
    \end{cases}
  \end{align}
  for $I \subseteq \intset{n}$ and $J \subseteq \intset{n'}$.
  Then $w$ is a valuated bimatroid.
\end{proposition}
\begin{proof}
  To distinguish rows and columns of $A$, we identify the rows and columns of $A$ with distinct sets $R$ and $C$, respectively.
  Consider an $n \times (n + n')$ skew function matrix $B \defeq \begin{pmatrix} I_n & A \end{pmatrix}$ with row set $R$ and column set $E \defeq R \cup C$.
  Then there is a one-to-one correspondence between a submatrix of $A$ and a submatrix of $B$ with row set $R$ given by $2^R \times 2^C \ni (I, J) \mapsto (R, (R \setminus I) \cup J) \in 2^R \times 2^E$.
  In particular, if $\card{I} = \card{J} \eqdef k$, then $\card{R} = \card{(R \setminus I) \cup J}$ and
  \begin{align} \label{eq:valuated_matroid}
    \vdet(B[(R \setminus I) \cup J])
    = \vdet \begin{pmatrix}
      I_k & A[R \setminus I, J] \\
      O   & A[I, J]
    \end{pmatrix}
    = \vdet(A[I, J])
    = -w(I, J).
  \end{align}
  Define a map $\funcdoms{\omega}{E}{\setR \cup \set{-\infty}}$ by
  \begin{align}
    \omega(X) \defeq \begin{cases}
      -\vdet(B[X]) \; (= w(R \setminus X, X \cap C)) & (\card{X} = n), \\
      -\infty                                           & (\text{otherwise})
    \end{cases}
  \end{align}
  for $X \subseteq E$.
  Then $w$ satisfies~\ref{item:VBM1} and~\ref{item:VBM2} if and only if $\omega$ is a valuated matroid, which was already shown in \cref{prop:valuated-matroid-and-valuation-skew-field}.
\end{proof}

Let $w$ be a valuated bimatroid over $(R, C)$.
By a kind of greedy algorithm, one can obtain sequences $\varnothing = I_0 \subseteq I_1 \subseteq \dotsb I_{n^*} \subseteq R$ and $\varnothing = J_0 \subseteq J_1 \subseteq \dotsb J_{n^*} \subseteq C$ with $n^* \defeq \min\set{\card{R}, \card{C}}$ such that $(I_k, J_k)$ is a maximizer of the right-hand side in
\begin{align}\label{def:d_k}
  d_k \defeq \set{w(I, J)}[\card{I} = \card{J} = k]
\end{align}
for every $k \in \intset{0, n^*}$~\cite{Murota1995c}.
Therefore, from \cref{prop:deg_det_is_valuated_bimatroid}, any algorithm to compute valuations of the Dieudonné determinants can be applied to compute $\zeta_k(A)$ defined by~\eqref{def:zeta_k}.

\section{Combinatorial Relaxation Algorithm}\label{sec:combinatorial-relaxation-algorithm}
Let $F$ be a split DVSF with valuation $v$, uniformizer $\pi$, valuation ring $R$, coefficient skew subfield $K$, and associated higher $\delta_0$-derivations $(\delta_d)_{d \in \setN}$.
Let $A = \prn{A_{i,j}} \in F^{n \times n}$ be a square matrix given as the $\pi$-adic expansion
\begin{align}\label{eq:input_form}
  A = \sum_{d=0}^l A_d \pi^d,
\end{align}
where $l \in \setN$ and $A_0, \dotsc, A_l \in K^{n \times n}$.
Note that $A$ is a matrix over $R$.
This section describes the combinatorial relaxation algorithm for computing $\vdet(A)$.

\subsection{Truncating Higher-Valuation Terms}
By technical reasons, our algorithm requires an upper bound $M$ on $\vdet(A)$ (or $\vdet(A) = +\infty$).
Indeed, we can assume $l = \Order(M)$ by the following proposition:

\begin{proposition}\label{prop:truncation}
  Let $F$ be a DVSF with uniformizer $\pi$ and let $A = \sum_{d=0}^l A_d \pi^d \in F^{n \times n}$ be a matrix in form of~\eqref{eq:input_form}.
  For any $M \in \setN$ and $\tilde{A} \defeq \sum_{d=0}^M A_d \pi^d$, the following hold:
  \begin{enumerate}
    \item If $\vdet(A) \le M$, then $\vdet(A) = \vdet(\tilde{A})$.
    \item If $\vdet(A) > M$, then $\vdet(\tilde{A}) > M$.
  \end{enumerate}
\end{proposition}

\begin{proof}
  Let $v$ and $R$ be the valuation and the valuation ring of $F$, respectively.
  Recall $J(R) = \pi R = R \pi$ from~\ref{item:DVR1} and let $\funcdoms{\phi}{R}{R / {J(R)}^{M+1}}$ be the natural homomorphism.
  It is easily checked that $\phi(a) \ne 0$ if and only if $v(a) \le M$ and $\phi(a) = \phi(b) \ne 0$ implies $v(a) = v(b) \le M$ for $a, b \in R$.

  Let $P = \prn{P_{i,j}}, Q = \prn{Q_{i,j}} \in R^{n \times n}$ be any square matrices over $R$ with $\phi(P) = \phi(Q)$.
  Let $D$ and $E$ be the Smith--McMillan forms of $P$ and $Q$, respectively.
  We show $\phi(D) = \phi(E)$ by tracing the procedure to obtain the Smith--McMillan forms $D, E$ given in the proof of \cref{prop:smith_mcmillan}.
  First, we find a matrix entry having the minimum valuation of each $P$ and $Q$, and move it to the top-left.
  If the minimum valuation $\zeta_1(P)$ of an entry in $P$ is larger than $M$, then $\phi(P) = O$ and thus $\phi(Q) = O$ by $\phi(P) = \phi(Q)$.
  Thus $\phi(D) = \phi(E) = O$ in this case.
  Suppose $v(P_{i,j}) = \zeta_1(P) \le M$.
  By $\phi(P_{i,j}) = \phi(Q_{i,j}) \ne 0$, it holds $v(P_{i,j}) = v(Q_{i,j})$ and $\zeta_1(P) = \zeta_1(Q)$.
  Hence the top-left entries of $\phi(D)$ and $\phi(E)$ are the same.
  After moving the $(i, j)$th entries in $P$ and $Q$ to the top-left, we eliminate the first row and columns except for the top-left entries.
  Since $\phi$ is a homomorphism, $\phi(P)$ remains to be the same as $\phi(Q)$ after this elimination.
  Applying the above arguments to the bottom-right $(n-1) \times (n-1)$ submatrix recursively, we have $\phi(D) = \phi(E)$.

  Let $\diag\prn{d_1, \dotsc, d_n}$ and $\diag\prn[\big]{\tilde{d}_1, \dotsc, \tilde{d}_n}$ be the Smith--McMillan forms of $A$ and $\tilde{A}$, respectively.
  By $\phi(A) = \phi(\tilde{A})$ and the above arguments, the images of their Smith--McMillan forms by $\phi$ are the same, i.e., $\phi\prn{d_i} = \phi\prn[\big]{\tilde{d}_i}$ for $i \in \intset{n}$.

  Suppose that $\vdet(A) \le M$.
  From $\sum_{i=1}^n v(d_i) = \vdet(A) \le M$ and $v(d_i) \ge 0$ for $i \in \intset{n}$, it holds $v(d_i) \le M$ and thus $\phi\prn[\big]{\tilde{d}_i} = \phi(d_i) \ne 0$.
  This means $v(d_i) = v\prn[\big]{\tilde{d}_i}$ for $i \in \intset{n}$.
  Hence $\vdet(A) = \sum_{i=1}^n v(d_i) = \sum_{i=1}^n v\prn[\big]{\tilde{d}_i} = \vdet(\tilde{A})$

  Next, suppose that $\vdet(A) > M$.
  If $v(d_i) \le M$ for all $i \in \intset{n}$, then $v(d_i) = v(\tilde{d}_i)$ and $\vdet(\tilde{A}) = \vdet(A) > M$ in the same way as above.
  If $v(d_n) > M$, then $\phi\prn[\big]{\tilde{d}_n} = \phi(d_n) = 0$, which implies $\vdet(\tilde{A}) \ge v\prn[\big]{\tilde{d}_n} > M$.
\end{proof}

From \cref{prop:truncation}, we can compute $\vdet(A)$ by computing it for $\tilde{A} \defeq \sum_{d=0}^M A_d \pi^d$ instead of $A$.
Hence we can assume $l = \Order(M)$ by truncating higher-valuation terms in $A$.

\subsection{Faithful Algorithm}\label{sec:faithful-algorithm}
This section describes the combinatorial relaxation algorithm which is faithful to the original algorithm of Murota~\cite{Murota1995a}.
Recall from \cref{sec:weighted-bipartite-matching} that $A$ is associated with the bipartite graph $G(A)$ equipped with an integral edge weight and $\tvdet(A)$ denotes the minimum weight of a perfect matching in $G(A)$.
By \cref{prop:tvdet_le_vdet}, $\tvdet(A)$ serves as a lower bound on $\vdet(A)$.
We say that $A$ is \emph{upper-tight} if $\tvdet(A) = \vdet(A)$.
The combinatorial relaxation algorithm to compute $\vdet(A)$ is the following:

\paragraph{Faithful Combinatorial Relaxation Algorithm}
\begin{enumerate}[label={\upshape{\textbf{Phase~\arabic*a.}}},ref={Phase~\arabic*a}]
  \setcounter{enumi}{-1}
  \item\label{item:Phase0}
    Set $A^1 \gets A$ and $k \gets 1$.
  \item\label{item:Phase1}
    Compute $\tvdet\prn[\big]{A^k}$ by solving the minimum-weight perfect matching problem.
    If $\tvdet\prn[\big]{A^k} > M$, output $+\infty$ and halt.
  \item\label{item:Phase2}
    If $A$ is upper-tight, output $\tvdet\prn[\big]{A^k}$ and halt.
  \item\label{item:Phase3}
    Find $A^{k+1} \in F^{n \times n}$ such that $\vdet\prn[\big]{A^k} = \vdet\prn[\big]{A^{k+1}}$ and $\tvdet\prn[\big]{A^k} < \tvdet\prn[\big]{A^{k+1}}$.
    Set $k \gets k+1$ and go back to \ref{item:Phase1}.
\end{enumerate}

Since the input matrix $A$ is over $R$, each edge in $G(A)$ has a nonnegative weight, from which $\tvdet(A) \ge 0$ holds.
Therefore, the number of iterations is at most $\vdet(A) \le M$.
In the remaining of this section, we explain details of the upper-testing testing in \ref{item:Phase2} and the matrix modification in \ref{item:Phase3}.

First, we consider \ref{item:Phase2}.
Denote by $\D\prn[\big]{A^k}$ the dual problem of the minimum-weight perfect matching problem on $G\prn[\big]{A^k}$ given in \cref{sec:weighted-bipartite-matching}.
For $p, q \in \setZ^n$, put
\begin{align}\label{eq:properlized_A}
  B = \prn[\big]{B_{i,j}} \defeq D\prn{\pi^{-p}} A^k D\prn{\pi^{-q}}.
\end{align}
Then for every $i, j \in \intset{n}$, we have
\begin{align}
  v\prn[\big]{B_{i,j}}
  = v\prn[\big]{\pi^{-p_i} A^k_{i,j} \pi^{-q_j}}
  = v\prn[\big]{A^k_{i,j}} - p_i - q_j,
\end{align}
which is nonnegative if $(p, q)$ is feasible to $\D\prn[\big]{A^k}$.
In particular, if $(p, q)$ is feasible, then $B \in R^{n \times n}$.

The \emph{tight coefficient matrix} $A^\# = \prn[\big]{A_{i,j}^\#}$ of $A^k$ with respect to a feasible solution $(p, q)$ of $\D(A)$ is the coefficient matrix of $\pi^0$ in the $\pi$-adic expansion of $B$.
In particular, when $F$ is a field, $A_{i,j}^\#$ is equal to the coefficient of $\pi^{p_i + q_j}$ in the $\pi$-adic expansion of $A_{i,j}$ for $i, j \in \intset{n}$.
Note that $A^\#$ depends on $(p, q)$.
Then $A^\#$ can be used for characterizing the optimality of $(p, q)$ and the upper-tightness of $A^k$ as follows:

\begin{proposition}\label{prop:tcf_trank_optimal}
  Let $A^\#$ be the tight coefficient matrix of $A^k$ with respect to an integral feasible solution $(p, q)$ of $\D\prn[\big]{A^k}$.
  Then $(p, q)$ is optimal if and only if $\trank A^\# = n$.
\end{proposition}
\begin{proof}
  For $i,j \in \intset{n}$, the element $A_{i,j}^\#$ is nonzero if and only if $v\prn[\big]{A_{i,j}^\#} = 0$, which is equivalent to $v\prn[\big]{A^k_{i,j}} = p_i + q_j$.
  Thus $G\prn[\big]{A^\#}$ coincides with the subgraph $G^\#$ of $G\prn[\big]{A^k}$ defined by~\eqref{def:tight-edges} with respect to $(p, q)$.
  By \cref{prop:bipartite-matching-complementarity}, having a perfect matching for $G\prn[\big]{A^\#}$ is equivalent to the optimality of $(p, q)$.
\end{proof}

\begin{proposition}\label{prop:tcf_upper-tight}
  Let $A^\#$ be the tight coefficient matrix of $A^k$ with respect to an integral optimal solution $(p, q)$ of $\D\prn[\big]{A^k}$.
  Then $A^k$ is upper-tight if and only if $A^\#$ is nonsingular.
\end{proposition}

\begin{proof}
  Since $\vdet(B) = \vdet\prn[\big]{A^k} - \tvdet\prn[\big]{A^k}$, the matrix $A$ is upper-tight if and only if $\vdet(C) = 0$.
  This is equivalent to the nonsingularity of $A^\#$ by \cref{prop:biproper_equivalence}.
\end{proof}

By \cref{prop:tcf_upper-tight}, we can check the upper-tightness of $A^k$ just by checking the nonsingularity of $A^\#$.

Modification in \ref{item:Phase3} is as follows.
Suppose that $A^k$ is not upper-tight.
Since the tight coefficient matrix $A^\#$ with respect to an integral dual optimal solution $(p, q)$ is singular by \cref{prop:tcf_upper-tight}, there exists $U \in \GL_n(K)$ such that
\begin{align}\label{eq:condition-of-U}
  \trank UA^\# = \rank UA^\# = \rank A^\# < n.
\end{align}
This $U$ can be obtained by the Gaussian elimination applied to $A^\#$.
We put $A^{k+1} \defeq U'A^k$, where $U' \defeq D\prn{\pi^p}UD\prn{\pi^{-p}}$.

\begin{lemma}\label{lem:update-A}
  It holds $\vdet\prn[\big]{A^k} = \vdet\prn[\big]{A^{k+1}}$ and $\tvdet\prn[\big]{A^k} < \tvdet\prn[\big]{A^{k+1}}$.
\end{lemma}
\begin{proof}
  We have
  \begin{align}
    \vdet(U')
    = \vdet\prn{D\prn{\pi^p}} + \vdet(U) + \vdet\prn{D\prn{\pi^{-p}}}
    = \vdet(U)
    = 0
  \end{align}
  and hence $\vdet\prn[\big]{A^k} = \vdet\prn[\big]{A^{k+1}}$.

  To prove $\tvdet\prn[\big]{A^k} < \tvdet\prn[\big]{A^{k+1}}$, it suffices to show that $(p, q)$ is feasible but not optimal to $\D\prn[\big]{A^{k+1}}$.
  We first show the feasibility.
  Using $B$ defined by~\eqref{eq:properlized_A}, we can rewrite $A^{k+1}$ as
  \begin{gather}
    A^{k+1}
    = U'A^k
    = D\prn{\pi^p}UD\prn{\pi^{-p}}D\prn{\pi^p}BD\prn{\pi^{q}}
    = D\prn{\pi^p}CD\prn{\pi^{q}},\label{eq:bar_A_U_B}
  \shortintertext{where}
    C \defeq UB.\label{def:C}
  \end{gather}
  Since $U, B \in R^{n \times n}$, the matrix $C$ is also over $R$.
  Thus we have $v\prn[\big]{A_{i,j}^{k+1}} \ge p_i + q_j$.
  Hence $(p, q)$ is feasible to $\D\prn[\big]{A^{k+1}}$.

  By~\eqref{eq:bar_A_U_B}, the tight coefficient matrix of $A^{k+1}$ with respect to $(p, q)$ is $UA^\#$.
  Therefore, by \cref{prop:tcf_trank_optimal}, $(p, q)$ is not optimal to $\D\prn[\big]{A^{k+1}}$.
\end{proof}

\subsection{Improved Algorithm}\label{sec:improved-algorithm}

To compute $A^{k+1}$ in \ref{item:Phase3}, we need to multiply $D(\pi^{-p})$, $U$, and $D(\pi^p)$ in this order from left to $A^k$.
This operation includes the computation of the coefficients in the $\pi$-adic expansion of $\pi^{-1} a$ for $a \in R$.
This, however, is impossible for the computational model assumed in \cref{sec:computational-model} because the oracle of computing the inverse of $\delta_0$ is needed.

To avoid left-multiplying $\pi^{-1}$, we slightly improve the above faithful procedure of combinatorial relaxation.
The improved algorithm does not modify the input matrix $A$.
Instead, the algorithm keeps track of $\tvdet\prn[\big]{A^k}$ and the matrix $C \in R^{n \times n}$ defined by~\eqref{def:C}.
The improved algorithm is outlined as follows.

\paragraph{Improved Combinatorial Relaxation Algorithm over DVSFs}
\begin{enumerate}[label={\upshape{\textbf{Phase~\arabic*b.}}},ref={Phase~\arabic*b}]
  \setcounter{enumi}{-1}
  \item\label{item:Phase0_C}
    Set $\gamma^0 \defeq 0$, $C^0 \defeq A$, and $k \gets 0$.
  \item\label{item:Phase1_C}
    Compute an integral optimal solution $(\Delta p, \Delta q)$ of $\D\prn[\big]{C^k}$ such that $\Delta p$ is nonpositive.
    Set $\gamma^{k+1} \defeq \gamma^k + \tvdet\prn[\big]{C^k}$.
    If $\gamma^{k+1} > M$, report $\vdet(A) = +\infty$ and halt.
    Set
    \[
      B^{k+1} \defeq D\prn[\big]{\pi^{-\Delta p}} C^k D\prn[\big]{\pi^{-\Delta q}}.\label{def:B_k_plus_1}
    \]
  \item\label{item:Phase2_C}
    If the coefficient matrix $A^\# \defeq B_0^{k+1}$ of $\pi^0$ in the $\pi$-adic expansion of $B^{k+1}$ is nonsingular, report $\vdet(A) = \gamma^{k+1}$ and halt.
  \item\label{item:Phase3_C}
    Take $U \in \GL_n(K)$ satisfying~\eqref{eq:condition-of-U} and set $C^{k+1} \defeq UB^{k+1}$.
    Put $k \gets k + 1$ and go back to \ref{item:Phase1_C}.
\end{enumerate}

The validity of the improved algorithm is guaranteed by the following lemma.
We denote by $\Pi(p, q)$ the objective function of the dual of the bipartite matching problem, i.e.,
\begin{align}
  \Pi(p, q) \defeq \sum_{i=1}^n p_i + \sum_{j=1}^n q_j.
\end{align}

\begin{lemma}\label{lem:correspondence-of-original-and-improved-combrel}
  For $k \ge 1$, we have $\gamma^k = \Pi(p, q)$ and $B^k = D\prn{\pi^{-p}} A^k D\prn{\pi^{-q}}$ for some integral optimal solution $(p, q)$ of $\D\prn[\big]{A^k}$.
\end{lemma}

\begin{proof}
  We show the claim by induction on $k$.
  The claim is clear when $k = 1$.
  Suppose that the claim holds for some $k \ge 1$.
  By the inductive assumption, $A^\# \defeq B_0^k$ is the tight coefficient matrix of $A^k$ with respect to an optimal solution $(p, q)$ of $\D\prn[\big]{A^k}$.
  Let $U \in \GL_n(K)$ be a matrix satisfying~\eqref{eq:condition-of-U}.
  We have $A^{k+1} = D(\pi^p) U D(\pi^{-p}) A^{k}$ and $C^k = UB^k$.
  Let $(\Delta p, \Delta q)$ be an optimal solution of $\D\prn[\big]{C^k}$ and put $\bar{p} \defeq p + \Delta p$ and $\bar{q} \defeq q + \Delta q$.
  Then we have
  \begin{align}
    C^k = UB^k = U D\prn{\pi^{-p}} A^k D\prn{\pi^{-q}} = D\prn[\big]{\pi^{-p}} A^{k+1} D\prn{\pi^{-q}}.
  \end{align}
  This means that $G\prn[\big]{C^k} = G\prn[\big]{A^{k+1}}$ and edge weights $w_{C^k}(e)$ and $w_{A^{k+1}}(e)$ for $e = \set{i, j} \in E\prn[\big]{C^k} = E\prn[\big]{A^{k+1}}$ satisfy
  \begin{align}
    w_{C^k}(e) = w_{A^{k+1}}(e) - p_i - q_j
  \end{align}
  for $i, j \in \intset{n}$.
  Therefore, $(\bar{p}, \bar{q})$ is optimal to $\D\prn[\big]{A^{k+1}}$ if and only if $(\Delta p, \Delta q)$ is optimal to $\D\prn[\big]{C^k}$.
  Thus we have
  \begin{align}
    \gamma^{k+1} &= \gamma^k + \tvdet\prn[\big]{C^k} = \gamma^k + \Pi(\Delta p, \Delta q) = \Pi(\bar{p}, \bar{q})
  \shortintertext{and}
    B^{k+1}
    &= D\prn[\big]{\pi^{-\Delta p}}C^kD\prn[\big]{\pi^{-\Delta q}} \\
    &= D\prn[\big]{\pi^{-\Delta p}} D\prn[\big]{\pi^{-p}} A^{k+1} D\prn{\pi^{-q}} D\prn[\big]{\pi^{-\Delta q}} \\
    &= D\prn[\big]{\pi^{-\bar{p}}} A^{k+1} D\prn[\big]{\pi^{-\bar{q}}},
  \end{align}
  as required.
\end{proof}

\begin{corollary}\label{cor:validity-of-improved-algorithm}
  The improved combinatorial relaxation algorithm correctly outputs $\vdet(A)$.
\end{corollary}

\begin{proof}
  Follows from \cref{prop:tvdet_le_vdet,prop:tcf_upper-tight,lem:update-A,lem:correspondence-of-original-and-improved-combrel}, and the assumption on $M$.
\end{proof}

We require $\Delta p$ in \ref{item:Phase1_C} to be nonpositive so that we can avoid left-multiplying $\pi^{-1}$ in the computation of~\eqref{def:B_k_plus_1}.
Here we describe how we can obtain such an optimal solution $(\Delta p, \Delta q)$ of $\D\prn[\big]{C^k}$.
First, we initialize $\Delta p$ and $\Delta q$ as zero vectors, which is feasible to $\D\prn[\big]{C^k}$ as the edge weight is nonnegative.
We then iterate the following procedure.
Construct the subgraph $G^\# = \prn[\big]{\intset{n} \sqcup \intset{n}, E^\#}$ of $G\prn[\big]{C^k}$ defined by~\eqref{def:tight-edges} with respect to $(\Delta p, \Delta q)$.
If $G^\#$ has a perfect matching, then $(\Delta p, \Delta q)$ is optimal from \cref{prop:bipartite-matching-complementarity} and we are done.
Otherwise, by \cref{thm:konig}, there exists $I, J \subseteq \intset{n}$ with $\card{I} + \card{J} < n$ such that $(i, j) \in E^\#$ implies $i \in I$ or $j \in J$.
We change $(\Delta p, \Delta q)$ into $(\Delta p', \Delta q')$ by
\begin{align}
  \Delta p'_i \defeq \begin{cases}
    \Delta p_i - 1 & (i \in I), \\
    \Delta p_i     & (i \in \intset{n} \setminus I),
  \end{cases}
  \quad
  \Delta q'_j \defeq \begin{cases}
    \Delta q_j     & (j \in J), \\
    \Delta q_j + 1 & (j \in \intset{n} \setminus J).
  \end{cases}\label{update-pq}
\end{align}
Note that $\Delta p'_i \le 0$ by $\Delta p_i \le 0$ for $i \in \intset{n}$.
The following lemma is well-known:

\begin{lemma}[{\cite{Kuhn1955}}]\label{lem:delta-p-q-objective-value-decrease}
  Let $(\Delta p, \Delta q)$ be a feasible but not optimal dual solution.
  Then $(\Delta p', \Delta q')$ given by~\eqref{update-pq} is also feasible and $\Pi(\Delta p, \Delta q) < \Pi(\Delta p', \Delta q')$.
\end{lemma}

By \cref{lem:delta-p-q-objective-value-decrease}, the updated $(\Delta p, \Delta q)$ is an improved feasible solution of $\D\prn[\big]{C^k}$.
If $\gamma^k + \Pi(\Delta p, \Delta q) > M$, then report $\vdet(A) = +\infty$ and halt immediately.
Otherwise, go back to the construction of $G^\#$ with respect to the updated $(\Delta p, \Delta q)$.

One more implementation issue on computing~\eqref{def:B_k_plus_1} is left: since the $\pi$-adic expansions of entries in $B^{k+1}$ might have infinitely many terms, we cannot store all of them.
We thus truncate higher-valuation terms relying on \cref{prop:truncation}.
Let
\begin{align}\label{def:tilde_B_k_plus_1}
  \tilde{B}^{k+1} \defeq \sum_{d=0}^{M - \gamma^{k+1}} B^{k+1}_d \pi^d,
\end{align}
where $B^{k+1}_d \in K^{n \times n}$ is the coefficient matrix of $\pi^d$ in the $\pi$-adic expansion of $B^{k+1}$ for $d \in \setN$.
We replace $B^{k+1}$ with $\tilde{B}^{k+1}$ in \ref{item:Phase1_C}.
This operation is called the \emph{truncation}.

\begin{lemma}\label{lem:validity-of-truncation}
  The improved algorithm returns $\vdet(A)$ even if the above truncation procedure is executed.
\end{lemma}

\begin{proof}
  We assume that the truncation is executed only at the $k$th iteration; the general statement follows from this by induction.
  From \cref{cor:validity-of-improved-algorithm}, this algorithm outputs $\vdet\prn[\big]{\tilde{B}^{k+1}} + \gamma^{k+1}$ if $\vdet\prn[\big]{\tilde{B}^{k+1}} + \gamma^{k+1} \le M$ and $+\infty$ otherwise.

  Suppose $\vdet(A) < M$.
  Since $\vdet(A) = \vdet\prn[\big]{C^{k+1}} + \gamma^{k+1} = \vdet\prn[\big]{B^{k+1}} + \gamma^{k+1}$ by \cref{lem:correspondence-of-original-and-improved-combrel}, it holds $\vdet\prn[\big]{B^{k+1}} \le M - \gamma^{k+1}$.
  This means $\vdet\prn[\big]{B^{k+1}} = \vdet\prn[\big]{\tilde{B}^{k+1}}$ by \cref{prop:truncation}.
  Thus, the output of the improved algorithm with truncation coincides with $\vdet(A)$.
  Conversely, suppose $\vdet(A) = +\infty$.
  Then we have $\vdet\prn[\big]{B^{k+1}} = +\infty > M - \gamma^{k+1}$, which implies $\vdet\prn[\big]{\tilde{B}^{k+1}} > M - \gamma^{k+1}$ by \cref{prop:truncation} again.
  Thus, the improved algorithm with truncation outputs $+\infty$.
\end{proof}

Now the first half of \cref{thm:complexity-of-combinatorial-relaxation} is proved as follows.
Recall that $\omega$ denotes the exponent in the time complexity to multiply two matrices over $K$.

\begin{proof}[{of the first half of \cref{thm:complexity-of-combinatorial-relaxation}}]
  The validity of the algorithm follows from \cref{lem:validity-of-truncation}.
  We analyze the running time.

  Suppose that the algorithm is implemented in a way that $C \in R^{n \times n}$ and $\gamma \in \setN$ is updated repeatedly.
  Let $m$ be the number of times the algorithm applied~\eqref{update-pq} in total.
  We have $m \le M$ because one application of~\eqref{update-pq} increases $\gamma$ at least by $1$.
  In each application, we solve the bipartite matching problem, which can be solved in $\Order\prn[\big]{n^{2.5}}$-time by the Hopcroft--Karp algorithm~\cite{Hopcroft1973}.
  Thus the total time complexity of this part is $\Order\prn[\big]{mn^{2.5}} = \Order\prn[\big]{Mn^{2.5}}$.

  For every $i, j \in \intset{n}$, the $(i, j)$th entry in $C$ is multiplied by $\pi$ from left at most $m$ times because one application of~\eqref{update-pq} increases $\Delta p_i$ by at most $1$.
  We compute the leading $\Order(M)$ coefficients in the $\pi$-adic expansion of each entry in $\pi C$.
  This can be done in $\Order(M^2)$-time by~\eqref{eq:coefficient-in-pi-a}.
  Since $C$ has $n^2$ entries, the total running time of this process is $\Order\prn{mM^2n^2} = \Order\prn{M^3n^2}$.

  Matrix computations in \ref{item:Phase2_C} and \ref{item:Phase3_C} can be done in $\Order\prn{Mn^\omega}$-time per each iteration as $B^{k+1}$ contains $\Order(M)$ terms due to the truncation.
  Summing it over $\Order(M)$ iterations, we obtain $\Order\prn{M^2n^\omega}$-time in total.
  Thus the desired time complexity is attained.
\end{proof}

\section{Matrix Expansion Algorithm}\label{sec:matrix-expansion-algorithm}
Let $F$ be a split DVSF with valuation $v$, uniformizer $\pi$, valuation ring $R$, coefficient skew subfield $K$, and associated higher $\delta_0$-derivations $(\delta_d)_{d \in \setN}$.
Let $A = \prn{A_{i,j}} \in F^{n \times n}$ be a square matrix given as the $\pi$-adic expansion~\eqref{eq:input_form} and suppose that $\vdet(A) \le M$ or $\vdet(A) = +\infty$.
This section describes the matrix expansion algorithm for computing $\vdet(A)$.

\subsection{Expanded Matrices}\label{sec:matrix_expansion}
For $i, d \in \setN$, let $A^{(i)}_d \in K^{n \times n}$ denote the coefficient matrix of $\pi^d$ in the $\pi$-adic expansion of $\pi^i A$.
Namely, for $i \in \setN$, the matrix $\pi^i A$ is written as
\begin{align}
  \pi^i A = \sum_{d=0}^{\infty} A^{(i)}_{d} \pi^d.
\end{align}
Note that $A_d^{(i)} = O$ for $d < i$ as the valuations of entries in $\pi^i A$ are at least $i$.
For $\mu \in \setN$, we define the \emph{$\mu$th-order expanded matrix} $\Omega_\mu(A)$ of $A$ as the following $\mu n \times \mu n$ block matrix
\begin{align}\label{def:expanded-matrix}
  \Omega_\mu(A) \defeq
  \begin{pNiceMatrix}
    A_0^{(0)} & A_1^{(0)} & \Cdots    & \Cdots    & \Cdots              & A_{\mu-1}^{(0)}     \\
    O         & A_1^{(1)} & A_2^{(1)} & \Cdots    & \Cdots              & A_{\mu-1}^{(1)}     \\
    \Vdots    & \Ddots    & \Ddots    & \Ddots    &                     & \Vdots              \\
    \Vdots    &           & \Ddots    & \Ddots    & \Ddots              & \Vdots              \\
    \Vdots    &           &           & \Ddots    & A_{\mu-2}^{(\mu-2)} & A_{\mu-1}^{(\mu-2)} \\
    O         & \Cdots    & \Cdots    & \Cdots    & O                   & A_{\mu-1}^{(\mu-1)}
  \end{pNiceMatrix}
  \in K^{\mu n \times \mu n}.
\end{align}
Expanded matrices satisfy the multiplicativity as follows (see also~\cite[Section~1.2]{Dumas1992}).
This is an extension of the result in~\cite{Vandooren1979} for rational function matrices over $\setC$.

\begin{lemma}\label{lem:omega_hom}
  Let $A \in R^{n \times n}$ and $B \in R^{n \times n}$ be matrices over a split DVR $R$.
  Then it holds
  \begin{align}
    \Omega_\mu(AB) = \Omega_\mu(A) \Omega_\mu(B)
  \end{align}
  for $\mu \in \setN$.
\end{lemma}
\begin{proof}
  Fix $i \in \intset{0, \mu - 1}$ and let $\pi^i A = \sum_{d=0}^{\infty} A_d^{(i)} \pi^d$ be the $\pi$-adic expansion of $\pi^i A$, where $\pi$ is a uniformizer of $R$.
  Similarly, for $d \in \intset{0, \mu - 1}$, let $\pi^d B = \sum_{j=0}^{\infty} B_j^{(d)} \pi^j$ be the $\pi$-adic expansion of $\pi^d B$.
  Then it holds
  \begin{align}
    \pi^i  AB
    = \prn{\sum_{d=0}^{\infty} A_d^{(i)} \pi^d} B
    = \sum_{d=0}^{\infty} A_d^{(i)} \prn{\sum_{j=0}^{\infty} B_j^{(d)} \pi^j}
    = \sum_{j=0}^{\infty} \prn{\sum_{d=0}^{j} A_d^{(i)} B_j^{(d)}} \pi^j, \label{eq:sd_AB}
  \end{align}
  where the inner sum of the last term stops at $d = j$ by $B_j^{(d)} = O$ for $j < d$.
  The equality~\eqref{eq:sd_AB} implies that the coefficient matrix of $\pi^j$ in the $\pi$-adic expansion of $\pi^i AB$ is
  \begin{align}
    \sum_{d=0}^j A_d^{(i)} B_j^{(d)} = \sum_{d=0}^{\mu-1} A_d^{(i)} B_j^{(d)}
  \end{align}
  for $j < \mu$, which is equal to the $(i+1, j+1)$st entry of $\Omega_\mu(A) \Omega_\mu(B)$.
\end{proof}

Let $\omega_\mu(A)$ denote the rank of $\Omega_\mu(A)$.
The following lemma claims that $\omega_\mu(A)$ coincides with that of the Smith--McMillan form (see \cref{prop:dvsf_smith_mcmillan}) of $A$.

\begin{lemma}\label{thm:rank_of_omega}
  Let $A \in R^{n \times n}$ be a matrix over a split DVR $R$.
  Then it holds $\omega_\mu(A) = \omega_\mu(D)$ for $\mu \in \setN$, where $D$ is the Smith--McMillan form of $A$.
\end{lemma}
\begin{proof}
  Let $S \in R^{n \times n}$ and $T \in R^{n \times n}$ be biproper matrices such that $SAT = D$.
  From \cref{lem:omega_hom}, we have
  \begin{align}
    \omega_\mu(D)
    = \rank \Omega_\mu(SAT)
    = \rank \Omega_\mu(S) \Omega_\mu(A) \Omega_\mu(T).
  \end{align}
  For $i \in \setN$, let $S_i^{(i)}$ be the coefficient matrix of $\pi^i$ in the $\pi$-adic expansion of $\pi^i S$, where $\pi$ is a uniformizer of $R$.
  Then $S_i^{(i)}$ is equal to the coefficient matrix of $\pi^0$ in the $\pi$-adic expansion of $\pi^{-i} S \pi^i$.
  Now $\pi^{-i} S \pi^i$ is biproper by $\prn{\pi^{-i} S \pi^i}^{-1} = \pi^{-i} S^{-1} \pi^i$.
  Thus, $S_i^{(i)}$ is nonsingular from \cref{prop:biproper_equivalence}.
  Since $\Omega_\mu(S)$ is a block triangular matrix having $S_i^{(i)}$ for the $(i+1)$st diagonal block, it is nonsingular.
  Similarly, $\Omega_\mu(T)$ is nonsingular.
  Therefore, we have $\omega_\mu(D) = \omega_\mu(A)$.
\end{proof}

Let $0 \le \alpha_1 \le \dotsb \le \alpha_r$ be the exponents of the Smith--McMillan form of $A \in R^{n \times n}$ with $r \defeq \rank A$.
Put
\begin{align}\label{def:N_d}
  N_d \defeq \card{\set{i \in \intset{r}}[\alpha_i \leq d]}
\end{align}
for $d \in \setN$.
\Cref{thm:rank_of_omega} leads us to the following lemma; a similar result based on the Kronecker canonical form is also known for matrix pencils over a field~\cite[Theorem~2.3]{Iwata2007}.

\begin{lemma}\label{lem:omega_mu_N_d}
  Let $A \in R^{n \times n}$ be a matrix over a split DVR $R$.
  For $\mu \in \setN$, it holds
  \begin{align}\label{eq:omega_mu_alpha}
    \omega_\mu(A) = \sum_{d=0}^{\mu-1} N_d,
  \end{align}
  where $N_d$ is defined by~\eqref{def:N_d}.
\end{lemma}
\begin{proof}
  Let $D$ be the Smith--McMillan form of $A$ and $D^{(i)}_d \in R^{n \times n}$ the coefficient matrix of $\pi^d$ in the $\pi$-adic expansion of $\pi^i D$ for $i,d \in \setN$.
  Since entries of $D$ are powers of $\pi$, the matrix $D$ commutes with $\pi$.
  This implies $D^{(i)}_d = D^{(0)}_{d-i} \eqdef D_{d-i}$ for $d \ge i$.
  Now $\Omega_\mu(D)$ is in the form
  \begin{align}
    \Omega_\mu(D) =
    \begin{pNiceMatrix}
      D_0    & D_1    & \Cdots & D_{\mu-2} & D_{\mu-1} \\
      O      & \Ddots & \Ddots &           & D_{\mu-2} \\
      \Vdots & \Ddots & \Ddots & \Ddots    & \Vdots    \\
      \Vdots &        & \Ddots & \Ddots    & D_1       \\
      O      & \Cdots & \Cdots & O         & D_0
    \end{pNiceMatrix}.\label{eq:Omega_mu_D}
  \end{align}
  Let $\alpha_1, \dotsc, \alpha_r$ be the exponents of the Smith--McMillan form $D$, where $r \defeq \rank A$.
  The $i$th diagonal entry of $D_d$ is $1$ if $i \le r$ and $\alpha_i = d$, and 0 otherwise.
  Thus from~\eqref{eq:Omega_mu_D}, each row and column in $\Omega_\mu(D)$ has at most one nonzero entry.
  Hence $\omega_\mu(D)$, which is equal to $\omega_\mu(A)$ by \cref{thm:rank_of_omega}, is equal to the number of nonzero entries in $\Omega_\mu(D)$.
  It is easily checked that the $(\mu-d)$th block row of $\Omega_\mu(D)$ contains $N_d$ nonzero entries for $d \in \intset{0, \mu-1}$.
\end{proof}

The equality~\eqref{eq:omega_mu_alpha} is a key identity that connects $\omega_\mu(A)$ and the Smith--McMillan form of $A$.
We remark that~\eqref{eq:omega_mu_alpha} can be rewritten as
\begin{align} \label{eq:N_leq_mu_omega}
  N_d = \omega_{d+1}(A) - \omega_d(A)
\end{align}
for $d \in \setN$.

\subsection{Legendre Conjugacy}\label{sec:legendre_conjugacy}

Let $A \in R^{n \times n}$ be a matrix of rank $r$ and $\alpha_1 \le \dotsc \le \alpha_r$ the exponents of the Smith--McMillan form of $A$.
Put $\zeta_k \defeq \zeta_k(A)$ for $k = \intset{0, r}$, where $\zeta_k(A)$ is defined by~\eqref{def:zeta_k}.
From $\alpha_k \le \alpha_{k+1}$ and~\eqref{eq:dvsf_smith_mcmillan_vd_i}, an inequality $\zeta_{k-1} + \zeta_{k+1} \ge 2 \zeta_k$ holds for all $k \in \intset{r-1}$.
In addition, for $\mu \in \setN$ put $\omega_\mu \defeq \omega_\mu(A)$ and define $N_{\mu}$ by~\eqref{def:N_d}.
From $N_{\mu-1} \le N_{\mu}$ and~\eqref{eq:N_leq_mu_omega}, we have $\omega_{\mu-1} + \omega_{\mu+1} \ge 2 \omega_\mu$ for all $\mu \ge 1$.
These two inequalities for $d_k$ and $\omega_\mu$ indicate the \emph{convexity} of $\zeta_k$ and $\omega_\mu$ in the following sense.
A (discrete) function $\funcdoms{f}{\setZ}{\setZ \cup \set{+\infty}}$ is said to be \emph{convex} if
\begin{align}
  f(x-1) + f(x+1) \ge 2f(x)
\end{align}
for all $x \in \setZ$.
We call a function $\funcdoms{g}{\setZ}{\setZ \cup \set{-\infty}}$ \emph{concave} if $-g$ is convex.
An integer sequence $\prn{a_k}_{k \in K}$ indexed by $K \subseteq \setZ$ can be identified with a function $\funcdoms{\check{a}}{\setZ}{\setZ \cup \set{+\infty}}$ by letting $\check{a}(k)$ be $a_k$ if $k \in K$ and $+\infty$ otherwise.
We can also identify $a$ with $\funcdoms{\hat{a}}{\setZ}{\setZ \cup \set{-\infty}}$ defined by $\hat{a}(k) \defeq a_k$ if $k \in K$ and $\hat{a}(k) \defeq -\infty$ otherwise.
In this way, we identify $(\zeta_0, \zeta_1, \dotsc, \zeta_r)$ and $(\omega_0, \omega_1, \omega_2, \dotsc)$ with discrete functions $\funcdoms{\check{\zeta}}{\setZ}{\setZ \cup \set{-\infty}}$ and $\funcdoms{\hat{\omega}}{\setZ}{\setZ \cup \set{+\infty}}$, respectively.
From the argument in the previous paragraph, both $(\zeta_0, \zeta_1, \dotsc, \zeta_r)$ and $(\omega_0, \omega_1, \omega_2, \dotsc)$ are convex.
Let $\funcdoms{f}{\setZ}{\setZ \cup \set{+\infty}}$ be a function such that $f(x) \in \setZ$ for some $x \in \setZ$.
The \emph{concave conjugate} of $f$ is a function $\funcdoms{f^\circ}{\setZ}{\setZ \cup \set{-\infty}}$ defined by
\begin{align}
  f^\circ(y) \defeq \inf_{x \in \setZ} (f(x) - xy)
\end{align}
for $y \in \setZ$.
Similarly for a function $\funcdoms{g}{\setZ}{\setZ \cup \set{-\infty}}$ with $g(y) \in \setZ$ for some $y \in \setZ$, the \emph{convex conjugate} of $g$ is a function $\funcdoms{g^\bullet}{\setZ}{\setZ \cup \set{+\infty}}$ given by
\begin{align}
  g^\bullet(x) \defeq \sup_{y \in \setZ} (g(y) + xy)
\end{align}
for $x \in \setZ$.
The maps $f \mapsto f^\circ$ and $g \mapsto g^\bullet$ are referred to as the \emph{concave} and \emph{convex discrete Legendre transform}, respectively.
In general $f^\circ$ is concave and $g^\bullet$ is convex.
If $f$ is convex and $g$ is concave,
\begin{align} \label{eq:legendre_of_legendre}
  \prn{f^\circ}^\bullet = f,
  \quad
  \prn{g^\bullet}^\circ = g
\end{align}
hold.
Hence the Legendre transformation establishes a one-to-one correspondence between discrete convex and concave functions.
See~\cite{Murota2003} for details of discrete convex/concave functions and their Legendre transform.

Indeed, the sequences of $\zeta_k$ and $-\omega_\mu$ are in the relation of Legendre conjugate.
This can be shown from the key identities~\eqref{eq:zeta-as-sum-of-alpha} and~\eqref{eq:omega_mu_alpha} that connect $\zeta_k(A)$ and $\omega_\mu(A)$ through the Smith--McMillan form of $A$.

\begin{theorem}\label{thm:legendre}
  Let $A \in R^{n \times n}$ be a matrix of rank $r$ over a split DVR $R$.
  Then it holds
  \begin{alignat}{2}
    \zeta_k(A)    &= \max_{\mu \ge 0}     (k\mu - \omega_\mu(A)) &\quad& (0 \le k \le r), \label{eq:d_inf_omega} \\
    \omega_\mu(A) &= \max_{0 \le k \le r} (k\mu - \zeta_k(A))    &\quad& (\mu \ge 0). \label{eq:omega_sup_d}
  \end{alignat}
\end{theorem}

\begin{figure}[tbp]
  \centering
  \begin{tikzpicture}[>=latex, x=0.8cm, y=0.7cm]
    \draw[thick, ->] (0, 0) -- (6.5, 0) node[right] {$x$};
    \draw[thick, ->] (0, 0) -- (0, 6.5) node[above] {$y$};
    \fill[pattern=dots] (0, 0) -- (0, 1) -- (1, 1) -- (1, 2) -- (2, 2) -- (2, 3) -- (3, 3) -- (3, 4) -- (6, 4) -- (6, 0) -- cycle;
    \draw[thin] (0, 1) -- (1, 1) -- (1, 2) -- (2, 2) -- (2, 3) -- (3, 3) -- (3, 4.5) -- (4, 4.5) --  (4, 5) -- (5, 5) -- (5, 6) -- (6, 6) -- (6, 0);
    \draw[thin] (3, 4) -- (0, 4) node[left] {$\mu$};
    \draw[dashed] (3, 4) -- (6, 4);
    \draw[dotted] (0, 1) -- (0, 1) node[left] (a1) {$\alpha_1$};
    \draw[dotted] (1, 2) -- (0, 2) node[left] (a2) {$\alpha_2$};
    \draw[dotted] (4, 5) -- (0, 5) node[left] (ar1) {$\alpha_{r-1}$};
    \draw[dotted] (5, 6) -- (0, 6) node[left] (ar) {$\alpha_{r}$};
    \node[below left] at (0, 0) {$\mathrm{O}$};
    \node[below] (1) at (1, 0) {$1\vphantom{mathrm{O}}$};
    \node[below] (2) at (2, 0) {$2\vphantom{mathrm{O}}$};
    \node[below] (r1) at (5, 0) {$r-1\vphantom{mathrm{O}}$};
    \node[below] (r) at (6, 0) {$r\vphantom{mathrm{O}}$};
    \node at ($(2)!0.5!(r1)$) {$\cdots\vphantom{mathrm{O}}$};
    \node at (1, 3) {$\omega_{\mu}$};
    \node[ellipse, fill=white, minimum width=3.5cm, minimum height=0.9cm] at (3.5, 1) {};
    \node at (3.5, 1) {$\sum_{i = 1}^r \min \set{\alpha_i, \mu}$};
  \end{tikzpicture}
  \caption{%
    Graphic explanation of~\eqref{eq:omega_mu_and_alpha}.
  }\label{fig:delta_r_and_omega}
\end{figure}
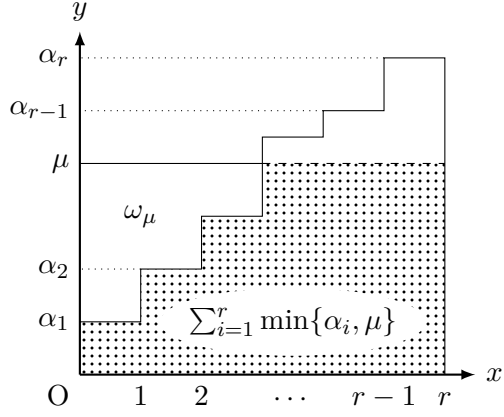

\begin{proof}
  Put $\zeta_k \defeq \zeta_k(A)$ for $k \in \intset{0, r}$ and $\omega_\mu \defeq \omega_\mu(A)$ for $\mu \in \setN$.
  Since $(\zeta_0, \zeta_1, \dotsc, \zeta_r)$ is convex and $(-\omega_0, -\omega_1, -\omega_2, \dotsc)$ is concave,~\eqref{eq:d_inf_omega} and~\eqref{eq:omega_sup_d} are equivalent by~\eqref{eq:legendre_of_legendre}.
  We show~\eqref{eq:omega_sup_d} as follows.

  First we give an equality
  \begin{align} \label{eq:omega_mu_and_alpha}
    \omega_\mu = r\mu - \sum_{i = 1}^r \min \set{\alpha_i, \mu}
  \end{align}
  for $\mu \in \setN$, where $\alpha_1 \le \dotsb \le \alpha_r$ are the exponents of the Smith--McMillan form of $A$.
  Figure~\ref{fig:delta_r_and_omega} graphically shows this equality.
  Let $x$ and $y$ be the coordinates along the horizontal and vertical axes in Figure~\ref{fig:delta_r_and_omega}, respectively.
  For $i \in \intset{r}$, the height of the dotted rectangle with $i-1 \leq x < i$ is $\min \set{\alpha_i, \mu}$.
  Hence the area of the dotted region is equal to $\sum_{i = 1}^r \min \set{\alpha_i, \mu}$.
  In addition, the width of the white rectangle with $d \leq y < d+1$ is equal to $N_d$ for $d = 0, \dotsc, \mu - 1$, where $N_d$ is defined by~\eqref{def:N_d}.
  Hence the area of the white stepped region is equal to $N_0 + \dotsb + N_{\mu-1} = \omega_{\mu}$ by~\eqref{eq:omega_mu_alpha}.
  Now we have~\eqref{eq:omega_mu_and_alpha} since the sum of the areas of these two regions is $r \mu$.

  Substituting~\eqref{eq:zeta-as-sum-of-alpha} into the right hand side of~\eqref{eq:omega_sup_d}, we have
  \begin{align} \label{eq:legendre_proof_1}
    \max_{0 \le k \le r} (k\mu - \zeta_k)
    = \max_{0 \le k \le r} \sum_{i=1}^k (\mu - \alpha_i)
    = k^* \mu - \sum_{i=1}^{k^*} \alpha_i,
  \end{align}
  where $k^*$ is the maximum $0 \le k \le r$ such that $\alpha_k \le \mu$.
  Since $\min \set{\alpha_i, \mu}$ is $\alpha_i$ if $i \le k^*$ and $\mu$ if $i > k^*$, it holds
  \begin{align} \label{eq:legendre_proof_2}
    \sum_{i=1}^r \min \set{\alpha_i, \mu}
    = (r - k^*)\mu + \sum_{i=1}^{k^*} \alpha_i.
  \end{align}
  From~\eqref{eq:legendre_proof_1} and~\eqref{eq:legendre_proof_2}, we have
  \begin{align}
    \max_{0 \le k \le r} (k\mu - \zeta_k)
    = r\mu - \sum_{i = 1}^r \min \set{\alpha_i, \mu},
  \end{align}
  in which the right hand side is equal to $\omega_\mu$ by~\eqref{eq:omega_mu_and_alpha}.
\end{proof}

\subsection{Reduction and Algorithm}\label{sec:reductions-and-algorithms}
We finally apply \cref{thm:legendre} to the computation of $\vdet(A)$ via the following lemma.

\begin{lemma}\label{lem:rank_and_d_r_formulas}
  Let $A \in F^{n \times n}$ be a matrix~\eqref{eq:input_form} of rank $r$ over a split DVSF $F$ such that $\vdet(A) \le M$ or $\vdet(A) = +\infty$.
  Then $A$ is nonsingular if and only if $\omega_{M+1}(A) - \omega_M(A) = n$.
  Furthermore, if $A$ is nonsingular, then it holds
  \begin{align}
    \vdet(A) = Mn - \omega_M(A).\label{eq:d_r_formula}
  \end{align}
\end{lemma}
\begin{proof}
  It holds $\omega_{M+1}(A) - \omega_M(A) = N_M \le n$ by~\eqref{eq:N_leq_mu_omega}.
  If $A$ is singular, then $N_M$ must be less than $n$.
  If $A$ is nonsingular, then $\alpha_i$ is at most $M$ for all $i \in \intset{r}$, which means $N_M = n$.

  Suppose that $A$ is nonsingular.
  From~\eqref{eq:omega_mu_alpha} and~\eqref{eq:d_inf_omega}, it holds
  \begin{align} \label{eq:d_r_min_2}
    \vdet(A) = \max_{\mu \ge 0} \sum_{d=0}^{\mu-1} (n - N_d).
  \end{align}
  Since $N_0 \le N_1 \le \dotsb \le N_M = N_{M+1} = \dotsb = n$, the maximum value of the right hand side of~\eqref{eq:d_r_min_2} is attained by $\mu = M$.
  Thus we have~\eqref{eq:d_r_formula}.
\end{proof}

From \cref{lem:rank_and_d_r_formulas}, we can compute $\vdet(A)$ just by calculating $\omega_M(A)$ and $\omega_{M+1}(A)$; we call this the \emph{matrix expansion algorithm}.
These matrices can be constructed in $\Order\prn{M^3 n^2}$-time by repeatedly applying~\eqref{eq:coefficient-in-pi-a} and the rank computation can be done in $\Order\prn{M^\omega n^\omega}$ arithmetic operations on $K$.
Thus we have the last half of \cref{thm:complexity-for-deg-Det-of-skew-polynomial-matrices}.

\section{Estimating Upper Bounds}\label{sec:estimating-upper-bounds}
\subsection{Bounds for Skew Polynomial Rings}\label{sec:bounds-for-skew-polynomial-rings}
Let $R$ be a split DVR with coefficient skew subfield $K$.
In the algorithms presented in \cref{sec:combinatorial-relaxation-algorithm,sec:matrix-expansion-algorithm}, we assume that an upper bound $M$ of $\vdet(A)$ is known beforehand (or $\vdet(A) = +\infty$) for $A \in R^{n \times n}$.
How can we know such $M$?
Recall that entries in the input matrix $A \in R^{n \times n}$ in~\eqref{eq:input_form} contain terms having valuations at most $l$.
One optimistic estimation of the upper bound is $ln$.
From the definition of the determinant, this is valid when $R$ is commutative, or equivalently, $R$ is isomorphic to a subring of $K\ssqbr{s}$.
This can be extended to the case of skew polynomial rings as follows.

Let $K$ be a skew field equipped with an automorphism $\sigma$ and a left $\sigma$-derivation $\delta$.
As stated in \cref{ex:skew_laurent_series}, the skew inverse Laurent series field $K\pprn[\big]{s^{-1}; \sigma, \delta}$ forms a complete split DVR with valuation $-\deg$ and uniformizer $s^{-1}$.
We denote by $K\ssqbr[\big]{s^{-1}; \sigma, \delta}$ the valuation ring of $K\pprn[\big]{s^{-1}; \sigma, \delta}$.
From \cref{ex:higher-derivations}, $K\ssqbr[\big]{s^{-1}; \sigma, \delta}$ is isomorphic to $K\ssqbr[\big]{t; (\delta_d)}$ by an isomorphism $s^{-1} \mapsto t$, where $\delta_d$ is given by~\eqref{eq:higher_derivation_for_skew_inverse_laurent} for $d \in \setN$.

\begin{proposition}\label{prop:skew-polynomials-bound}
  Let $F \defeq K\pprn[\big]{s^{-1}; \sigma, \delta}$ be a skew inverse Laurent field over a skew field $K$.
  For a nonsingular matrix $A = \sum_{d=0}^l A_d s^{-d} \in F^{n \times n}$ with $A_0, \dotsc, A_l \in K^{n \times n}$, we have $\vdet(A) = -\deg \Det A \le ln$.
\end{proposition}

\begin{proof}
  Consider
  \begin{align}
    B \defeq As^l = \sum_{d=0}^l A_{l-d} s^d \in {K[s; \sigma, \delta]}^{n \times n}.
  \end{align}
  Since $\vdet(B) = \vdet(A) + \vdet(I_n s^l) = \vdet(A) + nl$, it suffices to show $-\vdet(B) = \deg \Det B$ is nonnegative.

  The skew polynomial ring $K[s; \sigma, \delta]$ is known to be a (left and right) PID~\cite[Theorem~2.8]{Goodearl2004} as the usual polynomial ring $K[s]$.
  Let $D = UBV$ be the Jacobson normal form of $B$ (see \cref{prop:jacobson}).
  Here, $U, V \in \GL_n(K[s; \sigma, \delta]) \subseteq \GL_n(K\ssqbr[\big]{s^{-1}; \sigma, \delta})$ are biproper matrices.
  By \cref{prop:biproper_equivalence}, we have $\vdet(D) = \vdet(U) + \vdet(B) + \vdet(V) = \vdet(B)$.
  Since diagonal entries in $D$ are nonzero skew polynomials, they have nonnegative degrees.
  Thus we have $\vdet(B) = \vdet(D) \ge 0$.
\end{proof}

A \emph{skew polynomial matrix} over $K$ refers to a matrix over a skew polynomial ring over $K$.
As we have shown in the proof of \cref{prop:skew-polynomials-bound}, for a skew polynomial matrix $A = \sum_{d=0}^l A_{l-d} s^l \in {K[s; \sigma, \delta]}^{n \times n}$, we can reduce the computation of $\deg \Det A$ into that of $-\det \Det As^{-l}$, where
\begin{align}
  As^{-l} = \sum_{d=0}^l A_d s^{-d} \in {K\pprn[\big]{s^{-1}; \sigma, \delta}}^{n \times n}.
\end{align}
From \cref{prop:skew-polynomials-bound}, we can set $M \defeq ln$ for $As^{-l}$.
The coefficients of $s^{-1} a$ satisfy the following recursion formula.

\begin{lemma}\label{lem:inverse-skew-polynomial-coefficient-formula}
  Let $a = \sum_{d=0}^\infty a_d s^{-d} \in K\ssqbr[\big]{s^{-1}; \sigma, \delta}$ with $a_d \in K$ for $d \in \setN$.
  The coefficient $b_d$ of $s^{-d}$ in $s^{-1} a$ satisfies
  \begin{align}\label{eq:skew-polynomial-coefficient-in-pi-a}
    b_d = \begin{cases}
      \sigma^{-1}\prn{a_{d-1} - \delta(b_{d-1})} & (d \ge 1), \\
      0 & (d = 0).
    \end{cases}
  \end{align}
\end{lemma}

\begin{proof}
  By~\eqref{eq:skew_polynomial_commutation_rule}, we have
  \begin{align}\label{eq:proof-of-inverse-skew-polynomial-coefficient-formula}
    a
    &= s \prn[\big]{s^{-1}a} \\
    &= s \sum_{d=0}^\infty b_d s^{-d} \\
    &= \sum_{d=0}^\infty\prn{\sigma(b_d)s + \delta(b_d)} s^{-d} \\
    &= \sigma(b_0)s + \sum_{d=0}^\infty \prn{\sigma\prn{b_{d+1}} + \delta(b_d)} s^{-d}.
  \end{align}
  The equation~\eqref{eq:proof-of-inverse-skew-polynomial-coefficient-formula} means $\sigma(b_0) = 0$ and $\sigma\prn{b_{d+1}} + \delta(b_d) = a_d$ for $d \in \setN$, which imply~\eqref{eq:skew-polynomial-coefficient-in-pi-a}.
\end{proof}

From~\eqref{eq:skew-polynomial-coefficient-in-pi-a}, we can compute the leading $M$ coefficients of $s^{-1} a$ by $\Order(M)$ applications of $\sigma^{-1}$ and $\delta$.
This is improved from $\Order\prn{M^2}$ based on~\eqref{eq:coefficient-in-pi-a}.
Applying this improvement and plugging $ln$ into $M$ in the time complexities in \cref{thm:complexity-of-combinatorial-relaxation}, we obtain \cref{thm:complexity-for-deg-Det-of-skew-polynomial-matrices}.
We can compute $\ord \Det$ of matrices over $K[s; \sigma]$ in the same way.
See \cref{sec:application2} for an application of these computations to differential equations.

\subsection{Characterizing Split DVSFs with Bounds}
In \cref{sec:bounds-for-skew-polynomial-rings}, we described that the valuation of the Dieudonné determinant of nonsingular $A = \sum_{d=0}^l A_d \pi^d \in F^{n \times n}$ is bounded by $ln$ when $F$ is a skew inverse Laurent series field.
Indeed, the converse also holds in the following sense.

\begin{theorem}\label{thm:characeterizing-skew-polynomials}
  Let $F$ be a complete split DVSF with coefficient skew subfield $K$ and uniformizer $\pi$.
  Then every $A = \sum_{d=0}^l A_d \pi^l \in \GL_n(F)$ with $A_0, \dotsc, A_d \in K^{n \times n}$ satisfies $\vdet(A) \le ln$ if and only if $F$ is isomorphic to $K\pprn[\big]{s^{-1}; \sigma, \delta}$ with some automorphism $\sigma$ and left $\sigma$-derivation $\delta$ on $K$.
\end{theorem}

\begin{proof}
  The ``if'' part was shown in \cref{prop:skew-polynomials-bound}.
  We show the ``only if'' part.
  Let $(\delta_d)_{d \in \setN}$ be the higher $\delta_0$-derivatives corresponding to a complete split DVSF $F$.
  We put $\sigma \defeq {\delta_0}^{-1}$ and $\delta \defeq -{\delta_0}^{-1} \delta_1 {\delta_0}^{-1}$.
  The motivation of these notations is the following: if $F$ is isomorphic to $K\pprn[\big]{s^{-1}; \sigma', \delta'}$, then $\sigma' = \sigma$ and $\delta' = \delta$ by~\eqref{eq:higher_derivation_for_skew_inverse_laurent}.
  We can check that $\sigma$ is an automorphism and $\delta$ is a left $\sigma$-derivation.

  For $a \in K$, we put $\pi^{-1} a \pi \eqdef a' = \sum_{d=0}^\infty a'_d \pi^d$ with $a'_0, a'_1, \dotsc \in K$.
  We first show that if $a'_d = 0$ for any $a \in K$ and $d \ge 2$, then $F$ is isomorphic to $K\pprn[\big]{s^{-1}; \sigma, \delta}$.
  Suppose that $F$ satisfies this assumption and put $s \defeq \pi^{-1}$.
  Then it holds
  \begin{align}\label{eq:proof-of-dvsf-charcterization-sa}
    sa
    = \pi^{-1}a
    = a'\pi^{-1}
    = a'_0 \pi^{-1} + a'_1
    = a'_0 s + a'_1
  \end{align}
  for $a \in K$.
  From $a = \pi a' \pi^{-1}$ and~\eqref{eq:coefficient-in-pi-a} for $d = 0, 1$, we have $a = \delta_0\prn[\big]{a_0'}$ and $0 = \delta_0\prn[\big]{a_1'} + \delta_1\prn[\big]{a_0'}$.
  Solving these qualities for $a'_0$ and $a'_1$, we obtain
  \begin{align}
    a'_0 &= {\delta_0}^{-1}(a) = \sigma(a),\label{eq:proof-of-dvsf-charcterization-a0}\\
    a'_1 &= {\delta_0}^{-1}\prn[\big]{-\delta_1\prn[\big]{a_0'}} = -\prn[\big]{{\delta_0}^{-1} \delta_1 {\delta_0}^{-1}}(a) = \delta(a).\label{eq:proof-of-dvsf-charcterization-a1}
  \end{align}
  Substituting~\eqref{eq:proof-of-dvsf-charcterization-a0} and~\eqref{eq:proof-of-dvsf-charcterization-a1} into~\eqref{eq:proof-of-dvsf-charcterization-sa}, we have
  \begin{align}
    sa = \sigma(a)s + \delta(a),
  \end{align}
  which is nothing but the commutation rule~\eqref{eq:skew_polynomial_commutation_rule} of the skew polynomial ring $K[s; \sigma, \delta]$.
  Hence the ring generated by $\pi^{-1}$ over $K$, its Ore quotient skew field, and its completion $F$ with respect to the $\pi$-adic topology are isomorphic to $K[s; \sigma, \delta]$, $K(s; \sigma, \delta)$, and $K\pprn[\big]{s^{-1}; \sigma, \delta}$, respectively.

  Next, suppose that $F$ is not isomorphic to $K\pprn[\big]{s^{-1}; \sigma, \delta}$.
  From the contraposition of the above proof, there exists $a \in K$ such that $a'_d \ne 0$ for some $d \ge 2$; take such $a$ and let $k \ge 2$ be the minimum number with $a'_k \ne 0$.
  Consider
  \begin{align}
    A
    \defeq
    \begin{pmatrix}
      0 & 0 \\
      1 & a'_0
    \end{pmatrix}
    +
    \begin{pmatrix}
      1 & a \\
      0 & a'_1
    \end{pmatrix}
    \pi
    =
    \begin{pmatrix}
      \pi & a\pi \\
      1   & a'_0 + a'_1 \pi
    \end{pmatrix}
    \in F^{2 \times 2}.
  \end{align}
  The values of $l$ and $n$ for $A$ are $l = 1$ and $n = 2$.
  Multiplying an elementary matrix, we can transform $A$ into
  \begin{align}
    B \defeq
    \begin{pmatrix}
              1 & 0 \\
      -\pi^{-1} & 1
    \end{pmatrix}
    A
    =
    \begin{pmatrix}
      \pi & a\pi \\
      0   & a'_0 + a'_1 \pi - \pi^{-1} a \pi
    \end{pmatrix}
    =
    \begin{pmatrix}
      \pi & a\pi \\
      0   & - \sum_{d=k}^\infty a'_d \pi^d
    \end{pmatrix}.
  \end{align}
  Thus, $A$ is nonsingular and it holds
  \begin{align}
    \vdet(A) = \vdet(B) = v(\pi) + v\prn{\sum_{d=k}^\infty a'_d \pi^d} = 1 + k > 2 = ln,
  \end{align}
  where $v$ is the valuation on $F$.
\end{proof}

\Cref{thm:characeterizing-skew-polynomials} means that the condition ``$\vdet(A) \le ln$ for any $A = \sum_{d=0}^l A_d \pi^d \in \GL_n(F)$'' serves as a characterization of skew inverse Laurent series fields.
In this way, skew polynomials arise not only from an algebraic abstraction of linear differential/difference equations but also from the most natural condition for which the combinatorial relaxation and the matrix expansion algorithms are applicable.

\section{Application 1: Weighted Edmonds' Problem}\label{sec:application1}
This section describes applications of our algorithm to (commutative/noncommutative) weighted Edmonds' problem (WEP).
Throughout this section, we assume the arithmetic model on a field $K$.

Let $A = \sum_{d=0}^l A_{d-l} s^d$ be a square commutative or noncommutative linear polynomial matrix~\eqref{eq:linear_polynomial_matrix} over $K$.
That is, $A$ is in ${L(s)}^{n \times n}$, where $L \defeq K(x_1, \dotsc, x_m)$ in the commutative case and $L \defeq K\fsbr{x_1, \dotsc, x_m}$ in the noncommutative case.
Note that $L(s)$ is a split DVSF with valuation $-\deg$.
Instead of $A$, we deal with the following matrix
\begin{align}
  As^{-l} = \sum_{d=0}^l A_d s^{-d}.
\end{align}
Then we can compute $\vdet(A) = -\deg \Det A$ from $\vdet(As^{-l})$ by $\vdet(A) = \vdet(As^{-l}) - ln$.
Since $L(s)$ is a special case of skew rational function fields over $L$, i.e., $L(s) = L(s; \id, 0)$, we have $\vdet(As^{-l}) \le ln$ when $A$ is nonsingular by \cref{prop:skew-polynomials-bound}.

First, consider the combinatorial relaxation algorithm presented in \cref{sec:combinatorial-relaxation-algorithm}.
Since one cannot perform arithmetic operations on $L$ efficiently, it is not immediate to apply the combinatorial relaxation algorithm to $As^{-l}$.
In particular, the procedure of finding the matrix $U \in \GL_n(L)$ in \ref{item:Phase3_C} based on the Gaussian elimination on $L$ requires exponential number of arithmetic operations on $K$.
Nevertheless, in the noncommutative case, we can make use of the following property on nc-linear matrices given by Fortin--Reutenauer~\cite{Fortin2004}.

\begin{theorem}[{\cite[Theorem~1]{Fortin2004}}]\label{thm:MVSP}
  For an nc-linear matrix $B \in {K\fsbr{x_1, \dotsc, x_m}}^{n \times n'}$ over a field $K$, there exist $U \in \GL_n(K)$ and $V \in \GL_{n'}(K)$ such that $\trank UBV = \rank B$.
\end{theorem}

The problem of finding $U$ and $V$ satisfying $\trank UBV = \rank B$, which is a variant of nc-Edmonds' problem by \cref{thm:MVSP}, is called the \emph{maximum vanishing subspace problem} (MVSP) due to Hamada--Hirai~\cite{Hamada2020}.
The MVSP can be solved in deterministic polynomial-time~\cite{Hamada2020,Ivanyos2018}.
Therefore, by using the algorithms in~\cite{Hamada2020,Ivanyos2018} as oracles, we obtain a deterministic polynomial-time algorithm for the nc-WEP\@.
This algorithm indeed coincides with the steepest gradient descent algorithm given by Hirai~\cite{Hirai2019}.

\begin{theorem}[{\cite[Theorem~4.4]{Hirai2019}}]\label{thm:complexity-of-weighted-nc-edmonds-problem-combinatorial-relaxation}
  The nc-WEP for over a field $K$ can be solved in deterministic $\Order\prn{l^2mn^{\omega+2} + T_{\mathrm{MVSP}}(n, m)ln}$-time, where $T_{\mathrm{MVSP}}(n, m)$ denotes the time needed to solve the MVSP for an $n \times n$ nc-linear matrix with $m$ symbols over $K$.
\end{theorem}

\begin{proof}
  In \ref{item:Phase3_C} of each iteration, we solve the MVSP to obtain $U,V \in \GL_n(K)$ and put $C^{k+1} \defeq UB^{k+1}V$.
  This matrix multiplication can be done in $\Order\prn{lmn^{\omega+1}}$ arithmetic operations on $K$.
  Since the number of iterations is $\Order\prn{ln}$, we obtain the desired time complexity.
\end{proof}

We remark that the time complexity in \cref{thm:complexity-of-weighted-nc-edmonds-problem-combinatorial-relaxation} is in terms of the arithmetic model on $K$.
In case of $\setK = \setQ$, the bit-lengths of intermediate numbers are not bounded, even if an algorithm for MVSP guarantees the bounded bit-length.
In addition, since \cref{thm:complexity-of-weighted-nc-edmonds-problem-combinatorial-relaxation} relies on \cref{thm:MVSP}, we cannot apply the combinatorial relaxation for the commutative problem.

We next apply the matrix expansion algorithm in \cref{sec:matrix-expansion-algorithm} to the WEP\@.
This application is rather immediate than that of the combinatorial relaxation algorithm.
Namely, if $A$ is a commutative (noncommutative) linear polynomial matrix over a field $K$, then the expanded matrix $\Omega_\mu(As^{-l})$ given by~\eqref{def:expanded-matrix} is a commutative (resp.\ noncommutative) linear matrix.
Hence the rank computation of $\Omega_\mu(As^{-l})$ is nothing but solving the commutative (resp.\ noncommutative) Edmonds' problem.
By \cref{lem:rank_and_d_r_formulas} and \cref{prop:skew-polynomials-bound}, we obtain the following:

\begin{theorem}\label{thm:complexity-of-weighted-nc-edmonds-problem-matrix-expansion}
  The commutative (noncommutative) WEP over a field $K$ can be solved in deterministic $\Order\prn{T_{\mathrm{EP}}(ln^2, m)}$-time, where $T_{\mathrm{EP}}(n, m)$ denotes the time needed to solve commutative (resp.\ noncommutative) Edmonds' problem for an $n \times n$ commutative (resp.\ noncommutative) linear matrix with $m$ symbols over $K$.
\end{theorem}

The algorithms of Gurvits~\cite{Gurvits2004} and Ivanyos et al.~\cite{Ivanyos2018} deterministically solve nc-Edmonds' problem with polynomially bounded bit complexity when $K = \setQ$.
Using these algorithm as oracles, we obtain \cref{thm:nc-WEP-complexity}.

\begin{remark}
  In view of combinatorial optimization, the algorithm given in \cref{thm:nc-WEP-complexity} is regarded as pseudo-polynomial time algorithms since the running time depends on a polynomial of the maximum exponent $l$ of $s$ instead of $\poly(\log l)$.
  Recently, Hirai--Ikeda~\cite{Hirai2020b} presented algorithms to solve the nc-WEP over $K$ for an nc-linear polynomial matrix in form of
  \begin{align}\label{eq:nc-linear-polynomial-matrix-ikeda}
    A = \sum_{k=0}^m A_k x_k s^{w_k},
  \end{align}
  where $A_1, \dotsc, A_m \in K^{n \times n}$ and $w_1, \dotsc, w_m \in \setZ$.
  The nc-WEP for~\eqref{eq:nc-linear-polynomial-matrix-ikeda} includes the weighted linear matroid intersection problem.
  An algorithm of Hirai--Ikeda runs in strongly polynomial time, i.e., it runs in time polynomial of $n$ and $m$.

  As an extension of a different direction, it is natural to try to solve the (commutative) WEP for
  \begin{align}\label{eq:nc-linear-polynomial-matrix-kronecker}
    A = \sum_{k=0}^m A_k s^{w_k},
  \end{align}
  where $A_1, \dotsc, A_m \in K^{n \times n}$ and $w_1, \dotsc, w_m \in \setZ$.
  However, setting $w_k \defeq \prn{n+1}^k$ for $k \in \intset{m}$ would make the rank of~\eqref{eq:nc-linear-polynomial-matrix-kronecker} the same as that of a linear matrix $\sum_{k=0}^m A_k x_k \in {K[x_1, \dotsc, x_m]}^{n \times n}$ (the \emph{Kronecker substitution}).
  Since giving a deterministic polynomial-time algorithm for Edmonds' problem has been open for more than half a century, computing $\deg \det$ of~\eqref{eq:nc-linear-polynomial-matrix-kronecker} is also quite challenging.
\end{remark}

\section{Application 2: Linear Differential/Difference Equations}\label{sec:application2}
In this section, we explain that dimensions of solution spaces of linear differential and difference equations can be characterized as valuations of the Dieudonné determinants.
These formulas provide applications of our algorithms to analyses of linear time-varying systems.

\subsection{\texorpdfstring{$\sigma$}{σ}-Differential Equations}

Let $R$ be a commutative ring endowed with a ring automorphism $\funcdoms{\sigma}{R}{R}$ and a left $\sigma$-derivation $\funcdoms{\delta}{R}{R}$.
A \emph{$\sigma$-differential ring} is the triple $(R, \sigma, \delta)$, or $R$ itself when $\sigma$ and $\delta$ are clear.
A \emph{$\sigma$-differential field} is a $\sigma$-differential ring which is a field.
If $\sigma = \id$, then $\sigma$-differential rings and fields are simply called \emph{differential} rings and fields.
Similarly, $\sigma$-differential rings and fields with $\delta = 0$ are called \emph{difference} rings and fields.

A \emph{constant} of a $\sigma$-differential ring $(R, \sigma, \delta)$ is an element $a \in R$ such that $\sigma(a) = a$ and $\delta(a) = 0$.
The set of all constants of $(R, \sigma, \delta)$ is denoted by $\Const_{\sigma,\delta}(R)$ or by $\Const(R)$.
It is easily checked that $\Const(R)$ is a subring of $R$, and if $R$ is a field, so is $\Const(R)$.

An additive map $\funcdoms{\theta}{R}{R}$ is said to be \emph{pseudo-linear} if it satisfies
\begin{align}\label{def:pseudo_linear}
  \theta(ab) = \sigma(a)\theta(b) + \delta(a)b
\end{align}
for all $a,b \in R$.
Recall from \cref{ex:skew_laurent_series} that $R[s; \sigma, \delta]$ denotes the skew polynomial ring over $(R, \sigma, \delta)$.
Then $\theta$ induces a left $R[s; \sigma, \delta]$-module structure on $R$, where the action $\funcdoms{\bullet}{R[s; \sigma, \delta] \times R}{R}$ is defined by
\begin{align}\label{def:pseudo_linear_action}
  \prn{\sum_{d=0}^l a_d s^d} \bullet b
  \defeq
  \sum_{d=0}^l a_d \theta^d(b)
\end{align}
for $a_0, \dotsc, a_l, b \in R$.
It can be checked that $\bullet$ satisfies the axioms of actions; for example, by~\eqref{eq:skew_polynomial_commutation_rule} and~\eqref{def:pseudo_linear}, it holds
\begin{align}
  (sa) \bullet b
  = (\sigma(a)s + \delta(a)) \bullet b
  = \sigma(a)\theta(b) + \delta(a)b
  = \theta(ab)
  = s \bullet (ab)
\end{align}
for $a,b \in R$.
Abusing notations, we represent by $\theta$ in place of $s$ the indeterminate of the skew polynomial ring that acts on $R$ by~\eqref{def:pseudo_linear_action}.
We also write $p \bullet b$ as $p(b)$ for $p \in R[\theta; \sigma, \delta]$.

An $l$th-order (scalar) \emph{linear $\sigma$-differential equation} over $R$ is an equation for $y \in R$ in the form of
\begin{align}\label{eq:scalar_linear_sigma_differential_equation}
  a_0 y + a_1 \theta(y) + \dotsb + a_{l-1}\theta^{l-1}(y) + a_l \theta^l(y) = f,
\end{align}
where $a_0, \dotsc, a_l, f \in R$.
The equation~\eqref{eq:scalar_linear_sigma_differential_equation} can be written as $p(y) = f$ by using a skew polynomial $p \defeq a_0 + a_1 \theta + \dotsb + a_l \theta^l \in R[\theta; \sigma, \delta]$.
We call $\theta$ in~\eqref{eq:scalar_linear_sigma_differential_equation} the \emph{$\sigma$-differential operator}.
If $\sigma = \id$ and $\theta = \delta$, then $\sigma$-differential equations are called \emph{linear differential equations}.
Similarly, if $\delta = 0$ and $\theta = \sigma$, then $\sigma$-differential equations are said to be \emph{linear difference equations}.
The equation~\eqref{eq:scalar_linear_sigma_differential_equation} is said to be \emph{homogeneous} when $f = 0$ and \emph{inhomogeneous} when $f \ne 0$.

Let $\theta(y)$ denotes $\prn{\theta(y_i)}_{i \in \intset{n}}$ for $y = \prn{y_i}_{i \in \intset{n}} \in R^n$.
An $l$th-order $n$-dimensional (matrix) \emph{linear $\sigma$-differential equation} over $R$ is an equation for $y \in R^n$ in form of
\begin{align}\label{eq:matrix_linear_sigma_differential_equation}
  A_0 y + A_1 \theta(y) + \dotsb + A_{l-1}\theta^{l-1}(y) + A_l \theta^l(y) = f,
\end{align}
where $A_0, \dotsc, A_l \in R^{n \times n}$ and $f \in R^n$.
Using a skew polynomial matrix $A \defeq A_0 + A_1 \theta + \dotsb + A_l \theta^l \in {R[\theta; \sigma, \delta]}^{n \times n}$, the equation~\eqref{eq:matrix_linear_sigma_differential_equation} is simply expressed as
\begin{align}\label{eq:matrix_linear_difference_equation_operator}
  A(y) = f.
\end{align}
The \emph{solution space} of~\eqref{eq:matrix_linear_difference_equation_operator} is defined as $V \defeq \set{y \in R^n}[A(y) = f]$.
It is easily checked that $V$ forms an affine module\footnote{
  Affine modules are a generalization of affine spaces obtained by replacing tangent vector spaces with modules.
  They are nothing but affine spaces if $\Const(R)$ is a field.
} over $\Const(R)$ unless $V = \varnothing$.

Suppose that $R$ is a field $K$.
Indeed, any $\sigma$-differential equation over a $\sigma$-differential field is essentially either a (usual) differential or difference equation.
This follows from the following facts.

\begin{proposition}[{\cite[Lemma~5]{Bronstein2000},~\cite[Lemma~1]{Bronstein1996}}]\label{prop:pseudo_linear_properties}
  Let $(K, \sigma, \delta)$ be a $\sigma$-differential field.
  Then the following hold:
  \begin{enumerate}
    \item An additive map $\funcdoms{\theta}{K}{K}$ is pseudo-linear if and only if it is in the form of $\gamma \sigma + \delta$ for some $\gamma \in K$.
    \item If $\sigma \neq \id$, then there exists $\alpha \in K$ such that $\delta = \alpha(\sigma - \id)$.
  \end{enumerate}
\end{proposition}

By \cref{prop:pseudo_linear_properties}, a pseudo-linear map $\theta$ can be written as $\theta = \delta + \gamma$ if $\alpha = \id$ and as $\theta = (\alpha + \gamma)\sigma + \alpha$ if $\sigma \neq \id$.
Expanding $\theta^d$ for $d = 1, \dotsc, l$ using these equations, any $\sigma$-differential equation $p(y) = 0$ with $p \in K[\theta; \sigma, \delta]$ is represented as $q(y) = 0$ for some $q \in K[\delta; \id, \delta]$ if $\sigma = \id$ and as $q'(y) = 0$ for some $q' \in K[\sigma; \sigma, 0]$ if $\sigma \neq \id$.
A typical example of this reduction is the replacement of the difference operator in a difference equation by the shift operator.
Therefore, it essentially suffices to consider only differential equations ($\theta = \delta$) over a differential field and difference equations ($\theta = \sigma$) over a difference field.
Nonetheless, we make use of the notion of $\sigma$-differential equations whenever possible since it provides a useful framework unifying differential and difference equations.

\subsection{Dimensions of Solution Spaces}
Let $(K, \sigma, \delta)$ be a differential ($\sigma = \id$) or difference ($\delta = 0$) field.
We put $\theta \defeq \delta$ in the differential case and $\theta \defeq \sigma$ in the difference case.
Consider a differential or difference equation~\eqref{eq:matrix_linear_difference_equation_operator} over $K$ and suppose that~\eqref{eq:matrix_linear_difference_equation_operator} has at least one solution.
The solution space $V$ of~\eqref{eq:matrix_linear_difference_equation_operator} forms an affine space over $C \defeq \Const(K)$ as stated above.
Now our question is how large the dimension $\dim_C V$ of $V$ over $C$ is.
This quantity is rephrased as the number of values we must designate to determine a solution of~\eqref{eq:matrix_linear_difference_equation_operator} uniquely.
An upper bound on $\dim_C V$ is given in terms of $\deg \Det$ and $\ord \Det$ of $A$ as follows.
This is partially given in~\cite[Lemma~1.10]{VanderPut2003},~\cite[Corollary~4.9]{Singer2016},~\cite[Theorem~6]{Abramov2014b}, and~\cite[Corollary~2.2]{Taelman2006}, whereas they assume $\ch(K) = 0$ which is not needed to show the following.
Here, we describe complete a proof based on their proofs.

\begin{proposition}\label{prop:upper_bound_of_dim}
  Let $(K, \sigma, \delta)$ be a differential or difference field with $C \defeq \Const(K)$.
  Let $V$ be the solution space of $A(y) = f$ with $A \in {K[\theta; \sigma, \delta]}^{n \times n}$ and $f \in K^n$ and suppose $V \ne \varnothing$.
  Then the following hold:
  \begin{enumerate}
    \item If the field extension $K \extends C$ is infinite, then $\dim_C V$ is finite if and only if $A$ is nonsingular.\label{item:upper_bound_of_dim_1}
    \item If $A$ is nonsingular, it holds $\dim_C V \le \deg \Det A$ in the differential case and $\dim_C V \le \deg \Det A - \ord \Det A$ in the difference case.\label{item:upper_bound_of_dim_2}
  \end{enumerate}
\end{proposition}

\begin{proof}
  For any $v \in V$, the $C$-vector space $V - v \defeq \set{y - v}[y \in V]$ is the solution space of $A(y) = 0$.
  Hence it suffices to consider only homogeneous equations.
  Our proof consists of three steps: we show the claims for first-order homogeneous equations in Step~1, for scalar homogeneous equations in Step~2, and for general homogeneous equations in Step~3.

  (Step~1)
  Consider the case when $A = A_0 + I_n\theta$ and $f = 0$, i.e., the corresponding linear $\sigma$-differential equation is
  \begin{align}\label{eq:upper_bound_of_dim_step_1_eq}
    \theta(y) = -A_0 y.
  \end{align}
  We further require $A_0$ to be nonsingular only in the difference case.
  Since $A$ is nonsingular, it suffices to show only~\ref{item:upper_bound_of_dim_2}.
  Then $A\theta^{-1} = A_0\theta^{-1} + I_n$ is proper as a matrix over $K(\theta; \sigma, \delta)$ with valuation $-\deg$.
  Since $I_n$ is nonsingular, it holds $\deg \Det A \theta^{-1} = 0$ by \cref{prop:biproper_equivalence} and thus $\deg \Det A = n$.
  Similarly, in the difference case, it holds $\ord \Det A = 0$ by the nonsingularity of $A_0$.
  Therefore, our goal is to show $\dim_C V \le n$ in both cases.
  Since $\dim_K V \le n$ is clear, it suffices to prove $\dim_K V = \dim_C V$.

  Let $v_1, \dotsc, v_m \in V$ be solutions of~\eqref{eq:upper_bound_of_dim_step_1_eq} that are linearly dependent over $K$.
  We show that they are also dependent over $C$, which implies $\dim_K V = \dim_C V$.
  Without loss of generality, we assume that $v_2, \dotsc, v_m$ are linearly independent over $K$.
  Then there uniquely exists $c_2, \dotsc, c_m \in K$ such that $v_1 = \sum_{i=2}^m c_i v_i$.
  Then it holds
  \begin{align}
    0
    &= \theta\prn{v_1 - \sum_{i=2}^m c_i v_i}
     = \theta\prn{v_1} - \sum_{i=2}^m \theta\prn{c_i v_i} \\
    &= -A_0v_1 - \sum_{i=2}^m \prn{\sigma\prn{c_i}\theta\prn{v_i} + \delta\prn{c_i}v_i} \\
    &= -A_0\sum_{i=2}^m c_i v_i - \sum_{i=2}^m \prn{-\sigma\prn{c_i} A_0 v_i + \delta\prn{c_i}v_i}\\
    &= A_0\sum_{i=2}^m \prn{\sigma\prn{c_i} - c_i} v_i - \sum_{i=2}^m \delta(c_i) v_i.
  \end{align}
  In the differential case, we have $0 = -\sum_{i=2}^m \delta\prn{c_i} v_i$ by $\sigma = \id$.
  From the independence of $v_2, \dotsc, v_m$, it must holds $\delta\prn{c_i} = 0$, which means $c_i \in C$ for $i = 2,\dotsc, m$.
  In the difference case, we have $0 = \sum_{i=2}^m \prn{\sigma\prn{c_i} - c_i} v_i$ from $\delta = 0$ and the assumption that $A_0$ is nonsingular.
  Hence we obtain $\sigma\prn{c_i} = c_i$ and thus $c_i \in C$ for $i = 2,\dotsc, m$.
  Thus $v_1, \dotsc, v_m$ are also linearly dependent over $C$ in both cases.

  (Step~2)
  Consider a scalar homogeneous linear differential or difference equation $p(y) = 0$ with $p = \sum_{d=0}^l a_d \theta^d \in K[\theta; \sigma, \delta]$.
  When $p = 0$, the solution space $V$ coincides with $K$.
  Thus $\dim_C V = \dim_C K$ is infinite when $K \extends C$ is infinite.
  Suppose that $p \ne 0$ and $\deg p = l$, i.e., $a_l \ne 0$.
  In the difference case, as $\theta = \sigma$ is bijective, $p(y) = 0$ and $p'(y) = 0$ with $p' \defeq \theta^{-\ord p}p$ have the same solution spaces.
  Moreover, by $\deg p' = \deg p - \ord p$ and $\ord p' = 0$, it holds $\deg p' - \ord p' = \deg p - \ord p$.
  Therefore, in the difference case, we can assume $\ord p = 0$ (i.e., $a_0 \ne 0$) without loss of generality.

  We construct the following $l$-dimensional matrix linear differential or difference equation:
  \begin{align}\label{eq:companion_matrix_equation}
    \theta \begin{pmatrix}
      y_0 \\
      y_1 \\
      \vdots \\
      y_{l-2} \\
      y_{l-1}
    \end{pmatrix}
    =
    \begin{pNiceMatrix}
                     0 &                1 &      0 &               \Cdots &      0 \\
                \Vdots &           \Ddots & \Ddots &               \Ddots & \Vdots \\
                \Vdots &                  & \Ddots &               \Ddots &      0 \\
                     0 &           \Cdots & \Cdots &                    0 &      1 \\
      -\frac{a_0}{a_l} & -\frac{a_1}{a_l} & \Cdots & -\frac{a_{l-2}}{a_l} & -\frac{a_{l-1}}{a_l}
    \end{pNiceMatrix}
    \begin{pNiceMatrix}
      y_0 \\
      y_1 \\
      \Vdots \\
      y_{l-2} \\
      y_{l-1}
    \end{pNiceMatrix}.
  \end{align}
  If $y \in K$ is a solution of $p(y) = 0$, then $\trsp{\prn[\big]{y, \theta(y), \dotsc, \theta^{l-1}(y)}} \in K^n$ is a solution of~\eqref{eq:companion_matrix_equation}.
  Conversely, any solution of~\eqref{eq:companion_matrix_equation} is obtained in this way.
  Therefore, the solution space $W$ of~\eqref{eq:companion_matrix_equation} is isomorphic to $V$ as $C$-vector spaces.
  In the differential case, $\dim_C W = l = \deg p$ by the above proof of Step~1.
  In the difference case, the matrix in the right-hand side of~\eqref{eq:companion_matrix_equation} is nonsingular by $a_0 \ne 0$.
  Hence $\dim_C W = l = \deg p - \ord p$ again from Step~1.

  (Step~3)
  Consider a matrix homogeneous differential or difference equation $A(y) = 0$ with $A \in {K[\theta; \sigma, \delta]}^{n \times n}$.
  Let $D = UAW = \diag(d_1, \dotsc, d_n)$ be the Jacobson normal form of $A$ over $K[\theta; \sigma, \delta]$.
  Putting $z = \prn{z_1, \dotsc, z_n} \defeq W(y)$, the solution space sof $A(y) = 0$ and $D(z) = 0$ are isomorphic as $C$-vector spaces.
  Since $D$ is diagonal, the solution space of $D(z) = 0$ is the direct sum of the solution space $V_i$ of $d_i(z_i) = 0$ for $i \in \intset{n}$.
  Namely, it holds
  \begin{align}\label{eq:dim_C_V_decompose}
    \dim_C V = \sum_{i=1}^n \dim_C V_i.
  \end{align}

  If $A$ (and thus $D$) is singular, there exists $i \in \intset{n}$ such that $d_i = 0$.
  Thus $\dim_C V$ is infinite when $K \extends C$ is infinite by the above Step~2 and~\eqref{eq:dim_C_V_decompose}.
  Suppose that $A$ is nonsingular.
  Since $U$ and $W$ are invertible over $K[\theta; \sigma, \delta]$, they are biproper over $K(\theta; \sigma, \delta)$ with valuation $\deg$ and over $K(\theta; \sigma, 0)$ with valuation $\ord$ in the difference case.
  Thus $\deg \Det$ of $U$ and $W$ are $0$, which means $\deg \Det A = \deg \Det D = \sum_{i=1}^n \deg d_i$.
  Therefore, by Step~2 and~\eqref{eq:dim_C_V_decompose}, we have $\dim_C V \le \deg \Det A$ in the differential case, as desired.
  The completely analog holds in the difference case by replacing $\deg \Det$ with $\deg \Det - \ord \Det$.
\end{proof}

The upper bound on $\dim_C V$ given in \cref{prop:upper_bound_of_dim} may not be attained on some equations.
For example, consider a first-order linear differential equation $y' + y = 0$ over $\setC(t)$ with the usual differentiation $'$.
The solution of this equation over $\setC(t)$ is only $y = 0$ and thus the dimension of the solution space is $0$.
However, if the differential field $\setC(t)$ is extended to $\setC(t, \e^t)$, the solution space becomes $V \defeq \set{c\e^{-t}}[c \in \setC]$, which has dimension $1$ over $\setC$.
This is analogous to the situation of extending a field to its algebraic closure in order for $n$th-order algebraic equations to have $n$ solutions.
We explain such an extension briefly.

Let $(K, \sigma, \delta)$ be a differential or difference field.
A differential or difference ring $(R, \bar{\sigma}, \bar{\delta})$ is called a \emph{differential} or \emph{difference extension} of $K$ if $K$ is a subring of $R$ and $\bar{\sigma}$ and $\bar{\delta}$ coincides with $\sigma$ and $\delta$ on $K$, respectively.
A differential or difference equation $A(y) = f$ over $K$ is naturally extended to that over $R$.
Following~\cite{Abramov2014b}, we call an extension $R$ of $K$ \emph{adequate} if it satisfies the following:
\begin{enumerate}[label={\upshape{(AE\arabic*)}}]
  \item $C \defeq \Const(R)$ is a field.\label{item:AE1}
  \item Any scalar homogeneous differential or difference equation $p(y) = 0$ with $p \in K[\theta; \sigma, \delta] \setminus \set{0}$ has the solution space $V$ over $R$ such that $\dim_C V = \deg p$ in the differential case and $\dim_C V = \deg p - \ord p$ in the difference case.\label{item:AE2}
\end{enumerate}

Let $K$ be a differential field.
If $\Const(K)$ is algebraically closed, then there exists an adequate extension $R$ of $K$ such that $\Const(R) = \Const(K)$, called the \emph{universal} (\emph{differential}) \emph{Picard--Vessiot ring} of $K$~\cite[Section~3.2]{VanderPut2003}.
In addition, any differential field $K$ of characteristic $0$ has a difference extension whose constant field is the algebraic closure of $\Const(K)$~\cite{Abramov2014b}; see also~\cite[Exercise~1.5, 2:(c),(d), 3:(c)]{VanderPut2003}.
Therefore, there always exists an adequate extension of any differential field of characteristic $0$.

Next, suppose that $K$ is a difference field.
If $\Const(K)$ is algebraically closed, there exists an adequate extension $R$ of $K$ such that $\Const(R) = \Const(K)$, called the \emph{universal} (\emph{difference}) \emph{Picard--Vessiot ring} of $K$~\cite[Section~1.4]{VanderPut1997}.
Indeed, for any difference field $K$ of characteristic $0$, an adequate difference extension $R$ can be easily constructed~\cite[Proposition~4]{Abramov2014b}, while $\Const(R) = \Const(K)$ is no longer guaranteed.

We then turn to matrix, inhomogeneous equations.
As we will see below,~\ref{item:AE2} is indeed equivalent to the following:

\begin{enumerate}[label={\upshape{(AE2')}}]
  \item Any matrix differential or difference equation $A(y) = f$ with $A \in \GL_n\prn{K[\theta; \sigma, \delta]}$ and $f \in K^n$ has the solution space $V$ over $R$ such that $\dim_C V = \deg \Det A$ in the differential case and $\dim_C V = \deg \Det A - \ord \Det A$ in the difference case.\label{item:AE2_modified}
\end{enumerate}

\begin{lemma}\label{lem:AE2_equivalence}
  \ref{item:AE2} and~\ref{item:AE2_modified} are equivalent.
\end{lemma}

\begin{proof}
  It is clear that~\ref{item:AE2_modified} implies~\ref{item:AE2}; we show the converse holds.
  Let $(K, \sigma, \delta)$ be a differential or difference field and $R$ its extension satisfying~\ref{item:AE1} and~\ref{item:AE2}.
  As stated in the proof of \cref{prop:upper_bound_of_dim}, a matrix differential and difference equation is essentially reduced to $n$ scalar equations by considering the Jacobson normal form.
  This means that it suffices to consider only a scalar inhomogeneous equation $p(y) = f$ with $p \in K[\theta; \sigma, \delta] \setminus \set{0}$ and $f \in K \setminus \set{0}$.
  In addition, the solution space of $p(y) = f$ over $R$ is the translation of the solution space of $p(y) = 0$ over $R$ by any solution of $p(y) = f$.
  Therefore, our goal is to show that $p(y) = f$ has at least one solution over $R$.

  We first deal with the differential case.
  Let $q \defeq \theta f^{-1} p$.
  Then any solution $y \in R$ of $q(y) = 0$ is also a solution of $p(y) = cf$ for some $c \in C \defeq \Const(R)$ (see~\cite[Exercise~1.14,~1]{VanderPut2003}).
  By~\ref{item:AE2}, the dimension of the solution space $W$ of $q(y) = 0$ is $\deg q = \deg p + 1$, whereas that of $p(y) = 0$ is $\deg p < \deg q$.
  Therefore, there exists $v \in W$ that is not a solution of $p(v) = 0$, i.e., $p(v) = cf$ for some nonzero $c \in \mult{C}$.
  Then $c^{-1}v$ is a solution of $p(y) = f$, as required.
  The difference case can be in the same way by considering $q \defeq \prn{\theta - 1}\prn{f^{-1}p} = \theta f^{-1} p - f^{-1} p$.
\end{proof}

\Cref{prop:upper_bound_of_dim} and \cref{lem:AE2_equivalence} lead us to the following consequence.

\begin{theorem}\label{thm:dof_over_adequate_extension}
  Let $(K, \sigma, \delta)$ be a differential or difference field, $R$ its adequate extension, and $C \defeq \Const(R)$.
  Let $V$ be the solution space of $A(y) = f$ over $R$ with $A \in \GL_n\prn{K[\theta; \sigma, \delta]}$ and $f \in K^n$.
  Then it holds $\dim_C V = \deg \Det A$ in the differential case and $\dim_C V = \deg \Det A - \ord \Det A$ in the difference case.
\end{theorem}

Since $\deg$ and $\ord$ are discrete valuations, we can apply our algorithms to compute the dimension of solution spaces of linear differential or difference equations over an adequate extension.

\addcontentsline{toc}{section}{Acknowledgments}
\section*{Acknowledgments}
The author thanks Hiroshi Hirai for teaching me about valuation theory.
This work was supported by JST ACT-I Grant Number JPMJPR18U9, Japan, and Grant-in-Aid for JSPS Research Fellow Grant Number JP18J22141, Japan.

\addcontentsline{toc}{section}{References}
\bibliographystyle{abbrv}
\bibliography{references}

\begin{thebibliography}{10}

\bibitem{Abramov2014b}
S.~A. Abramov and M.~A. Barkatou.
\newblock {On solution spaces of products of linear differential or difference
  operators}.
\newblock {\em ACM Communications in Computer Algebra}, 48(4):155--165, 2014.

\bibitem{Amitsur1966}
S.~A. Amitsur.
\newblock {Rational identities and applications to algebra and geometry}.
\newblock {\em Journal of Algebra}, 3(3):304--359, 1966.

\bibitem{Beckermann2006}
B.~Beckermann, H.~Cheng, and G.~Labahn.
\newblock {Fraction-free row reduction of matrices of Ore polynomials}.
\newblock {\em Journal of Symbolic Computation}, 41(5):513--543, 2006.

\bibitem{Bronstein2000}
M.~Bronstein.
\newblock {On solutions of linear ordinary differential equations in their
  coefficient field}.
\newblock {\em Journal of Symbolic Computation}, 29(6):841--877, 2000.

\bibitem{Bronstein1996}
M.~Bronstein and M.~Petkov{\v{s}}ek.
\newblock {An introduction to pseudo-linear algebra}.
\newblock {\em Theoretical Computer Science}, 157(1):3--33, 1996.

\bibitem{Brungs1973}
H.~H. Brungs.
\newblock {Left Euclidean rings}.
\newblock {\em Pacific Journal of Mathematics}, 45(1):27--33, 1973.

\bibitem{Brungs1984}
H.~H. Brungs and G.~T{\"{o}}rner.
\newblock {Skew power series rings and derivations}.
\newblock {\em Journal of Algebra}, 87(2):368--379, 1984.

\bibitem{Chrystal1897}
G.~Chrystal.
\newblock {A fundamental theorem regarding the equivalence of systems of
  ordinary linear differential equations, and its application to the
  determination of the order and the systematic solution of a determinate
  system of such equations}.
\newblock {\em Transactions of the Royal Society of Edinburgh}, 38(1):163--178,
  1897.

\bibitem{Cohen1946}
I.~S. Cohen.
\newblock {On the structure and ideal theory of complete local rings}.
\newblock {\em Transactions of the American Mathematical Society}, 59(1):54,
  1946.

\bibitem{Cohn1977}
P.~M. Cohn.
\newblock {\em {Skew Field Constructions}}.
\newblock London Mathematical Society Lecture Note Series. Cambridge University
  Press, Cambridge, 1977.

\bibitem{Cohn1985}
P.~M. Cohn.
\newblock {\em {Free Rings and Their Relations}}, volume~19 of {\em London
  Mathematical Society Monograph}.
\newblock Academic Press, London, 2nd edition, 1985.

\bibitem{Cohn1995}
P.~M. Cohn.
\newblock {\em {Skew Fields: Theory of General Division Rings}}, volume~57 of
  {\em Encyclopedia of Mathematics and Its Applications}.
\newblock Cambridge University Press, Cambridge, 1995.

\bibitem{Cohn2003}
P.~M. Cohn.
\newblock {\em {Further Algebra and Applications}}.
\newblock Springer, London, 2003.

\bibitem{Dieudonne1943}
J.~Dieudonn{\'{e}}.
\newblock {Les d{\'{e}}terminants sur un corps non commutatif}.
\newblock {\em Bulletin de la Soci{\'{e}}t{\'{e}} Math{\'{e}}matique de
  France}, 71:27--45, 1943.

\bibitem{Dress1990}
A.~W.~M. Dress and W.~Wenzel.
\newblock {Valuated matroids: a new look at the greedy algorithm}.
\newblock {\em Applied Mathematics Letters}, 3(2):33--35, 1990.

\bibitem{Dress1992}
A.~W.~M. Dress and W.~Wenzel.
\newblock {Valuated matroids}.
\newblock {\em Advances in Mathematics}, 93(2):214--250, 1992.

\bibitem{Dumas1992}
F.~Dumas.
\newblock {Skew power series rings with general commutation formula}.
\newblock {\em Theoretical Computer Science}, 98(1):99--114, 1992.

\bibitem{Edmonds1967}
J.~Edmonds.
\newblock {Systems of distinct representatives and linear algebra}.
\newblock {\em Journal of Research of the National Bureau of Standards},
  71B(4):241--245, 1967.

\bibitem{Elliger1967}
S.~Elliger.
\newblock {Potenzbasiserweiterungen}.
\newblock {\em Journal of Algebra}, 7(2):254--262, 1967.

\bibitem{Fortin2004}
M.~Fortin and C.~Reutenauer.
\newblock {Commutative/noncommutative rank of linear matrices and subspaces of
  matrices of low rank}.
\newblock {\em S{\'{e}}minaire Lotharingien de Combinatoire}, 52, 2004.

\bibitem{Garg2016}
A.~Garg, L.~Gurvits, R.~Oliveira, and A.~Wigderson.
\newblock {A deterministic polynomial time algorithm for non-commutative
  rational identity testing}.
\newblock In {\em Proceedings of the 57th Annual IEEE Symposium on Foundations
  of Computer Science (FOCS '16)}, pages 109--117, 2016.

\bibitem{Giesbrecht2013}
M.~Giesbrecht and M.~S. Kim.
\newblock {Computing the Hermite form of a matrix of Ore polynomials}.
\newblock {\em Journal of Algebra}, 376:341--362, 2013.

\bibitem{Goodearl2004}
K.~R. Goodearl and R.~B. {Warfield, Jr.}
\newblock {\em {An Introduction to Noncommutative Noetherian Rings}}.
\newblock Cambridge University Press, Cambridge, second edition, 2004.

\bibitem{Gurvits2004}
L.~Gurvits.
\newblock {Classical complexity and quantum entanglement}.
\newblock {\em Journal of Computer and System Sciences}, 69(3):448--484, 2004.

\bibitem{Hamada2020}
M.~Hamada and H.~Hirai.
\newblock {Computing the nc-rank via discrete convex optimization on CAT(0)
  spaces}, 2020.

\bibitem{Hezavehi1982}
M.~M. Hezavehi.
\newblock {Matrix valuations and their associated skew fields}.
\newblock {\em Results in Mathematics}, 5(1-2):149--156, 1982.

\bibitem{Hirai2019}
H.~Hirai.
\newblock {Computing the degree of determinants via discrete convex
  optimization on Euclidean buildings}.
\newblock {\em SIAM Journal on Applied Geometry and Algebra}, 3(3):523--557,
  2019.

\bibitem{Hirai2020b}
H.~Hirai and M.~Ikeda.
\newblock {A cost-scaling algorithm for computing the degree of determinants},
  2020.

\bibitem{Hopcroft1973}
J.~E. Hopcroft and R.~M. Karp.
\newblock {An $n^{5/2}$ algorithm for maximum matchings in bipartite graphs}.
\newblock {\em SIAM Journal on Computing}, 2:225--231, 1973.

\bibitem{Ivanyos2018}
G.~Ivanyos, Y.~Qiao, and K.~V. Subrahmanyam.
\newblock {Constructive non-commutative rank computation is in deterministic
  polynomial time}.
\newblock {\em Computational Complexity}, 27(4):561--593, 2018.

\bibitem{Iwata2007}
S.~Iwata and R.~Shimizu.
\newblock {Combinatorial analysis of generic matrix pencils}.
\newblock {\em SIAM Journal on Matrix Analysis and Applications},
  29(1):245--259, 2007.

\bibitem{Jacobson1943}
N.~Jacobson.
\newblock {\em {The Theory of Rings}}, volume~2 of {\em Mathematical Surveys
  and Monographs}.
\newblock AMS, Providence, RI, 1943.

\bibitem{Kabanets2004}
V.~Kabanets and R.~Impagliazzo.
\newblock {Derandomizing polynomial identity tests means proving circuit lower
  bounds}.
\newblock {\em Computational Complexity}, 13(1--2):1--46, 2004.

\bibitem{Khochtali2017}
M.~Khochtali, J.~{Rosenkilde n{\'{e}} Nielsen}, and A.~Storjohann.
\newblock {Popov form computation for matrices of Ore polynomials}.
\newblock In {\em Proceedings of the 42nd International Symposium on Symbolic
  and Algebraic Computation (ISSAC '17)}, pages 253--260, New York, NY, 2017.
  ACM Press.

\bibitem{Konig1931}
D.~K{\"{o}}nig.
\newblock {Gr{\'{a}}fok {\'{e}}s m{\'{a}}trixok}.
\newblock {\em Matematikai {\'{e}}s Fizikai Lapok}, 38:116--119, 1931.

\bibitem{Krylov2008}
P.~A. Krylov and A.~A. Tuganbaev.
\newblock {\em {Modules over Discrete Valuation Domains}}, volume 145 of {\em
  de Gruyter Expositions in Mathematics}.
\newblock Walter de Gruyter, Berlin, 2008.

\bibitem{Kuhn1955}
H.~W. Kuhn.
\newblock {The Hungarian method for the assignment problem}.
\newblock {\em Naval Research Logistics Quarterly}, 2:83--97, 1955.

\bibitem{Lam1999}
T.~Y. Lam.
\newblock {\em {Lectures on Modules and Rings}}, volume 189 of {\em Graduate
  Texts in Mathematics}.
\newblock Springer, New York, NY, 1999.

\bibitem{Levandovskyy2011}
V.~Levandovskyy and K.~Schindelar.
\newblock {Computing diagonal form and Jacobson normal form of a matrix using
  Gr{\"{o}}bner bases}.
\newblock {\em Journal of Symbolic Computation}, 46(5):595--608, 2011.

\bibitem{Lovasz1989}
L.~Lov{\'{a}}sz.
\newblock {Singular spaces of matrices and their application in combinatorics}.
\newblock {\em Boletim da Sociedade Brasileira de Matem\'{a}tica},
  20(1):87--99, 1989.

\bibitem{Moriyama2013}
S.~Moriyama and K.~Murota.
\newblock {Discrete Legendre duality in polynomial matrices (in Japanese)}.
\newblock {\em The Japan Society for Industrial and Applied Mathematics},
  23(2):183--202, 2013.

\bibitem{Murota1995a}
K.~Murota.
\newblock {Computing the degree of determinants via combinatorial relaxation}.
\newblock {\em SIAM Journal on Computing}, 24(4):765--796, 1995.

\bibitem{Murota1995c}
K.~Murota.
\newblock {Finding optimal minors of valuated bimatroids}.
\newblock {\em Applied Mathematics Letters}, 8(4):37--41, 1995.

\bibitem{Murota2003}
K.~Murota.
\newblock {\em {Discrete Convex Analysis}}.
\newblock SIAM, Philadelphia, 2003.

\bibitem{Murota2000}
K.~Murota.
\newblock {\em {Matrices and Matroids for Systems Analysis}}, volume~20 of {\em
  Algorithms and Combinatorics}.
\newblock Springer, Berlin, 2010.

\bibitem{Neumann1949}
B.~H. Neumann.
\newblock {On ordered division rings}.
\newblock {\em Transactions of the American Mathematical Society}, 66(1):202,
  1949.

\bibitem{Ore1933}
O.~Ore.
\newblock {Theory of non-commutative polynomials}.
\newblock {\em Annals of Mathematics}, 34(3):480--508, 1933.

\bibitem{Ore1955}
O.~Ore.
\newblock {Graphs and matching theorems}.
\newblock {\em Duke Mathematical Journal}, 22(4):625--639, 1955.

\bibitem{Paykan2017}
K.~Paykan and A.~Moussavi.
\newblock {Study of skew inverse Laurent series rings}.
\newblock {\em Journal of Algebra and Its Applications}, 16(12):1750221, 2017.

\bibitem{Roux1986}
B.~Roux.
\newblock {Anneaux non commutatifs de valuation discr{\`{e}}te ou finie}.
\newblock {\em Comptes Rendus de l'Acad{\'{e}}mie des Sciences, S{\'{e}}rie I},
  302(9):259--262 and 291--293, 1986.

\bibitem{Schrijver2003}
A.~Schrijver.
\newblock {\em Combinatorial Optimization}, volume~24 of {\em Algorithms and
  Combinatorics}.
\newblock Springer, Berlin, 2003.

\bibitem{Schwartz1980}
J.~T. Schwartz.
\newblock {Fast probabilistic algorithms for verification of polynomial
  identities}.
\newblock {\em Journal of the ACM}, 27(4):701--717, 1980.

\bibitem{Singer2016}
M.~F. Singer.
\newblock {Algebraic and algorithmic aspects of linear difference equations}.
\newblock In {\em Galois Theories of Linear Difference Equations: An
  Introduction}, volume 211 of {\em Mathematical Surveys and Monograph}, pages
  1--41. AMS, Providence, RI, 2016.

\bibitem{Smits1968}
T.~H.~M. Smits.
\newblock {Skew polynomial rings}.
\newblock {\em Indagationes Mathematicae}, 30(1):209--224, 1968.

\bibitem{Taelman2006}
L.~Taelman.
\newblock {Dieudonn{\'{e}} determinants for skew polynomial rings}.
\newblock {\em Journal of Algebra and Its Applications}, 5(1):89--93, 2006.

\bibitem{Valiant1979}
L.~G. Valiant.
\newblock {Completeness classes in algebra}.
\newblock In {\em Proceedings of the 11th Annual ACM Symposium on Theory of
  Computing (STOC '79)}, pages 249--261, New York, NY, 1979. ACM Press.

\bibitem{VanderPut1997}
M.~van~der Put and M.~F. Singer.
\newblock {\em {Galois Theory of Difference Equations}}, volume 1666 of {\em
  Lecture Notes in Mathematics}.
\newblock Springer-Verlag, Berlin, 1997.

\bibitem{VanderPut2003}
M.~van~der Put and M.~F. Singer.
\newblock {\em {Galois Theory of Linear Differential Equations}}, volume 328 of
  {\em Grundlehren der mathematischen Wissenschaften}.
\newblock Springer-Verlag, Berlin, 2003.

\bibitem{Vandooren1979}
P.~M. {Van Dooren}, P.~Dewilde, and J.~Vandewalle.
\newblock {On the determination of the Smith-Macmillan form of a rational
  matrix from its Laurent expansion}.
\newblock {\em IEEE Transactions on Circuits and Systems}, 26(3):180--189,
  1979.

\bibitem{Verghese1981}
G.~C. Verghese and T.~Kailath.
\newblock {Rational matrix structure}.
\newblock {\em IEEE Transactions on Automatic Control}, 26(2):434--439, 1981.

\bibitem{Vidal1981}
R.~Vidal.
\newblock {Anneaux de valuation discr{\`{e}}te complets non commutatifs}.
\newblock {\em Transactions of the American Mathematical Society},
  267(1):65--81, 1981.

\bibitem{Warner1993}
S.~Warner.
\newblock {\em {Topological Rings}}, volume 178 of {\em North-Holand
  Mathematics Studies}.
\newblock Elsevier, North Holland, 1993.

\end{thebibliography}

\end{document}